\documentclass[11pt]{article}
\usepackage[utf8]{inputenc}
\usepackage[margin=0.9in]{geometry}
\usepackage{authblk}
\usepackage{amsmath,amssymb,amsthm,mathtools}
\usepackage{aliascnt}
\usepackage{microtype}
\usepackage{enumitem}
\usepackage{booktabs}
\usepackage{array}
\usepackage{tikz}
\usetikzlibrary{arrows.meta,calc,positioning,decorations.pathreplacing}
\usepackage{cite}
\usepackage[hidelinks]{hyperref}
\usepackage[nameinlink,capitalize,noabbrev]{cleveref}
\hypersetup{
  pdftitle={The Kikuchi Hierarchy is Sharp for kXOR},
  pdfauthor={Alexander Schmidhuber and Matthew B. Hastings}
}
\newtheorem{theorem}{Theorem}[section]
\newaliascnt{proposition}{theorem}
\newtheorem{proposition}[proposition]{Proposition}
\aliascntresetthe{proposition}
\newaliascnt{lemma}{theorem}
\newtheorem{lemma}[lemma]{Lemma}
\aliascntresetthe{lemma}
\newaliascnt{corollary}{theorem}
\newtheorem{corollary}[corollary]{Corollary}
\aliascntresetthe{corollary}
\newaliascnt{definition}{theorem}
\theoremstyle{definition}
\newtheorem{definition}[definition]{Definition}
\aliascntresetthe{definition}
\theoremstyle{remark}
\newaliascnt{remark}{theorem}
\newtheorem{remark}[remark]{Remark}
\aliascntresetthe{remark}
\crefformat{equation}{(#2#1#3)}
\crefrangeformat{equation}{(#3#1#4)--(#5#2#6)}
\crefmultiformat{equation}{(#2#1#3)}{ and (#2#1#3)}{, (#2#1#3)}{, and (#2#1#3)}
\newcommand{\E}{\mathbb E}
\newcommand{\bbP}{\mathbb P}

\newcommand{\1}{\mathbf 1}
\newcommand{\cB}{\mathcal B}
\newcommand{\cC}{\mathcal C}
\newcommand{\cG}{\mathcal G}
\newcommand{\cH}{\mathcal H}
\newcommand{\cL}{\mathcal L}
\newcommand{\cO}{\mathcal O}
\newcommand{\cP}{\mathcal P}
\newcommand{\cQ}{\mathcal Q}
\newcommand{\cR}{\mathcal R}
\newcommand{\cS}{\mathcal S}
\newcommand{\cT}{\mathcal T}
\newcommand{\cW}{\mathcal W}
\newcommand{\sd}{\mathbin{\triangle}}
\newcommand{\bigsd}{\mathop{\bigtriangleup}}
\newcommand{\bd}{\bar d}
\newcommand{\eps}{\varepsilon}
\newcommand{\Dkl}{\mathsf D_{k,\ell}}
\newcommand{\xstar}{{x^\star}}
\newcommand{\rhoone}{\rho_1}
\newcommand{\ip}[2]{\left\langle #1,#2\right\rangle}
\DeclareMathOperator{\Tr}{Tr}
\DeclareMathOperator{\diag}{diag}
\DeclareMathOperator{\sign}{sign}
\DeclareMathOperator{\Var}{Var}
\DeclareMathOperator{\Cert}{Cert}
\DeclareMathOperator{\Def}{Def}
\DeclareMathOperator{\val}{val}
\DeclareMathOperator{\Ran}{Ran}
\DeclareMathOperator{\rank}{rank}

\DeclareMathOperator{\err}{err}
\DeclareMathOperator{\Adv}{Adv}
\numberwithin{equation}{section}
\allowdisplaybreaks
\setlist[enumerate]{leftmargin=2.1em,itemsep=0.2em,topsep=0.3em}
\title{The Kikuchi Hierarchy is Sharp for $k$XOR}
\author[1,2]{Alexander Schmidhuber}
\author[2]{Matthew B. Hastings}

\affil[1]{Center for Theoretical Physics, MIT}
\affil[2]{Microsoft Station Q}
\date{}

\begin{document}
\maketitle

\begin{abstract}
Planted noisy $k$XOR and the strong refutation of random $k$XOR are
 governed by a conjectured trade-off between signal strength and time: Level $\ell$ of the Kikuchi hierarchy should achieve the smooth curve
\begin{equation*}
 m\ \gtrsim\ \rho^{-2}n^{k/2}/\ell^{k/2-1}\ \text{clauses}
 \quad\Longleftrightarrow\quad
 \text{solvable in time }n^{O(\ell)},
\end{equation*}
where $\rho$ is the bias of the planted signal or, for refutation, the
target advantage.  However, every spectral analysis of sparse $k$XOR to date
loses polylogarithmic factors against this curve, a loss that enters the exponent of the running time.

We show that a normalized variant of the Kikuchi hierarchy
achieves the sharp conjectured trade-off, with no
logarithmic loss, at every arity $k\ge3$.  At the scale above, our
algorithms achieve strong detection, weak recovery, and
strong refutation; an additional cleanup step boosts
weak recovery to exact recovery, and the refutation certificates yield sum-of-squares proofs of degree $O(\ell)$.  We also prove matching lower bounds in the same model.  The inference and refutation upper bounds transfer to more general planting laws and predicates.
Finally, we give a quantum algorithm that achieves a quartic speedup
over the classical spectral algorithms for detection and weak recovery.

The proofs rest on two key ingredients: a  normalization of the sparse Kikuchi matrix, and a
sharp count of the closed walks in its trace expansion. We use a closely related trace-walk count to prove Feige's 2008 hypergraph Moore bound conjecture in a companion paper~\cite{SHmoore}.
\end{abstract}

\section{Introduction}\label{sec:intro}

Planted inference problems concern the detection or recovery of a low-rank signal hidden inside a large random
structure, such as a spike in a Gaussian tensor or a planted assignment in a random
constraint satisfaction problem.  For many such problems, the signal-to-noise ratio (SNR) at which the
task becomes statistically possible and the ratio  at which known
efficient algorithms begin to succeed differ by polynomial factors, and
in between the problem appears computationally hard.  These
statistical--computational gaps come with sharp and remarkably
consistent predictions, many originating in statistical
physics~\cite{ZK,BPW}.  Within the computationally hard regime,
a more refined phenomenon may appear: running time can be traded smoothly
against signal strength, producing a hierarchy of superpolynomial algorithms whose runtime increases as the SNR decreases towards the statistical threshold.  The precise shape of this trade-off is the central object of study~\cite{WEM,KWB}.

The Kikuchi hierarchy,\footnote{The name refers to the Kikuchi
approximation of free energies in statistical
mechanics~\cite{Kikuchi1951}; see also the mean-field interpretation
in~\cite{Hastings2019} and Appendix B of~\cite{SOKB}.} defined by
Wein, El~Alaoui, and Moore~\cite{WEM} (and in a different
language by
Hastings~\cite{Hastings2019}), is a canonical spectral realization of
this trade-off.  At level $\ell$, the observations are assembled into a
matrix indexed by the $\ell$-element subsets of the variables, built so
that its leading eigenvalues and eigenvectors carry information about
the hidden signal.  For $k>2$, raising the level of the hierarchy lowers the amount of
signal the algorithm needs, at the price of $n^{O(\ell)}$ running time.
The hierarchy has since become the standard route to smooth time--data
trade-offs in planted inference and random-CSP
refutation~\cite{GKM,HKM,HKMMS,SOKB,SZ}; closely related Kikuchi-graph
and even-cover methods also power recent lower bounds for locally
decodable and correctable codes~\cite{AGKM,KothariManohar}.

This paper concerns the central sparse example, \emph{planted noisy $k$XOR}, defined
$m$ clauses each specifying the parity of $k$ out of $n$ Boolean variables.  For this problem, the signal strength corresponds to the number of sampled clauses. The
conjectured trade-off between the hierarchy level $\ell$ and the number of clauses is
\begin{equation}\label{eq:scale}
 m\ \asymp\ \frac{\Dkl}{\rho^2},
 \qquad
 \Dkl:=\frac{n^{k/2}}{\ell^{k/2-1}},
\end{equation}
where $\rho$ is the bias of each noisy parity check; for strong
refutation, $\rho$ is replaced by the
target advantage $\eps$.  For inference, the same scale is predicted by the
low-degree likelihood-ratio calculation~\cite{HopkinsSteurer,KWB}; and
at $\rho=1$ it is the density at which Feige conjectured that even
covers of length $O(\ell\log n)$ must appear in any $k$-uniform
hypergraph, $k\ge3$~\cite{Feige}.

This curve, however, has so far been realized by the spectral Kikuchi
hierarchy only up to logarithmic
losses, that is, at a trade-off
\begin{equation}\label{eq:scale-polylog}
 m\ \asymp\ \frac{\Dkl}{\rho^2}\cdot\mathrm{polylog}(n).
\end{equation}
This loss matters as it sits in the
exponent of the running time.  Inverting \cref{eq:scale}, the level
required at clause density $\Delta=m/n$ is
$\ell\asymp n/(\rho^2\Delta)^{2/(k-2)}$ for $k>2$, so a polylogarithmic
loss in clauses multiplies the exponent of the $n^{O(\ell)}$ running time
by a polylogarithmic factor and thus creates a superpolynomial
slowdown.  Removing these factors at growing levels has been open since the Kikuchi
hierarchy was introduced~\cite{WEM}, and through all subsequent analyses
of spectral algorithms for this and related sparse problems~\cite{GKM,HKM,HKMMS,SOKB,SZ}.  %the spectral
%template of~\cite{RRS} predates the hierarchy, and  At the constant-level endpoint a log-free polynomial-time refutation is known~\cite{dOrsiTrevisan} (see \cref{sec:related}).

In this work we close the gap: a normalized variant of the Kikuchi hierarchy achieves the tight
curve \cref{eq:scale} up to constants, with no logarithmic loss, for
detection, weak recovery, and two-sided strong refutation, at every
arity $k\ge3$ (and at $k=2$, where the curve is flat).  For detection
and weak recovery we also prove matching lower bounds in the same
fixed-$m$ model, against a class of low-degree statistics broad
enough to allow arbitrary preprocessing of the observed clause
supports.  \Cref{tab:ledger} in \cref{sec:results} records the
resulting picture task by task.

\subsection{The obstruction, and how it is removed}\label{sec:obstruction}

The logarithm in \cref{eq:scale-polylog} is not merely an artifact of loose
bookkeeping.  At small levels, the \emph{unnormalized} sparse Kikuchi
matrix has a genuine localization obstruction.  At the target curve, a
row of the level-$\ell$ matrix has mean degree of order
$\rho^{-2}\ell$, but the degrees fluctuate, and among $\binom n\ell$
rows some exceed the mean by a logarithmic factor.  A high-degree row
whose neighbors are otherwise nearly isolated forms a signed star, and
a star with $q$ edges carries an eigenvalue of size $\sqrt q$ localized
on it.  More generally, existing unnormalized
analyses pay for extreme row degrees through a \emph{maximum-degree}
estimate, and that estimate is where the logarithm enters. %whether the unnormalized operator itself becomes log-free at sufficiently large levels remains open (\cref{sec:discussion}).  
The dense Gaussian analogue of the problem, tensor PCA~\cite{RichardMontanari}, has no
degree fluctuations; this is why sharp unnormalized bounds are
possible there~\cite{KX} but do not by themselves solve sparse $k$XOR.

 Our solution to this obstruction has two
parts.  First, we normalize locally.  At even $k$, if $A$ is the signed
level-$\ell$ adjacency matrix of the Kikuchi graph, $D$ its diagonal matrix of degrees, and
$\bd$ its expected mean degree, we study
\[
 \Gamma:=D+\bd I,
 \qquad
 K:=\Gamma^{-1/2}A\,\Gamma^{-1/2}.
\]
Both terms in $\Gamma$ are necessary: the realized degree $D$ tames the
dense rows that defeat the unnormalized matrix, while the deterministic
floor $\bd I$ tames very sparse rows.  This
degree-plus-mean reweighting was introduced by Hsieh, Kothari, and
Mohanty in their work on the hypergraph Moore bound~\cite{HKM}.  At odd arity, where a single clause cannot act
on the level-$\ell$ slice, we use the same principle but with clauses
grouped into rooted pairs and with diagonal completions engineered to be
\emph{monotone}: they never decrease when clauses are added.

Second, we bound the spectrum of the normalized matrix by a sharp trace
computation.
Expanding a high even trace of $K$ produces closed walks on the Kikuchi
graph, and averaging the random signs kills every walk that uses some
clause an odd number of times.  Call a clause \emph{active} at a given
point of the walk if it has been used an odd number of times so far.
Each surviving walk then decomposes chronologically into
\emph{openings}, where a clause is used for the first time or becomes
active again, and \emph{closings}, where an active clause is reused.
Openings are paid locally.  If a walk opens a clause at row $S$, and
$M_S$ denotes the degree already created at $S$ by the clauses the walk
has exposed, then previously seen clauses have total weight at most
$M_S$ and are charged to the corresponding part of the denominator,
while one computes that fresh clauses have conditional expected mass at
most the mean degree $d=\bd$ and are charged to the floor:
\begin{equation}\label{eq:opening-identity}
 \frac{M_S}{d+M_S}+\frac{d}{d+M_S}=1.
\end{equation}

The main combinatorial task is to correctly account for closings, and this is where the second half of
the mechanism enters: a careful count of the \emph{excess}.  A fresh
clause in the walk normally explores new variables that have not appeared before in the walk. We define the excess $\chi$ to count the
variables a walk revisits instead. Each such recycled variable is
rare, costing a factor $O(t/n)$ in a walk of length $2t$.  A parity
argument shows that at most $O(\ell+\chi)$ active clauses can be chosen as next step
at the current row, so each of the $t$ closings costs
$O\bigl((\ell+\chi)/(d+\ell)\bigr)$, and summing the resulting
geometric series over the excess is what produces our sharp spectral upper bound.
\Cref{sec:techniques} gives a more detailed quantitative form of this accounting,
and \cref{sec:trace} the proofs.

The lower bounds are a mirror image of the upper bounds.  Conditioning
on the clause supports and expanding the likelihood ratio in the
orthonormal basis of sign monomials reduces everything to counting families of clauses whose supports cancel to a prescribed parity boundary. The even covers generated by the trace walks are precisely the case of empty boundary.
A direct ordered count of these families, with no logarithmic slack,
yields detection and recovery lower bounds at the curve \cref{eq:scale}
in the same fixed-$m$ model as the algorithms.

A closely related count of even trace walks resolves Feige's 2008
conjecture on the hypergraph Moore bound~\cite{Feige} in our companion
paper~\cite{SHmoore}; a concurrent and independent proof appears
in~\cite{BKNPRW}.

\subsection{Contributions}\label{sec:contributions}

Throughout, $k\ge2$ and the bias $\rho$ (for refutation, the target
advantage $\eps$) are fixed, and the level $\ell$ ranges over
essentially the full hierarchy, up to $\ell\le n^{1-\delta}/\log n$.
Formal statements appear in \cref{sec:results}.

\begin{enumerate}
\item\emph{Sharp inference (\cref{thm:main-inference}).}  At
$m\gtrsim\rho^{-2}\Dkl$, a spectral algorithm running in
$n^{\ell+O_k(1)}$ time achieves strong detection and weak recovery (up to a global sign for even $k$).  This is the curve \cref{eq:scale}
with no logarithmic loss, at every arity and every level in the stated
range.
\item\emph{Accurate and exact recovery from an independent sample
(\cref{thm:main-boost}).}  A targeted majority vote boosts weak
recovery to error $\vartheta$ using
$O_k(\rho^{-2}n\log(1/\vartheta))$ additional clauses, and to exact
recovery using $O_k(\rho^{-2}n\log n)$.  The same logarithmic
target-error dependence appears in the unconditional total-sample
minimax rate (\cref{thm:minimax}).
\item\emph{Log-free two-sided refutation, with certificates and
sum-of-squares proofs (\cref{thm:main-refutation}).}  At
$m\gtrsim\eps^{-2}\Dkl$, an algorithm outputs a bound on the maximal
absolute advantage $\sup_x|V_I(x)|$ (\cref{eq:VI}) that is sound on
\emph{every} input
and at most $\eps$ on a random one; it can be evaluated exactly, and it
yields degree-$O_k(\ell)$ sum-of-squares proofs.
\item\emph{Matching lower bounds in the same model
(\cref{thm:main-lower}).}  Below the curve, no statistic of sign degree
$O(\ell)$ distinguishes planted from null, and no pointwise-bounded
estimator of sign degree $O(\ell)$ attains nonvanishing overlap, even
when the coefficients may depend arbitrarily on the support hypergraph.

\item\emph{A quartic quantum speedup (\cref{thm:quantum-inference}).}
These normalized Kikuchi operators fit the guided sparse Hamiltonian framework
of~\cite{SOKB}, giving quantum detection and weak recovery with gate complexity $\widetilde O(n^{\ell(1/4 + o(1))}) \cdot \mathrm{poly}(n)$. This improves existing quantum algorithms for this problem by extending them from even $k$ to all $k$, proving efficient recovery, and sharpening the SNR-runtime trade-off to match the log-free sample curve. The stated runtime is a
quartic quantum speedup in the Kikuchi dimension when $\ell$ is taken to grow slowly with $n$, in particular for $\ell = n^{o(1)}$ (\cref{rem:quantum-scope}).

%, we also prove efficient recovery and extend from  on the
%log-free sample curve, at both parities, with gate complexity
%.  
\item\emph{Transfer to Boolean CSPs (\cref{sec:boolean-csp}).}  Every
nonzero Fourier character of a fixed planting law is an exact noisy-XOR
channel, so the inference upper bounds apply channelwise, to every
character of degree at least two, for fixed known nonuniform local
planting laws (degree-one characters are handled directly); the
refutation and sum-of-squares upper bounds extend to every fixed
predicate in the uniform random-scope/random-literal model.
\end{enumerate}

\subsection{Related work}\label{sec:related}

\paragraph{The hierarchy and its uses.}
The Kikuchi hierarchy was introduced by Wein, El~Alaoui, and
Moore~\cite{WEM}, for tensor PCA and for the refutation of random
even-arity $k$XOR, and subsequently became a basic tool for random-CSP
refutation and planted inference~\cite{GKM,HKM,HKMMS,SOKB}, and for
lower bounds on locally decodable and correctable
codes~\cite{AGKM,KothariManohar}.  Our refutation results and lower bounds build on the
spectral and sum-of-squares literature for random
CSPs~\cite{Feige02,CGL,AOW,BarakMoitra,RRS,GKM,Grigoriev,Schoenebeck,KMOW}.
The successive sparse analyses reduced the loss in
\cref{eq:scale-polylog} from $\log^{4k+1}n$~\cite{GKM} to a single
$\log n$~\cite{HKM}, and to $(\log n)^{1/(k+1)}$ at odd
arity~\cite{HKMMS}, but never reached \cref{eq:scale}.  At the
constant-level endpoint, Naor and
Verstra\"ete~\cite{NaorVerstraete} established the high-density
endpoint $\ell=\Theta(1)$ of the corresponding even-cover problem, for
all uniformities $k\ge3$. At the same endpoint, Allen, O'Donnell, and Witmer~\cite{AOW} (even $k$) 
and d'Orsi and
Trevisan~\cite{dOrsiTrevisan} (odd $k$) obtained log-free polynomial-time strong
refutation, the later via a non-backtracking Ihara--Bass
construction; neither extends to the growing-level trade-off proved
here.  %Sharp matrix concentration in the
%sense of~\cite{BBvH,BvH}
%does not remove the loss on its own: it demotes %the dimensional factor
%of noncommutative Khintchine and matrix Bernstein
%bounds~\cite{Tropp} to a lower-order term multiplying the
%free-probability prediction for the spectral norm.
Kothari and Xu~\cite{KX} obtained sharp unnormalized bounds in the dense
setting of tensor PCA, answering the analogous conjecture in \cite{WEM}.  Kothari~\cite{KothariECCC} proves a
whole-spectrum approximation for random even-uniform Kikuchi graphs,
sharp up to a logarithmic factor, and Basu, Hsieh, Lin, and
Manohar~\cite{BHLM} obtain list recovery for randomly planted CSPs
through a degree-$O(\ell)$ sum-of-squares relaxation, with a
logarithmic sample loss at fixed bias (and worse known bias dependence
in the noisy specialization).

\paragraph{Recent work of Mao.}
In recent work, Mao~\cite{Mao2026} gave a log-free recovery, and hence
detection, algorithm for $k\ge3$ that does not go through Kikuchi
matrices, using color coding and dynamic programming.  His formal
theorem applies when the degree parameter $D$ is a sufficiently large
multiple of $k$ with $Ck\le D\le0.1\log n/\log\log n$, for an absolute
constant $C$.  It achieves exact recovery with constant success
probability at the threshold, and with probability $1-o(1)$ when
$D\to\infty$ or the sample size is asymptotically above the threshold.
The two results are complementary: Mao's algorithm is non-spectral, whereas ours is spectral,
covers the full level range \cref{eq:range} at fixed bias, and
additionally produces pointwise-sound refutation certificates and
sum-of-squares proofs.  Mao also proves a low-degree detection converse
at the same curve, for ordinary degree-$D$ polynomials in a Bernoulli
observation model, whereas our converse is proved directly in the fixed-sample
occurrence model, permits arbitrary preprocessing of the multihypergraph, and also treats pointwise-bounded recovery
estimators (\cref{def:sign-degree}).

\paragraph{Lower-bound framework.}
Our lower bounds use the low-degree framework
of~\cite{HopkinsSteurer,KWB,SchrammWein,BHK+16,HKP+17,Hop18}, extended to allow arbitrary
preprocessing of the sampled supports.  They match detection and
pointwise-bounded recovery within that framework. Naturally, these are not
unconditional runtime lower bounds.  \Cref{rem:ld-status} discusses what
such bounds do and do not certify.

\paragraph{Earlier hierarchy lower bounds.}
Earlier evidence for the polynomial-time barrier came from lower bounds
for polynomial calculus and for LP/SDP
hierarchies~\cite{Grigoriev,AlekhnovichRazborov,Schoenebeck,ODonnellWitmer,MoriWitmer,KMOW}.
These results rule out specified proof systems or relaxation
hierarchies rather than arbitrary algorithms; the trade-off
of~\cite{KMOW} is the most general, and for $k$XOR its degree bound
retains the denominator $\log\Delta$ (\cref{rem:kmow}).

\subsection{How to read this paper}\label{sec:reading}

\Cref{sec:results} states the model and the main theorems, and
\cref{sec:techniques} is a self-contained overview of all proofs.  After
the preliminaries of \cref{sec:prelim}, the paper splits into four
largely independent threads:
\begin{itemize}
\item\emph{Inference.}  \Cref{sec:trace} proves the two trace theorems;
\cref{sec:even,sec:odd} apply them at even and odd arity;
\cref{sec:exact} boosts weak recovery to exact recovery and completes
the proofs of \cref{thm:main-inference,thm:main-boost}.
\item\emph{Refutation.}  \Cref{sec:refutation} constructs the sound
certificates at both parities, and \cref{sec:evaluation} supplies exact
evaluation, sum-of-squares proofs, and the proof of
\cref{thm:main-refutation}.  This thread needs \cref{sec:trace}, the
even matrix and its null norm from \cref{sec:even}
(\cref{eq:even-operator,prop:even-null}), and the mutual exact-two rule
and the counting quantities $D_+$ and $\Delta_{r,\ell}$ from
\cref{sec:odd}; it is independent of the rest of the inference thread.
\item\emph{Lower bounds.}  \Cref{sec:lower} is self-contained given
\cref{sec:results} and proves \cref{thm:main-lower} and
\cref{thm:minimax}.
\item\emph{Extensions.}  \Cref{sec:quantum} develops the quantum
algorithm, and \cref{sec:boolean-csp} the transfer to Boolean CSPs;
the CSP transfer uses the earlier upper bounds as black boxes, while
the quantum section builds quantum-ready variants of the normalized
operators.
\end{itemize}
\Cref{sec:discussion} collects open problems.

\section{Model and main results}\label{sec:results}

In planted noisy $k$XOR one observes an ordered sequence
$I=((F_1,y_1),\ldots,(F_m,y_m))$ of clauses.  The supports
$F_a\in\binom{[n]}k$ are independent and uniform, and the labels
$y_a\in\{\pm1\}$.  Under the null law $Q_m$, the labels are independent
unbiased signs, independent of the supports.  Under the planted law $P_{m,x^\star}$, there is a
hidden assignment $\xstar\in\{\pm1\}^n$ and
\begin{equation}\label{eq:model}
 y_a=\eta_a\,\xstar^{F_a},
 \qquad
 \E\eta_a=\rho,
 \qquad
 x^F:=\prod_{i\in F}x_i\ \ \text{for }x\in\{\pm1\}^n,
\end{equation}
where the noise signs $\eta_a\in\{\pm1\}$ are independent of one
another and of the supports.  Repeated supports
remain distinct constraints.  Each clause asserts the parity equation
$x^{F_a}=y_a$; for an assignment $x\in\{\pm1\}^n$ we write $\val_I(x)$
for the fraction of the $m$ equations it satisfies, and
\begin{equation}\label{eq:VI}
 V_I(x):=2\val_I(x)-1=\frac1m\sum_{a=1}^m y_ax^{F_a}
\end{equation}
for its advantage over random guessing.  The problem is also known as sparse learning parities with noise~\cite{Ale03}, specializing
~\cite{BKW}, and is a standard testbed for
statistical--computational gaps in random constraint
satisfaction~\cite{FPV}; at $k=2$ it is the censored block
model~\cite{ABBS}.

Two further probability spaces appear alongside $Q_m$.  Let
$\mathbf P_m$ denote the joint law of $(\xstar,I)$, where $\xstar$ is
uniform on $\{\pm1\}^n$ and, conditional on $\xstar=x$, the instance
has law $P_{m,x}$; and let
\[
 P_m:=2^{-n}\sum_{x\in\{\pm1\}^n}P_{m,x}
\]
be the corresponding marginal law on instances, the \emph{planted
mixture}.  Detection
distinguishes $P_m$ from $Q_m$, and recovery probabilities refer to the
joint law $\mathbf P_m$.  The three
tasks are as follows.

\begin{definition}[Detection, recovery, refutation]\label{def:tasks}
\begin{enumerate}
\item An algorithm achieves \emph{strong detection} if it outputs a value
in $\{\textsf{planted},\textsf{null}\}$ whose total error,
$P_m(\textsf{null})+Q_m(\textsf{planted})$, tends to $0$.
\item An algorithm achieves \emph{weak recovery} if, under
$\mathbf P_m$, it outputs $\widehat x\in\{\pm1\}^n$ with
$|\ip{\widehat x}{\xstar}|\ge c_kn$ with probability $1-o(1)$, for a
constant $c_k>0$; it achieves \emph{$\vartheta$-accurate recovery} if
$\err_k(\widehat x,\xstar)\le\vartheta$ with probability $1-o(1)$, and
\emph{exact recovery} if $\vartheta<1/n$.  For odd $k$ the orientation
of $\xstar$ is identifiable, and our algorithms achieve the oriented
guarantee $\ip{\widehat x}{\xstar}\ge c_kn$.
\item An algorithm achieves \emph{two-sided strong refutation at
advantage $\eps$} if, on \emph{every} input $I$, it outputs a value
$\Cert(I)\ge\sup_x|V_I(x)|$, and under $Q_m$ it outputs
$\Cert(I)\le\eps$ with probability $1-o(1)$.  Strong refutation
ordinarily asks only for a certified upper bound on $\sup_xV_I(x)$;
controlling $\sup_x|V_I(x)|$ is the stronger two-sided statement, and
the two coincide at odd $k$, where $V_I(-x)=-V_I(x)$.
\end{enumerate}
%The inference algorithms use randomized subsampling, candidate
%generation, rounding, validation, and eigensolvers.  Even refutation is
%deterministic.  Odd refutation uses a random rooting, but its soundness
%holds for every rooting.  
All probabilities above are over the instance
and the internal randomness jointly.  %The asymmetry in (3) is the point
%of a certificate: soundness is pointwise and unconditional, and only
%the value attained is a probabilistic statement about the null model.
\end{definition}

Throughout the classical upper bounds, the arity $k\ge2$, the range
parameter $\delta>0$, and the bias $\rho\in(0,1]$ are fixed
independently of $n$; for refutation, the target advantage
$\eps\in(0,1)$ is also fixed.  The target recovery error $\vartheta$ may
depend on $n$.  Put $r:=\lfloor k/2\rfloor$ and assume
\begin{equation}\label{eq:range}
\begin{cases}
 k/2\le\ell\le n/3,&k\text{ even},\\
 k-1\le\ell\le n/3,&k\text{ odd},
\end{cases}
\qquad
 \ell\log n\le n^{1-\delta}.
\end{equation}
For $\widehat x\in\{\pm1\}^n$, the recovery error is measured by
\begin{equation}\label{eq:recovery-error}
 \err_k(\widehat x,\xstar):=
 \begin{cases}
 n^{-1}\min_{\sigma\in\{\pm1\}}d_H(\widehat x,\sigma\xstar),
   &k\text{ even},\\[1mm]
 n^{-1}d_H(\widehat x,\xstar),&k\text{ odd},
 \end{cases}
\end{equation}
which quotients by the global sign exactly when the model requires it.
Our main results are as follows.

\begin{theorem}[Detection and weak recovery at level $\ell$]
\label{thm:main-inference}
There are constants $C_{k,\delta},c_k>0$ such that, for every level satisfying
\cref{eq:range},
\begin{equation}\label{eq:inference-threshold}
 m\ge C_{k,\delta}\frac{\Dkl}{\rho^2}
\end{equation}
suffices for a spectral algorithm running in $n^{\ell+O_k(1)}$ time to
achieve strong detection and to output, under planting, an assignment
$\widehat x$ with
\[
 \begin{cases}
 |\ip{\widehat x}{\xstar}|\ge c_kn,&k\text{ even},\\[1mm]
 \ip{\widehat x}{\xstar}\ge c_kn,&k\text{ odd},
 \end{cases}
\]
with probability $1-o(1)$.
\end{theorem}

\begin{theorem}
\label{thm:main-boost}
Let $\widetilde x$ satisfy $|\ip{\widetilde x}{\xstar}|\ge cn$ for a
constant $c>0$.  There is a constant $\vartheta_0(k)\in(0,1)$ such
that, for every $1/(2n)\le\vartheta\le\vartheta_0(k)$, an
independent sample of
\begin{equation}\label{eq:cleanup-curve}
 m_{\rm c}\ge C_{k,c}\frac n{\rho^2}\log\frac1\vartheta
\end{equation}
clauses suffices for an $O(km_{\rm c}+n)$-time majority vote to output
$\widehat x$ with $\err_k(\widehat x,\xstar)\le\vartheta$ with
probability $1-o(1)$.  Taking $\vartheta=1/(2n)$ gives exact recovery.
\end{theorem}

\begin{remark}
\label{rem:cleanup-converse}
Every estimator based on $m$ total clauses has minimax expected error
at least $c_k\exp(-C_k\rho^2m/n)$, so minimax expected error $\vartheta$
requires $m=\Omega_k(n\rho^{-2}\log(1/\vartheta))$; by a separate
argument, exact recovery with probability $1-o(1)$ requires
$m=\Omega_k(n\rho^{-2}\log n)$ (\cref{thm:minimax}).  These are lower
bounds on the \emph{total} sample size, matching the target-error
dependence of \cref{eq:cleanup-curve}; they do not assert that an
additional fresh cleanup sample is necessary once the spectral stage's
output is given.  The argument is standard, and sharp fixed-noise
thresholds are known in closely related parity-measurement
models~\cite{HWX,ALS}; see the preamble to \cref{sec:minimax} for
the precise statements and model distinctions.  
%We include
%\cref{thm:minimax} because we need the bound %uniformly in $\rho$, as a
%rate for the expected error, and in the same %fixed-$m$ model as the
%algorithms.  As a by-product, it certifies the %$n\rho^{-2}\log n$
%total-sample floor as necessary in the %noisy-$k$XOR specialization of
%the setting of~\cite{BHLM}.
\end{remark}

For fixed $k\ge3$, the exact-recovery cleanup term
$O_k(\rho^{-2}n\log n)$ is lower order than the spectral term throughout
\cref{eq:range}; for $k=2$ it is the familiar coverage cost.

\begin{theorem}[Two-sided strong refutation at level $\ell$]
\label{thm:main-refutation}
Fix $\eps\in(0,1)$ and a level satisfying \cref{eq:range}.  There is an
algorithm which, on every input,
outputs a bound satisfying with certainty
\[
 \Cert(I)\ge\sup_{x\in\{\pm1\}^n}|V_I(x)|.
\]
For a random null instance,
\begin{equation}\label{eq:refutation-threshold}
 m\ge C_{k,\delta}\frac{\Dkl}{\eps^2}
\end{equation}
implies $\Cert(I)\le\eps$ with probability $1-o(1)$.  The even
algorithm is deterministic.  The odd algorithm uses a random rooting, but
soundness holds for every rooting.  In the standard explicit-input regime
$m=n^{O_k(1)}$, deterministically sound exact evaluation takes
$n^{O_k(\ell)}$ bit
operations and yields a degree-$O_k(\ell)$ SoS proof.  %For larger $m$,
%the odd construction uses deterministic threshold-sized blocks and
%costs $O(m)+J n^{\ell+O_k(1)}$ to construct and
%$J n^{O_k(\ell)}\operatorname{poly}(\log m)$ bit %operations to evaluate,
%%where $J=\Theta(m\eps^2/\Dkl)$ is the number of blocks.
\end{theorem}

The lower bounds are phrased for a class of statistics that is
deliberately generous on the support side.  Condition on the supports
$F=(F_1,\ldots,F_m)$; the monomials $y_A:=\prod_{a\in A}y_a$ then form
an orthonormal basis under the null.  A statistic has \emph{sign degree}
$D$ if only terms with $|A|\le D$ occur, while their coefficients may
depend arbitrarily on the full support sequence.  This class contains
the usual degree-$D$ polynomial encodings and permits unrestricted
preprocessing of the support hypergraph
(\cref{def:sign-degree,lem:class-containment}).

\begin{theorem}[Matching inference lower bounds]
\label{thm:main-lower}
For every $k\ge2$, there are constants $c_k,C_k>0$ such that, in the same
fixed-$m$ model:
\begin{enumerate}
\item If $D\le n/(2k)$ and
$m\rho^2D^{k/2-1}/n^{k/2}\le c_k$, every positive-variance sign-degree-$D$
statistic $g$ satisfies
\[
 \frac{|\E_Pg-\E_Qg|}{\sqrt{\Var_Q(g)}}
 \le C_k\frac{m\rho^2D^{k/2-1}}{n^{k/2}},
\]
where $Q=Q_m$ and $P=P_m$ is the planted mixture
(\cref{thm:ld-detection}).
\item If $d\le n/(4k)$ and
$m\rho^2(2d)^{k/2-1}/n^{k/2}\le c_k$, every estimator
$f=(f_1,\ldots,f_n)$ with pointwise bounds $|f_i|\le1$ and sign degree at
most $d$ in each coordinate satisfies
\[
 \E\left(\frac1n\sum_i x_i^\star f_i(I)\right)^2
 \le \frac1n+C_k\sqrt{\frac dn}.
\]
Here the expectation is over a uniform planted assignment and its planted
instance.  Thus such estimators have vanishing overlap for $d=o(n)$ below the
same curve (\cref{thm:ld-recovery}).
\end{enumerate}
Together with the unconditional cleanup converse of
\cref{rem:cleanup-converse}, this matches every upper bound of
\cref{thm:main-inference,thm:main-boost} within its class.
\end{theorem}

\begin{remark}[Quantitative extension]\label{rem:quantitative-ranges}
The formulation above with $\rho$ and $\eps$ set to  fixed constants is chosen for readability.  Our proofs
are actually uniform in the wider ranges
\begin{align}
 n^{-k}\le\rho\le1
 &\qquad\text{(even inference: polynomial pool size in the comparison)},
 \label{eq:even-floor}\\
 \rho^{-2}\sqrt{\ell/n}\le c_{k,\delta}
 &\qquad\text{(odd inference: cell sparsity, \cref{sec:odd})},
 \label{eq:modbias}\\
 \ell/n\le c_{k,\delta}\eps^{12}
 &\qquad\text{(odd refutation: deficit budget, \cref{sec:refutation})}.
 \label{eq:odd-ref-modbias}
\end{align}
%The even certificate and the lower bounds need no condition beyond
%their stated hypotheses.  For fixed positive %$\rho$ and $\eps$, all
%three conditions follow from \cref{eq:range} for %sufficiently large
%$n$.
\end{remark}

\begin{corollary}[The matched inference curve]\label{cor:curve}
For $k>2$ and clause density $\Delta=m/n$, the inference threshold is,
up to constants depending on $(k,\delta)$,
\[
 \ell\gtrsim \frac{n}{(\rho^2\Delta)^{2/(k-2)}}.
\]
Together with \cref{thm:main-lower}, this shows that the spectral upper
curve agrees, up to constants, with the sign-degree
obstruction in the same fixed-sample model.  Replacing $\rho$ by $\eps$
gives the refutation upper curve; the best canonical sum-of-squares
converse there still has polylogarithmic slack.
\end{corollary}
%\begin{proof}
%Divide \cref{eq:inference-threshold} by $n$, use
%$\Dkl/n=(n/\ell)^{k/2-1}$, and solve for $\ell$.
%\end{proof}

\begin{table}[t]
\centering\small\setlength{\tabcolsep}{5pt}\renewcommand{\arraystretch}{1.2}
\begin{tabular}{@{}>{\raggedright\arraybackslash}p{0.14\textwidth}
>{\raggedright\arraybackslash}p{0.27\textwidth}
>{\raggedright\arraybackslash}p{0.19\textwidth}
>{\raggedright\arraybackslash}p{0.30\textwidth}@{}}
\toprule
Task & Best prior general-level spectral/SoS bound & This work &
 Matching converse\\
\midrule
Detection &
 $\rho^{-2}\mathsf D\log n$, even $k$~\cite{WEM,SOKB};
 polylogarithmic losses at all $k$, via
 refutation~\cite{RRS,GKM} &
 $\rho^{-2}\mathsf D$, all $k\ge2$ &
 $\rho^{-2}\mathsf D$ against sign-degree-$O(\ell)$ tests
 (\cref{thm:ld-detection})\\
Weak recovery &
 $\mathsf D\log n$ at fixed bias, via degree-$O(\ell)$ SoS list
 recovery in related planted-CSP models~\cite{BHLM} &
 $\rho^{-2}\mathsf D$, all levels in \cref{eq:range} &
 $\rho^{-2}\mathsf D$ against pointwise-bounded sign-degree-$O(\ell)$
 estimators (\cref{thm:ld-recovery})\\
Recovery to error $\vartheta$ &
 exact/list recovery with a $\log n$ sample loss, in related planted
 models~\cite{BHLM} &
 $+\,C_k\rho^{-2}n\log(1/\vartheta)$ &
 $\Omega_k(\rho^{-2}n\log(1/\vartheta))$, total-sample minimax
 (\cref{thm:minimax})\\
Strong refutation (two-sided) &
 $\mathsf D\operatorname{polylog}n$ at fixed
 $\eps$~\cite{RRS,GKM,HKM} &
 $\eps^{-2}\mathsf D$, sound on every input, SoS &
 canonical SoS lower bound, polylogarithmic slack~\cite{KMOW}\\
\bottomrule
\end{tabular}
\caption{The level-$\ell$ trade-off, task by task, at the scale
$\mathsf D=\Dkl$; constants depending on $(k,\delta)$ are suppressed.
At growing level, the prior bounds shown all retain logarithmic losses;
log-free constant-level endpoints are discussed in \cref{sec:related}.
The first two rows have matching upper and fixed-sample sign-degree
lower curves, and the cleanup row's target-error dependence matches the
total-sample minimax rate; the refutation upper curve is log-free, but
its canonical sum-of-squares converse is not.  In recent
\emph{non-spectral} work, Mao~\cite{Mao2026} achieves the log-free
curve for small $\ell \leq 0.1 \log(n)/\mathrm{loglog}(n)$ for detection and exact recovery at $k\ge3$ by color coding and
dynamic programming, with constant success probability at the threshold
and probability $1-o(1)$ above it; his low-degree detection converse is
formulated for ordinary degree-$D$ polynomials in a Bernoulli
observation model (\cref{sec:related}).}
\label{tab:ledger}
\end{table}

\begin{remark}[What ``sharp'' means in this paper]\label{rem:sharp-meaning}
Sharpness here is task-specific.  Detection and pointwise-bounded
recovery match their fixed-sample sign-degree converses, and the
cleanup term's target-error dependence matches the unconditional
total-sample minimax rate, by a standard argument
included to complete the picture.  The first two converses are
theorem-level lower bounds for specified classes of statistics, not
unconditional runtime lower bounds, and the spectral algorithms
themselves are not shown to be sign-degree-$O(\ell)$ polynomials.  For
refutation we prove a log-free upper curve, not a matching canonical
sum-of-squares lower bound.  \Cref{rem:ld-status} discusses the
evidential status of low-degree converses.
\end{remark}

Two extensions are developed in
\cref{sec:quantum,sec:boolean-csp}.  For $\ell\ge k$, the normalized constructions admit a quantum
algorithm for detection and weak recovery with gate complexity
$\widetilde O(\binom{n}{L}^{1/4}L^{O_k(L)}n^{O_k(1)})$ in the range
$L^2\log n=o(n)$ (\cref{thm:quantum-inference}).  And since every
nonzero Fourier character of a fixed local planting law is an exact
noisy-XOR channel, the inference upper bounds transfer channelwise (for degree $\geq 2$) to
fixed known nonuniform local planting laws, while the refutation and
sum-of-squares upper bounds transfer to fixed predicates in the uniform
random-scope/random-literal model.

\section{Overview of the proofs}\label{sec:techniques}

All the spectral and certificate upper bounds in this paper share two
ingredients: a locally normalized matrix whose null norm we can bound
sharply, and a task-specific comparison that converts the norm bound
into a planted witness or a sound certificate.  Recovery adds a
rounding step (which we prove without any eigengap assumption on the Johnson graph).  The lower
bounds count closely related parity configurations in the fixed-sample
likelihood expansion.  This section explains each piece in sequence.

\subsection{One normalized matrix: the even case}\label{sec:tech-even}

Let $k=2r$ be even and fix a level $\ell$.  The level-$\ell$ Kikuchi
matrix is indexed by the $\ell$-element subsets $S\subseteq[n]$.  A
single clause $(F,y)$ contributes the signed partial matching that pairs
$S$ with its symmetric difference $S\sd F$ whenever exactly half of $F$
lies in $S$; summing over clauses gives a signed adjacency matrix $A$
with diagonal degree matrix $D$.  Under planting,
each edge sign is biased toward its planted parity $\xstar^{F_a}$, so
the top of the spectrum of $A$ sees the signal; under the null, the
signs are independent, and the first question is how tightly the
null norm concentrates.

As explained in
\cref{sec:obstruction}, it does not concentrate tightly enough without an additional normalization due to localized high-degree structures. 
We thus work with the
interpolation
\[
 \Gamma:=D+\bd I,
 \qquad
 K:=\Gamma^{-1/2}A\,\Gamma^{-1/2},
\]
with $\bd$ the mean degree. At odd arity the same
principle applies, with a locally sparse numerator and a diagonal
completion that only increases when clauses are added; we return to
this in \cref{sec:tech-odd}.

\subsection{The trace count and the charging scheme}\label{sec:tech-trace}

Our main technical ingredient is a sharp bound on
$\|\Gamma^{-1/2}A\,\Gamma^{-1/2}\|$ under the null, proved by the trace
method with an accounting that never inspects the maximum degree.
Expand a high even trace $\E\Tr K^{2t}$.  Each term is a closed walk of
length $2t$ through $\ell$-subsets, labeled at every step by the clause
used, and carrying one inverse denominator for each departure row.
Averaging the random signs kills every walk in which some clause is
used an odd number of times, so a surviving walk can be read
chronologically as a sequence of \emph{openings}, where a clause is used
for the first time or becomes active again, and \emph{closings}, where
an active clause is reused (\cref{fig:walk}).

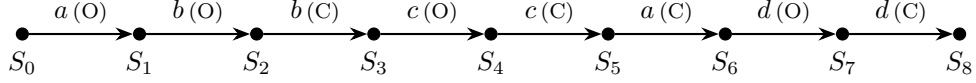
\begin{figure}[t]
\centering
\begin{tikzpicture}[scale=1.0,every node/.style={font=\small},>=Stealth]
  \foreach \i in {0,...,7} {
    \node[draw,circle,inner sep=1.6pt,fill=black] (s\i) at (\i*1.55,0) {};
  }
  \node[draw,circle,inner sep=1.6pt,fill=black] (s8) at (8*1.55,0) {};
  \foreach \i in {0,...,8} {
    \node[below=3pt] at (\i*1.55,0) {$S_{\i}$};
  }
  \draw[->,thick] (s0) -- node[above]{$a$\,{\scriptsize(O)}} (s1);
  \draw[->,thick] (s1) -- node[above]{$b$\,{\scriptsize(O)}} (s2);
  \draw[->,thick] (s2) -- node[above]{$b$\,{\scriptsize(C)}} (s3);
  \draw[->,thick] (s3) -- node[above]{$c$\,{\scriptsize(O)}} (s4);
  \draw[->,thick] (s4) -- node[above]{$c$\,{\scriptsize(C)}} (s5);
  \draw[->,thick] (s5) -- node[above]{$a$\,{\scriptsize(C)}} (s6);
  \draw[->,thick] (s6) -- node[above]{$d$\,{\scriptsize(O)}} (s7);
  \draw[->,thick] (s7) -- node[above]{$d$\,{\scriptsize(C)}} (s8);
  \path (-0.15,-0.75) -- node[below=5pt]{closed: $S_{8}=S_0$; every label used an even number of times} (12.55,-0.75);
\end{tikzpicture}
\caption{A surviving trace walk of length $2t=8$ with clause labels
$a,b,c,d$.  Each step is an opening (O) or a closing (C) of an active
label; a step is an opening when its label is fresh, or seen before but
currently inactive.  There are exactly $t$ of each.  Openings are
charged by the identity
\cref{eq:opening-identity}; closings are charged through the parity
bound of \cref{lem:active} on active legal labels.  A fresh label that
reuses previously exposed vertices raises the \emph{excess} $\chi$ and
pays an additional factor $O(t/n)$ per recycled vertex.}
\label{fig:walk}
\end{figure}

The openings are paid \emph{locally}, through the
identity \cref{eq:opening-identity} of \cref{sec:obstruction}: with
$d=\bd$ the floor and $M_S$ the degree already created at the departure
row $S$ by the clauses exposed on the current branch of the trace
expansion, previously exposed openings are charged to the realized
degree and fresh ones to the floor, at total cost one.  This identity
replaces the maximum-degree union bound of the previous Kikuchi
analyses.

Closings contain the real combinatorics.  At any moment, only clauses
used an odd number of times are \emph{active}, and a parity argument
(\cref{lem:active}) shows that at most $(\ell+2\chi)/r$ active clauses
can be legal at the current row, where the \emph{excess} $\chi$ counts
the vertices that were recycled when fresh clauses were introduced.
Every recycled vertex costs a factor $O(t/n)$.  There are exactly $t$
closings, so after summing over $\chi$ the moment takes the form
\[
 \E\Tr K^{2t}
 \le \binom n\ell
 \left(\frac{C\ell}{rd+\ell}\right)^t
 \sum_{\chi\ge0}
 \left(\frac{Ct}{n}\right)^\chi
 \left(1+\frac{2\chi}{\ell}\right)^t.
\]
Taking $t\asymp\ell\log n$ makes the last sum geometric in the range
\cref{eq:range} and yields the log-free norm scale
$\|K\|\lesssim\sqrt{\ell/(rd+\ell)}$, with failure probability
$n^{-\Omega(\ell)}$.  \Cref{tab:charges} summarizes the accounting.
At the sample threshold, $d\asymp\rho^{-2}\ell$, so this norm scale is
$\lesssim\rho$, exactly small enough to separate null from planted.

\begin{table}[t]
\centering\small\renewcommand{\arraystretch}{1.35}
\begin{tabular}{@{}llll@{}}
\toprule
move & available weight & charge & paid by\\
\midrule
opening, exposed clause & $\le M_S$ & $M_S/(d+M_S)$ & realized degree\\
opening, fresh clause & mass $\le d$ & $d/(d+M_S)$ & the floor $d$\\
\quad with $h$ recycled vertices & mass $\le d\,(Ct/n)^h$ &
 $\frac{d}{d+M_S}\,(Ct/n)^h$ & floor $+$ excess\\
closing & $\le(\ell+2\chi)/r$ active legal &
 $\le\dfrac{\ell+2\chi}{rd+\ell+2\chi}$ & parity (\cref{lem:active})\\
\bottomrule
\end{tabular}
\caption{The charging scheme for one step of a trace walk, at the
departure row $S$ with branch-exposed degree $M_S$.  The first two
charges sum to exactly one, which is the identity
\cref{eq:opening-identity}, and the third row refines the second when a
fresh clause recycles $h$ previously exposed vertices.}
\label{tab:charges}
\end{table}

\subsection{Dependence, branches, and monotone completions}
\label{sec:tech-branch}

One issue hides inside this description, which we now describe.  The walk \emph{branches} on the choice of each
fresh clause, so incompatible prefixes of the trace expansion do not
share a filtration: there is no single conditioning under which all the
charges above could be justified simultaneously.  We therefore argue one
branch at a time.  Fix a branch and attach each inverse denominator to
its departure row.  The denominators of the completed instance can only
exceed the floor plus the degree created by the clauses already exposed
on that branch,
because unexposed clauses add nonnegative degree; replacing each inverse
denominator by its chronological upper bound therefore only increases
the weight, and makes it measurable with respect to the data of that one
branch.

For the Poisson constructions at odd arity, one more monotonicity is
needed: inserting the unexposed background may \emph{destroy} previously
retained numerator terms (a new candidate can spoil an exact-two
condition) but must only \emph{increase} the completed denominator.  We
call a diagonal with this property a \emph{monotone completion}, and we
design one for each construction in this paper.  The multivariate
Mecke equation, stated in the monotone deletion form of
\cref{lem:poisson-insertion} together with its fixed-sample analogue,
is iterated through the trace
by the backward recursion of \cref{lem:branch-recursion}; this
integrates each fresh label at its first use and justifies the charges
branch by branch.  The single-label trace theorem
(\cref{thm:single-trace}) carries out the core argument; the symmetric
pair theorem (\cref{thm:packet-trace}) runs the same accounting for
moves that toggle two residual labels of a single signed family at
once, under four short structural hypotheses that each application
verifies in a few lines.

\subsection{Odd arity: rooted pairs}\label{sec:tech-odd}

For $k=2r+1$, a single clause toggles an odd number of vertices and
maps the level-$\ell$ slice out of itself.  We restore parity by
pairing.  Root every clause at one of its variables, leaving an even
residual of size $2r$, and move along pairs of clauses with the same
root and disjoint residuals.  The common root cancels, so the pair
again acts on $\ell$-subsets.  One family of rooted
proposals suffices: under the null every clause carries its own
independent sign, so the trace average forces \emph{both} labels of
every surviving step to even multiplicity, and the opening/closing
accounting runs symmetrically in the two labels of a move.  A
\emph{mutual exact-two} rule (the pair through a given root and row is
retained only if its two members are the only candidates legal there,
at both endpoints) makes each root's contribution a partial matching.
Because a candidate arriving later can destroy exact-twoness, the
retained numerator shrinks under insertion, and it is normalized by a
completion that charges the \emph{potential} pairs of each root--row
cell, capped at one, and therefore only grows: exactly the monotone
structure that \cref{sec:tech-branch} demands.  The same cap-one
principle serves classical inference, refutation, and the quantum
oracle, and the classical and quantum inference operators coincide.  The pair trace then gives the null scale $\sqrt{\ell/d}$,
where the floor $d$ is a pair baseline fixed in \cref{sec:odd}.  At the
odd sample curve $d\asymp\rho^{-4}\ell$, so the scale is $\rho^{2}$,
matching the $\rho^{2}$ signal of a rooted pair, whose two clauses each
carry bias $\rho$.

Inference and refutation then part ways in how they treat the planted
mean and the clauses that the retention discards.  For inference, we
adjoin independent \emph{dummy} candidates so that every rooted type
has the same intensity.  This smooths the planted mean into a highly
symmetric matrix whose entries depend only on the distance between two
$\ell$-subsets, which is the structure that recovery needs; it also
keeps the linear part of the centered planted matrix small
(\cref{sec:odd-comparison}).  For refutation, no clause may be thrown
away silently: the same exact-two graph is built from the actual
clauses, every omitted transition is charged to an explicit scalar
deficit controlled by the triple collisions of the rooted cells, and
the null value is bounded by the same pair trace, run directly in the
fixed-sample model.  A generalized-eigenvalue bound plus the deficit
gives a certificate that is sound on every input and for every
rooting.

\subsection{From spectra to assignments}\label{sec:tech-recovery}

A vector with a large Rayleigh quotient has constant planted energy in a
symmetric matrix connecting pairs of $\ell$-subsets at one fixed
distance.  Elementary \emph{down operators}, which sum amplitudes over
supersets, convert this energy into constant mass in the planted
direction of a matrix of one-coordinate marginals, the one-particle
reduced density matrix (\cref{thm:transfer}).  A constant-rank spectral
projector then yields a short list of candidate assignments without any
eigengap assumption; an independent sample validates the list; and a
targeted majority vote drives the error down at a cleanup cost whose
$\log(1/\vartheta)$ dependence matches the total-sample minimax rate.

\subsection{Lower bounds: counting even covers}\label{sec:tech-lower}

For the lower bounds, condition on the supports and expand in the
orthonormal basis of sign monomials.  A monomial survives the planted
average precisely when the symmetric difference of its supports equals a
prescribed boundary $B$.  For $B=\varnothing$, these are exactly the
labeled even covers that the trace walks of \cref{sec:tech-trace}
generate, which is the sense in which the lower bounds mirror the upper
bounds.  A direct ordered count of these configurations, valid even with
repeated supports, controls them with no logarithmic slack.  The case
$B=\varnothing$ bounds detection; $|B|=2$ bounds pointwise-bounded
recovery; and standard binary-testing and Fano arguments give the
minimax cleanup and exact-recovery converses.  Because the coefficients
of a sign-degree statistic may depend arbitrarily on the support
sequence, the resulting bounds tolerate arbitrary preprocessing of the
support hypergraph.

\subsection{Exact evaluation, quantum algorithms, and CSPs}
\label{sec:tech-ext}

Exact rational inertia tests turn the spectral pencils into
deterministically sound certificates, and rational positive-semidefinite
factorizations turn those into degree-$O_k(\ell)$ sum-of-squares
proofs.  The quantum extension feeds quantum-ready normalized variants
of the same operators, together with a guide state built from products
of clauses, into the guided sparse-Hamiltonian framework
of~\cite{SOKB}.  For Boolean CSPs, each nonzero Fourier character of a
known local planting law produces an exact noisy-XOR channel, so the
inference upper bounds transfer channelwise, and predicate refutation
is assembled through the Fourier expansion in the random-literal model.

\section{Preliminaries}\label{sec:prelim}

This section collects the machinery shared by the rest of the paper:
the operators living on the level-$\ell$ slice, the sampling and
deletion devices, the independent-copy transfer, the eigenvalue computations,
and the pipeline that converts
spectral energy into Boolean assignments.  It can be skimmed on first
reading and consulted as the tools are invoked.
\Cref{tab:notation} records the recurring notation; constants such as
$c_r$, $C_{k,\delta}$ depend only on their subscripts and may change
from one occurrence to the next.

\begin{table}[t]
\centering\small\renewcommand{\arraystretch}{1.15}
\begin{tabular}{@{}lll@{}}
\toprule
symbol & meaning & where\\
\midrule
$Q_m$, $P_{m,x}$, $P_m$, $\mathbf P_m$ & null, planted, planted-marginal, and joint laws & \cref{sec:results}\\
$S\sd T$;\ $A\mathbin{\dot\cup}B$ & symmetric difference; union of disjoint sets & --\\
$r$, $\Dkl$ & $\lfloor k/2\rfloor$; the sample scale $n^{k/2}/\ell^{k/2-1}$ & \cref{eq:scale,eq:range}\\
$\eta_a$; $\vartheta$; $\eps$ & noise signs ($\E\eta_a=\rho$); target recovery error; target advantage & \cref{sec:results}\\
$\Omega_\ell$, $N$ & the slice $\binom{[n]}\ell$ and its size $\binom n\ell$ & \cref{sec:slice}\\
$A_F$, $P_F$ & partial matching of an even set $F$; its support projector & \cref{eq:matching}\\
$J_s$, $b_{s,\ell}$ & Johnson adjacency at distance $s$; its degree & \cref{eq:johnson-degree}\\
$Z$ & planted gauge $\diag(\xstar^S)$; analysis only & \cref{eq:gauge}\\
$\partial_\ell$, $\rhoone$ & down operators; one-particle density matrix & \cref{sec:common-recovery}\\
$t$, $\chi$ & half-length and excess of a trace walk & \cref{sec:trace}\\
$d$ & the floor of a normalized trace bound & \cref{thm:single-trace}\\
$\bd$; $D$ & even mean degree; even degree matrix & \cref{eq:even-operator}\\
$\Gamma$; $K$ & degree-plus-floor normalizer; normalized Kikuchi matrix & \cref{eq:even-operator}\\
$\alpha$, $q$, $\pi$ & odd cell sparsity; candidate and clause intensities & \cref{eq:odd-pi,eq:odd-alpha}\\
$\cG_{\rm o}$; $\Gamma_{\rm o}$; $H_{\rm o}$ & capped odd completion; odd normalizer and operator & \cref{eq:odd-completed,eq:odd-H}\\
$d_\circ$; $d_*$, $\bar g$ & odd population scale; odd-refutation floor and mean capped degree & \cref{sec:odd,sec:refutation}\\
$D_{u,S}$, $D_+$ & rooted-cell counts and their maximum & \cref{eq:odd-cell,eq:odd-alpha}\\
\bottomrule
\end{tabular}
\caption{Recurring notation.  The floor $d$ is instantiated as $\bd$ at
even arity, as $n\alpha^2$ for odd inference (classically and
quantumly), and as $d_*$ for odd refutation.}
\label{tab:notation}
\end{table}

\subsection{Slice operators and the gauge}\label{sec:slice}

Fix the level $\ell$ and write
\[
 \Omega_\ell:=\binom{[n]}\ell,
 \qquad N:=|\Omega_\ell|=\binom n\ell.
\]
For an even set $F\subseteq[n]$ of size $2s$, let $A_F$ be the partial
matching on $\Omega_\ell$ defined by
\begin{equation}\label{eq:matching}
 (A_F)_{S,T}:=\1\{S\sd T=F\}.
\end{equation}
It pairs $S$ with $S\sd F$ when $|S\cap F|=s$.  Its square $P_F:=A_F^2$ is
the diagonal projector onto the legal rows,
$P_F(S,S)=\1\{|S\cap F|=s\}$.  Since $A_F$ is an involution on its support,
\begin{equation}\label{eq:domination}
 |v^{\mathsf T}A_Fv|\le v^{\mathsf T}P_Fv
 \qquad\text{for every }v.
\end{equation}
The number of legal rows depends only on $s$:
\begin{equation}\label{eq:support-count}
 \Tr P_F=\1^{\mathsf T}A_F\1
 =\binom{2s}s\binom{n-2s}{\ell-s}.
\end{equation}
If $C\cap D=\varnothing$ are two even sets, the matchings $A_C$ and $A_D$
commute, and $B_{C,D}:=A_CA_D$ is a symmetric partial matching with support
projector $B_{C,D}^2=P_CP_D$; it moves $S$ to $S\sd C\sd D$ whenever
$|S\cap C|=|C|/2$ and $|S\cap D|=|D|/2$.  For rows $S,T$, we write
$E_{S,T}$ for the matrix unit with a single one in position $(S,T)$.
Throughout, $S\sd T$ denotes the symmetric difference and
$A\mathbin{\dot\cup}B$ a union of disjoint sets.

For $0\le s\le\ell$, let $J_s$ be the distance-$s$ Johnson adjacency,
$(J_s)_{S,T}=\1\{|S\sd T|=2s\}$, with degree
\begin{equation}\label{eq:johnson-degree}
 b_{s,\ell}:=\binom\ell s\binom{n-\ell}s.
\end{equation}
The support count and the Johnson degree are related by
$N\,b_{s,\ell}=\binom n{2s}\Tr P_F$.

Finally, for the planted assignment define the diagonal gauge
\begin{equation}\label{eq:gauge}
 Z:=\diag\bigl(\xstar^S:S\in\Omega_\ell\bigr).
\end{equation}
If $S\sd T=F$ then $Z_{S,S}Z_{T,T}=\xstar^F$, so conjugation by $Z$
removes the planted monomial from every signed matching.  All
normalizers in this paper are diagonal and hence commute with $Z$.  The
gauge is used only in the analysis; no algorithm needs to know it.

\subsection{Sampling devices: subsampling, deletion, and the branch
recursion}\label{sec:devices}

All theorems in this paper concern the fixed-$m$ laws of
\cref{sec:results}.  Poisson processes enter only through the following
explicit algorithmic reduction, which is classical; we state it in the
form we use.

\begin{lemma}[Poisson subsampling]\label{lem:subsample}
Let $I$ consist of $m$ i.i.d.\ clauses from the null or a planted law, and
let $0<\mu\le m/2$.  Draw $M\sim\operatorname{Pois}(\mu)$ independently of
$I$; declare failure if $M>m$, and otherwise retain the first $M$
clauses.  Failure has probability at most $e^{-\mu/8}$.  There is a
coupling to an ideal Poisson clause process of total intensity $\mu$,
with atoms i.i.d.\ from the same one-clause law, under which, on the
event $M\le m$, the retained collection agrees with the ideal process
\emph{pathwise}.  Moreover, assigning
each retained clause independently to one of finitely many pools with
fixed probabilities produces independent Poisson pools with the
corresponding sub-intensities.
\end{lemma}
\begin{proof}
Extend $I$ to an infinite i.i.d.\ clause sequence (a proof device: the
algorithm reads past index $m$ only on the failure event, where it
stops).  The first $M$ terms of that sequence form exactly the stated
Poisson process, and on $M\le m$ they are exactly the retained clauses.
Independent assignment is Poisson thinning, and the failure bound is
the Poisson upper tail
$\bbP\{\operatorname{Pois}(\mu)>2\mu\}\le e^{-\mu/8}$.
\end{proof}

The inference algorithms of
\cref{sec:even,sec:odd,sec:exact} apply \cref{lem:subsample} once, with
$\mu$ a fixed constant multiple of the sample threshold, and analyze the
resulting Poisson pools; on the failure event the detector outputs
\textsf{null} and the recovery output is arbitrary, which adds $o(1)$ to
every error probability since $\mu\to\infty$ in all our uses.  A
refutation algorithm cannot subsample, because its certificate must cover
every clause of the given input.  This costs nothing at either parity:
the even null trace holds directly for a deterministic number of
i.i.d.\ clauses (\cref{prop:even-null}), and so does the odd pair
trace, through the fixed-sample deletion bound \cref{eq:fixed-deletion}
below; both certificates are therefore analyzed
in the fixed-$m$ model with no Poisson device at all.

The next lemmas are the dependence backbone of the trace proofs of
\cref{sec:trace}.  The odd constructions below impose local conditions
whose truth value may change when further candidates or clauses are
inserted, so a sum of incompatible trace branches can never be
conditioned on one common filtration.  Instead, after fixing one
branch, we insert its fresh labels at their first uses and compare the
completed configuration with the chronological one.  All the
probability this comparison needs is the multivariate Mecke equation,
in the following monotone form, together with its elementary
fixed-sample analogue.

\begin{lemma}[Monotone deletion]\label{lem:poisson-insertion}\label{lem:fixed-deletion}
Let $F\ge0$ be insertion-nonincreasing: adjoining a further marked
label never increases $F$.
\begin{enumerate}
\item If $\Xi$ is a marked Poisson process with finite intensity
measure $\nu$, then for every nonnegative measurable $g$ and $s\ge1$,
\begin{equation}\label{eq:poisson-insertion}
 \E\!\sum_{(z_1,\ldots,z_s)\in\Xi_{\ne}^{s}}
 F(\Xi)g(z_1,\ldots,z_s)
 \le
 \int\!\cdots\!\int
 \E F(\Xi)g(z_1,\ldots,z_s)
 \,d\nu(z_1)\cdots d\nu(z_s).
\end{equation}
\item If instead $\Xi$ consists of $m$ i.i.d.\ marked labels
$Z_1,\ldots,Z_m$ and $F$ is a function of the unordered sample, then
for $1\le s\le m$,
\begin{equation}\label{eq:fixed-deletion}
 \E\sum_{\substack{i_1,\ldots,i_s\le m\\ \text{distinct}}}
 F(Z_1,\ldots,Z_m)\,g(Z_{i_1},\ldots,Z_{i_s})
 \ \le\
 (m)_s\,\E\bigl[F(Z_{s+1},\ldots,Z_m)\,g(Z_1,\ldots,Z_s)\bigr],
\end{equation}
where $(m)_s:=m(m-1)\cdots(m-s+1)$.
\end{enumerate}
Both parts hold conditionally after any finite set of labels has been
exposed explicitly, with the intensity, respectively the number of
unexposed labels, adjusted accordingly.
\end{lemma}
\begin{proof}
For (i), the multivariate Mecke equation makes the left side equal to
$\int\cdots\int\E F(\Xi+\sum_j\delta_{z_j})\,g(z_1,\ldots,z_s)\,d\nu^s$;
delete the inserted atoms one at a time and use monotonicity.  For
(ii), exchangeability makes the left side equal to
$(m)_s\,\E[F(Z_1,\ldots,Z_m)\,g(Z_1,\ldots,Z_s)]$, and deleting
$Z_1,\ldots,Z_s$ from the argument of $F$ can only increase it.
Physical labels are part of the marks, so coincident supports remain
distinct atoms throughout.
\end{proof}

A trace expansion applies this one-step bound many times, once per
fresh label, along every branch of a finite tree of chronological
prefixes.  The following recursion packages the iteration.

\begin{lemma}[Branch recursion]\label{lem:branch-recursion}
Let $\Xi$ be a marked Poisson process of finite intensity $\nu$, and let
$\mathfrak T$ be a finite rooted tree of chronological prefixes.  At a node
$P$, a finite list of atoms has been exposed explicitly and the remaining
background is an independent copy of $\Xi$.  Every node functional below is
evaluated on this residual background; the explicitly exposed atoms are
part of the recorded prefix rather than arguments of $\Xi$.  Let
$\Phi_P(\Xi)\ge0$ be an insertion-nonincreasing prefix weight.  A child is of one of two types.
\begin{enumerate}
\item An \emph{old} child introduces no atom and multiplies by a nonnegative
factor $\psi_{P,j}(\Xi)$ such that $\Phi_P\psi_{P,j}$ remains
insertion-nonincreasing.
\item A \emph{fresh} child of type $\beta$ selects an ordered tuple of
$s_\beta$ distinct background atoms, where $s_\beta\in\{1,2\}$ (this
range suffices for all our applications).  Its transition factor
satisfies
\[
 0\le\psi_{P,\beta}(\Xi;z_1,\ldots,z_{s_\beta})
 \le h_{P,\beta}(z_1,\ldots,z_{s_\beta})/d,
\]
where $h_{P,\beta}\ge0$ is deterministic, and after the selected atoms are
moved to the explicit list the child prefix weight is
insertion-nonincreasing in the residual background.
\end{enumerate}
Let $T_P$ be the sum of the products of all future transition factors below
$P$, with $T_P=1$ at leaves.  Suppose deterministic bounds $B_P$ are chosen
backwards, with $B_P=1$ at leaves, so that pointwise at every internal node
\begin{equation}\label{eq:branch-recursion}
\begin{aligned}
 B_P\ge{}&\sum_{j\text{ old}}\psi_{P,j}(\Xi)B_{P,j}\\
 &+\sum_{\beta\text{ fresh}}
 \int \frac{h_{P,\beta}(\mathbf z)}d
 B_{P,\beta,\mathbf z}\,d\nu^{s_\beta}(\mathbf z).
\end{aligned}
\end{equation}
Here a deterministic uniform bound may replace
$B_{P,\beta,\mathbf z}$ for all atoms of one combinatorial branch type.
Then
\begin{equation}\label{eq:branch-conclusion}
 \E[\Phi_P(\Xi)T_P(\Xi)]\le B_P\,\E\Phi_P(\Xi)
\end{equation}
at every node.
\end{lemma}
\begin{proof}
Induct backwards from the leaves.  Old children are bounded directly by the
first line of \cref{eq:branch-recursion}.  For a fresh child,
the multivariate Mecke equation rewrites its contribution as
\[
 \int \E\!\left[
  \Phi_P\!\left(\Xi+\sum_i\delta_{z_i}\right)
  \psi_{P,\beta}\!\left(\Xi+\sum_i\delta_{z_i};\mathbf z\right)
  T_{P,\beta,\mathbf z}(\Xi)
 \right]d\nu^{s_\beta}(\mathbf z).
\]
Apply the inductive bound at the child.  Insertion monotonicity and the
transition estimate bound the resulting child prefix by
$\Phi_P(\Xi)h_{P,\beta}(\mathbf z)/d$.  Sum over the fresh types and combine
with the old-child bound in \cref{eq:branch-recursion}.  This proves
\cref{eq:branch-conclusion}.
\end{proof}

\begin{remark}\label{rem:fixed-recursion}
\Cref{lem:branch-recursion} holds verbatim over a deterministic number
of i.i.d.\ marked labels: a fresh child selects an ordered tuple of
distinct unexposed sample labels, the Mecke step in the proof is
replaced by \cref{eq:fixed-deletion}, and the fresh-child integral in
\cref{eq:branch-recursion} becomes
$(m')_{s_\beta}\,\E[h_{P,\beta}(\mathbf Z)B_{P,\beta,\mathbf Z}]/d$,
with $m'$ the number of unexposed labels and $\mathbf Z$ i.i.d.\ copies
of one label.  All deterministic bounds, in
particular the floor $d$, are unchanged.
\end{remark}

\subsection{The independent-copy transfer}\label{sec:copy}

Both planted comparisons in this paper average over an independent
copy: of the clause sample at even arity, and of the noise marks on a
shared candidate skeleton at odd arity.  The average cannot simply be
pulled through a norm: rare copies have atypically large degrees or
weights, and paying for them uniformly would cost exactly the logarithm
we remove.  The following elementary transfer isolates what is actually
needed.

\begin{lemma}[Independent-copy transfer]\label{lem:copy-transfer}
Let $A'$ be an independent copy of a random object $A$ and suppose
$X(A)=\E_{A'}S(A,A')$, where $S(A,A')$ is self-adjoint.  Let
$R(A,A')\succ0$ satisfy $|v^{\mathsf T}Sv|\le v^{\mathsf T}Rv$ for all
$v$.  Suppose that, on an event $\cG$ depending only on $A$,
\begin{equation}\label{eq:copy-mean}
 \E_{A'}R(A,A')\preceq C\,\Gamma(A),
\end{equation}
that there is a pair event $\cH$ with
\begin{equation}\label{eq:copy-pointwise}
 R(A,A')\preceq M\,\Gamma(A)\quad\text{on }\cG\cap\cH,
 \qquad
 \E_{A'}\bigl[\1_{\cH^c}R(A,A')\bigr]\preceq\zeta\,\Gamma(A)
 \quad\text{on }\cG,
\end{equation}
and that jointly
$\bbP\{\|R^{-1/2}SR^{-1/2}\|>\lambda_0\}\le p_0$.  Then, outside an event
of probability at most $\bbP(\cG^c)+\sqrt{p_0}$,
\begin{equation}\label{eq:copy-conclusion}
 \|\Gamma^{-1/2}X\,\Gamma^{-1/2}\|
 \le C\lambda_0+M\sqrt{p_0}+\zeta.
\end{equation}
\end{lemma}
\begin{proof}
Discard actual samples for which the conditional probability of the event
$\{\|R^{-1/2}SR^{-1/2}\|>\lambda_0\}$ exceeds $\sqrt{p_0}$; by Markov's
inequality their probability is at most $\sqrt{p_0}$.  For every
remaining actual sample in $\cG$ and every vector $v$, split the copy
expectation of $|v^{\mathsf T}Sv|$ into the good trace event, its
complement inside $\cH$, and $\cH^c$.  The three parts are at most
$C\lambda_0\,v^{\mathsf T}\Gamma v$,
$M\sqrt{p_0}\,v^{\mathsf T}\Gamma v$,
and $\zeta\,v^{\mathsf T}\Gamma v$ respectively, using
\cref{eq:copy-mean,eq:copy-pointwise}.  Since
$|v^{\mathsf T}Xv|\le\E_{A'}|v^{\mathsf T}Sv|$, taking the supremum
over $v$ proves the claim; no union bound over data-dependent vectors
is needed anywhere, which is the reason for routing the comparison
through an independent copy rather than through a net.
\end{proof}

The lemma is also used conditionally: if $A$ and $A'$ are i.i.d.\ only
conditional on a base $\sigma$-field $\cB$, and the hypotheses
\cref{eq:copy-mean,eq:copy-pointwise} hold conditionally on
$(\cB,A)$, the identical proof gives the identical conclusion.
This is the form used in \cref{sec:odd}, where the decoupled copy of
the marks shares the candidate skeleton with the actual sample.

\subsection{Eigenvalue computation and exact spectral enclosures}
\label{sec:solvers}

Detection and recovery need only a unit vector whose Rayleigh quotient
is additively close to the top eigenvalue, with no eigengap assumption;
this is a standard Krylov guarantee.

\begin{lemma}\label{lem:krylov}
Let $H$ be a symmetric $N\times N$ matrix with $\|H\|\le1$, accessible
through a matrix--vector product oracle, and let $\gamma\in(0,1)$.
Random-start block Krylov iteration with
$O(\gamma^{-1/2}\log(N/\gamma))$ products outputs a unit vector $v$ with
\[
 v^{\mathsf T}Hv\ \ge\ \lambda_{\max}(H)-\gamma
\]
with probability at least $1-N^{-2}$, and every output satisfies
$v^{\mathsf T}Hv\le\lambda_{\max}(H)$ deterministically.
\end{lemma}
\begin{proof}
Apply the random-start Lanczos guarantee of~\cite{KW} (see
also~\cite{MuscoMusco}) to $H+I\succeq0$: after
$O(\gamma'^{-1/2}\log(N/\gamma'))$ iterations it yields a unit $v$ with
$v^{\mathsf T}(H+I)v\ge(1-\gamma')\lambda_{\max}(H+I)$, except with
probability $N^{-2}$; take $\gamma'=\gamma/2$ and use
$\lambda_{\max}(H+I)\le2$.  The Krylov spaces of $H+I$ and $H$
coincide.
\end{proof}

In our applications $\gamma$ is a small constant multiple of the
planted signal scale ($\rho$ at even arity, $\rho^2$ at odd), so the
iteration count is $O_{k,\rho}(\log N)$, and $n^{O_k(1)}$ throughout
the range of \cref{rem:quantitative-ranges}.

Refutation is different: a high Rayleigh quotient is not a certificate.
What saves us is that spectral upper bounds of rational pencils can be
verified exactly.  If $C=C^{\mathsf T}$ and $\Gamma\succ0$ are rational,
then for rational $\theta$,
\begin{equation}\label{eq:pencil}
 \lambda_{\max}(\Gamma^{-1/2}C\,\Gamma^{-1/2})\le\theta
 \ \Longleftrightarrow\
 \theta\Gamma-C\succeq0,
 \qquad
 \|\Gamma^{-1/2}C\,\Gamma^{-1/2}\|\le\theta
 \ \Longleftrightarrow\
 \theta\Gamma\mp C\succeq0.
\end{equation}
Rational bisection combined with exact inertia tests therefore produces
a zero-error upper enclosure.  The inertia of an indefinite rational
symmetric matrix is computed by symmetric elimination with $2\times2$
block pivots in the style of Bunch--Kaufman~\cite{BunchKaufman} (a
plain $LDL^{\mathsf T}$ need not exist; the positive semidefinite
factorizations used for the sum-of-squares proofs are unaffected), and
fraction-free elimination controls coefficient growth since
intermediate entries are minors~\cite{Bareiss}.  Together with rational
bisection this gives a conservative dense bound of $n^{O_k(\ell)}$ bit
operations; the matrices themselves, and their products with vectors,
are constructed in $n^{\ell+O_k(1)}$ time (\cref{sec:evaluation}).

\subsection{From Johnson energy to Boolean assignments}
\label{sec:common-recovery}

The spectral arguments of
\cref{sec:even,sec:odd} both end at a unit vector with
constant positive normalized energy in one fixed Johnson channel; this
subsection is the common pipeline that turns such a vector into a
Boolean assignment of constant overlap.  Down operators transfer the
energy directly to the one-particle density matrix; no diagonalization
of the Johnson scheme is needed.

For $1\le j\le n$, define
\[
 \partial_j:\mathbb R^{\Omega_j}\longrightarrow\mathbb R^{\Omega_{j-1}},
 \qquad
 (\partial_jf)(R):=\sum_{i\notin R}f(R\cup\{i\}),
\]
and, for $1\le s\le\ell$,
\begin{equation}\label{eq:iterated-down}
 \partial_\ell^{[s]}
 :=\frac1{s!}\partial_{\ell-s+1}\cdots\partial_{\ell-1}\partial_\ell.
\end{equation}
Thus
$(\partial_\ell^{[s]}f)(R)=\sum_{S\supseteq R,\,|S|=\ell}f(S)$ for
$|R|=\ell-s$.

\begin{lemma}\label{lem:down-identities}
With $J_0:=I$,
\begin{align}
 \partial_\ell^{\mathsf T}\partial_\ell&=\ell I+J_1,
 \label{eq:down-one}\\
 (\partial_\ell^{[s]})^{\mathsf T}\partial_\ell^{[s]}
 &=\sum_{j=0}^s\binom{\ell-j}{s-j}J_j,
 \label{eq:down-s}\\
 \|\partial_j\|^2&=j(n-j+1).
 \label{eq:down-norm}
\end{align}
\end{lemma}
\begin{proof}
Two $\ell$-sets have a common $(\ell-1)$-subset exactly when they are
equal or have Johnson distance one, proving \cref{eq:down-one}.  More
generally, if $S,T\in\Omega_\ell$ have Johnson distance $j$, then the
number of $(\ell-s)$-sets contained in both is
$\binom{\ell-j}{\ell-s}=\binom{\ell-j}{s-j}$ for $j\le s$, and zero
otherwise; this proves \cref{eq:down-s}.  Finally,
$\partial_j^{\mathsf T}\partial_j=jI+J_1$ on $\Omega_j$.  This
nonnegative matrix has constant row sum $j+j(n-j)=j(n-j+1)$, which is
both its Perron eigenvalue and its operator norm.
\end{proof}

\begin{definition}[One-particle density matrix]\label{def:rdm}
For $u\in\mathbb R^{\Omega_\ell}$ define
$B_u\in\mathbb R^{[n]\times\binom{[n]}{\ell-1}}$ by
\[
 B_u(i,R):=u_{R\cup\{i\}}\1\{i\notin R\},
 \qquad
 \rhoone(u):=\frac1\ell B_uB_u^{\mathsf T}.
\]
For unit $u$, $\rhoone(u)\succeq0$ and $\Tr\rhoone(u)=1$.
\end{definition}

\begin{lemma}\label{lem:rdm-identity}
For every $u\in\mathbb R^{\Omega_\ell}$ and
$\bar x:=\xstar/\sqrt n$,
\begin{equation}\label{eq:rdm-identity}
 \bar x^{\mathsf T}\rhoone(Zu)\bar x
 =\frac{\|\partial_\ell u\|^2}{n\ell}.
\end{equation}
\end{lemma}
\begin{proof}
For $R\in\binom{[n]}{\ell-1}$,
\[
 \bar x^{\mathsf T}B_{Zu}(\cdot,R)
 =n^{-1/2}\xstar^R\sum_{i\notin R}u_{R\cup\{i\}}
 =n^{-1/2}\xstar^R(\partial_\ell u)(R).
\]
Square, sum over $R$, and divide by $\ell$.
\end{proof}

\begin{theorem}[Johnson-to-RDM transfer]\label{thm:transfer}
Fix $s\ge1$ and $0<c\le1$.  There are $c_{s,c}>0$ and $n_0(s,c)$ such
that, for every $s\le\ell\le n/3$ and $n\ge n_0(s,c)$, every unit vector
$u\in\mathbb R^{\Omega_\ell}$ satisfying
\begin{equation}\label{eq:johnson-energy}
 u^{\mathsf T}J_su\ge c\,b_{s,\ell}
\end{equation}
obeys
\begin{equation}\label{eq:transfer}
 \bar x^{\mathsf T}\rhoone(Zu)\bar x\ge c_{s,c}.
\end{equation}
\end{theorem}
\begin{proof}
By \cref{eq:down-s} and $u^{\mathsf T}J_ju\ge-\|J_j\|\ge-b_{j,\ell}$,
\[
 \|\partial_\ell^{[s]}u\|^2
 \ge c\,b_{s,\ell}
 -\sum_{j=0}^{s-1}\binom{\ell-j}{s-j}b_{j,\ell}
 \ \ge\ \frac c2\,b_{s,\ell}
 \qquad(n\ge n_0(s,c)),
\]
because the exact identity
$\binom{\ell-j}{s-j}b_{j,\ell}/b_{s,\ell}
=\binom sj\binom{n-\ell}{j}/\binom{n-\ell}{s}$ makes each lower-order
ratio $O_s((n-\ell)^{j-s})$, negligible for $\ell\le n/3$.  On the
other hand, \cref{eq:iterated-down,eq:down-norm} give
\[
 \|\partial_\ell^{[s]}u\|^2
 \le\frac1{(s!)^2}
 \prod_{h=1}^{s-1}(\ell-h)(n-\ell+h+1)\,
 \|\partial_\ell u\|^2,
\]
and, writing $(a)_s:=a(a-1)\cdots(a-s+1)$, we have
$(s!)^2b_{s,\ell}=(\ell)_s(n-\ell)_s$.  Combining the two displays
yields
$\|\partial_\ell u\|^2
\ge\tfrac c2\,\ell(n-\ell)\prod_{h=1}^{s-1}
\tfrac{n-\ell-h}{n-\ell+h+1}
\ge c'_{s,c}\,\ell(n-\ell)$
for large $n$.  Now apply \cref{eq:rdm-identity} and use
$(n-\ell)/n\ge2/3$.
\end{proof}

\begin{lemma}\label{lem:rounding}
Let $\sigma\succeq0$ have trace one, put $\bar x:=\xstar/\sqrt n$, and
assume $\bar x^{\mathsf T}\sigma\bar x\ge c$, $0<c\le1$.  Let $P$
be the spectral projector of $\sigma$ onto eigenvalues at least $c/4$.
Then
\begin{equation}\label{eq:rounding-projector}
 \rank P\le\frac4c,
 \qquad
 \|P\bar x\|^2\ge\frac c2.
\end{equation}
From any orthonormal basis of $\Ran P$, randomized clipped rounding of
the basis vectors and their negatives, repeated $O_c(\log n)$ times,
produces a list of $O_c(\log n)$ Boolean vectors that, with probability
$1-o(1)$, contains one with
$|\ip z{\xstar}|\ge c'_cn$ for a constant $c'_c>0$.
\end{lemma}
\begin{proof}
The rank bound follows from $\Tr\sigma=1$; and
$c\le\|P\bar x\|^2+c/4$ because the form is at most $c/4$ on
$\Ran(I-P)$ and at most one on $\Ran P$.  Since
$\sum_j\ip{w_j}{\bar x}^2\ge c/2$ over an orthonormal basis
$w_1,\ldots,w_R$ of $\Ran P$ with $R\le4/c$, some $w$ has
$\ip w{\bar x}\ge\gamma_0:=\sqrt{c^2/8}$ after a sign flip.  Clip its
coordinates at $K/\sqrt n$ with $K:=2/\gamma_0$; since at most $n/K^2$
coordinates exceed the clip, of total absolute mass at most
$\sqrt n/K$ by Cauchy--Schwarz, the clipped vector $w'$ still has
$\ip{w'}{\xstar}\ge\tfrac12\gamma_0\sqrt n$.  Rounding coordinate $i$
independently to a sign of mean $w'_i\sqrt n/K$ gives expected overlap
at least $\gamma_0^2n/4$, hence overlap $\ge\gamma_0^2n/8$ with
constant probability by Hoeffding.  Repeat over the $R$ basis vectors,
their negatives, and $O_c(\log n)$ independent rounds.
\end{proof}

\begin{lemma}\label{lem:validation}
Let $\cL$ be a random list of at most $C\log n$ Boolean vectors,
independent of a validation pool that consists either of exactly
$m_{\rm v}$ i.i.d.\ clauses or of a Poisson number $M_{\rm v}$ of clauses
of mean $m_{\rm v}$, where
$m_{\rm v}\ge C_{k,c}\rho^{-2}\log n$.  Suppose that with probability
$1-o(1)$ the list contains $z$ with $|q_z|\ge c$, where
$q_z:=n^{-1}\ip z{\xstar}$.  Define $T(z)$ as the empirical advantage
of $z$ on the realized validation pool: the average of $y\,z^{F}$ over
its clauses $(F,y)$, with $T:=0$ if the pool is empty.  Then,
with probability $1-o(1)$, the maximizer of $|T|$ over $\cL$ satisfies
$|q_z|\ge c'_{k,c}>0$.  For odd $k$, the maximizer of $T$ over
$\cL\cup(-\cL)$ additionally satisfies $q_z\ge c'_{k,c}$, fixing the
global sign.
\end{lemma}
\begin{proof}
For fixed $z$, sampling a uniform $k$-subset without replacement gives
\begin{equation}\label{eq:validation-mean}
 \E T(z)=\rho\,q_z^k+O_k(\rho/n),
\end{equation}
since the product of $k$ coordinates sampled without replacement from a
$\pm1$ sequence of mean $q_z$ has mean $q_z^k+O_k(1/n)$.  In the Poisson
case, condition on $M_{\rm v}\ge m_{\rm v}/2$, which fails with
probability $e^{-cm_{\rm v}}$.  Hoeffding's inequality at deviation
$t:=\rho c^k/8$, with $C_{k,c}$ large, gives $|T(z)-\E T(z)|\le t$
simultaneously for all list elements, and the $O_k(\rho/n)$ term is at
most $t$ for large $n$.  The good candidate then has
$|T|\ge\rho c^k-2t$, so the maximizer $z^*$ satisfies
$\rho|q_{z^*}|^k\ge\rho c^k-4t=\rho c^k/2$, i.e.\
$|q_{z^*}|\ge c\,2^{-1/k}$.  For odd $k$, $T(-z)=-T(z)$ while
$\E T(z)$ has the sign of $q_z$, so maximizing the signed score over
$\cL\cup(-\cL)$ fixes the orientation.
\end{proof}

\section{Sharp trace bounds for normalized Kikuchi operators}
\label{sec:trace}

This section proves the log-free null estimate in its two natural
forms, from which every spectral upper bound in the paper follows.  The
\emph{single-label} theorem (\cref{thm:single-trace}) is the model
case (one family of signed labels, one degree-plus-floor
normalizer) and covers even arity directly.  The \emph{symmetric
pair} theorem (\cref{thm:packet-trace}) runs the same opening/closing
accounting for moves that toggle two disjoint residual labels of a
single signed family at once, with a retained numerator that may shrink
under insertion and a monotone completion in the denominator; it covers
odd inference (\cref{sec:odd}), whose operator the quantum algorithm
reuses unchanged (\cref{sec:quantum}), and odd refutation
(\cref{sec:refutation}).  Its hypotheses
\ref{P-monotone}--\ref{P-mass} are short enough that each application
verifies them in a few lines.

\subsection{Walks, excess, and the active frontier}

Fix a residual size $2r$.  A \emph{trace word} of length $2t$ consists of
rows $S_0,S_1,\ldots,S_{2t}=S_0$ in $\Omega_\ell$ together with, at each
step, either one residual label of size $2r$, or an unordered pair of
disjoint residual labels of size $2r$ each, whose partial matching carries
$S_{i-1}$ to $S_i$.  A residual is \emph{legal} at $S$ if exactly $r$ of
its vertices lie in $S$.  Labels are physical: coincident supports remain
distinct labels.  Let $c_i$ be the number of distinct residual labels that
have appeared by step $i$, and define the \emph{excess}
\begin{equation}\label{eq:excess}
 \chi_i:=\ell+rc_i-
 \Bigl|S_0\cup\bigcup_{a\text{ used by step }i}C_a\Bigr|,
\end{equation}
which measures how many vertices the walk has revisited rather than
explored.  Because every step toggles its residual labels,
\begin{equation}\label{eq:parity-identity}
 S_i\sd S_0=\bigsd_{a\in\cO_i}C_a,
\end{equation}
where $\cO_i$ is the set of residual labels used an odd number of times by
step $i$; we call these labels \emph{active}.  At the first use of a
residual, exactly $r$ of its vertices lie outside the current row.  If
$h$ of those exterior vertices were exposed earlier, that first use
increases the excess by exactly $h$, and hence
\begin{equation}\label{eq:excess-sum}
 \chi_{2t}=\sum_{\text{first uses }a}h_a.
\end{equation}
These identities depend only on the residual labels and remain valid when a
step toggles two disjoint residuals at once.

The excess controls how many active labels can ever be legal at the
current row, the quantity that governs the cost of closings.

\begin{lemma}[Active legal labels]\label{lem:active}
At every prefix, with $S$ the current row and $\chi$ the current excess,
\begin{equation}\label{eq:active}
 \#\{a\in\cO:|C_a\cap S|=r\}\le\frac{\ell+2\chi}{r}.
\end{equation}
\end{lemma}
\begin{proof}
Let $U:=\bigcup_{a\in\cO}C_a$ and $U^+:=U\setminus S_0$.  At the first
use of each active label, at least $r-h_a$ exterior vertices are new;
these vertices are distinct across first uses and lie in $U^+$, so
$|U^+|\ge r|\cO|-\chi$.  Put
$b:=|S\setminus S_0|=|S_0\setminus S|$.  By
\cref{eq:parity-identity}, $U^+\cap S=S\setminus S_0$, so this
intersection has size $b$.  Every vertex of $U^+\setminus S$ lies in a
positive even number of active residuals, hence in at least two, and every
vertex of $S_0\setminus S$ lies in at least one.  Since the total number
of active residual incidences is $2r|\cO|$, the number lying inside $S$ is
at most
\[
 2r|\cO|-2(|U^+|-b)-b\le2\chi+b\le2\chi+\ell.
\]
Every active label legal at $S$ contributes exactly $r$ such incidences.
\end{proof}

\subsection{Chronologies}

A \emph{chronology} records only the combinatorial skeleton of a walk:
the opening/closing pattern, which labels at a step are first uses, the
allocation of the final excess among those first uses, and the recycled
exterior positions.  Old-label identities are deliberately not encoded;
they are summed later through a degree or port-capacity bound.  The
point of the next lemma is that chronologies are cheap: exponential in
$t$, with no factor depending on $n$ or on the sample size.

\begin{lemma}\label{lem:chronology}
Fix $r\ge1$, $t\ge1$, and $\chi\ge0$.
\begin{enumerate}
\item For single-label words of length $2t$ with exactly $t$ openings,
the chronologies with final excess $\chi$ number at most $C_r^t2^\chi$.
\item For pair words of length $2t$, in which each step uses an
unordered pair of disjoint residual labels and every label has even
total multiplicity, the chronologies with final excess $\chi$ number at
most $C_r^t2^\chi$.
\end{enumerate}
\end{lemma}
\begin{proof}
For (1): each of the $2t$ steps is an opening or a closing, so the
opening pattern can be chosen in at most $2^{2t}=4^t$ ways, and
recording at each opening whether the label is a first use or a
previously seen inactive one costs a further $2^t$; both factors are
absorbed into $C_r^t$.  There are at most $t$ first uses.  Distributing
the excess $\chi$ among them as ordered nonnegative summands is choosing
a composition, of which there are at most
$\binom{\chi+t-1}{t-1}\le2^{\chi+t}$; and at each first use, the set of
recycled exterior positions is a subset of its $r$ exterior slots, which
costs at most $2^r$ per first use, hence $2^{rt}$ in total.
For (2): a chronology records, at each step, how many of the two
labels are active and the first-use/old
status of each inactive label, which costs $C^t$.  A word has at most
$2t$ distinct labels, hence at most $2t$ first uses, so the excess
allocation costs at most $2^{\chi+2t}$ and the recycled exterior
positions at most $2^{2r\cdot2t}$, by the same counts as in (1).
\end{proof}

\subsection{The single-label trace theorem}

\begin{theorem}[Normalized single-label trace]\label{thm:single-trace}
Fix $r\ge1$ and $\delta>0$, and let $\ell$ satisfy \cref{eq:range}.  Let
$\{C_a\}$ be either a Poisson family or a deterministic number of
independent uniform residual supports in $\binom{[n]}{2r}$.  Let
$\sigma_a$ be independent Rademacher signs and let $u_a\in[-1,1]$ be
coefficients independent of the signs and, conditional on any auxiliary
marks, independent of the supports.  Define
\begin{equation}\label{eq:single-operator}
 \cS:=\sum_a u_a\sigma_aA_{C_a},
 \qquad
 G:=\sum_a|u_a|P_{C_a},
 \qquad
 \Gamma:=G+dI,
\end{equation}
where $d>0$ is deterministic.  Assume that, after every fixed branch
prefix of a $2t$-step trace and for every $h\ge0$, the conditional expected
weighted mass of labels that are new on that branch, legal at its current
row, and recycle exactly $h$ previously exposed exterior vertices is at
most
\begin{equation}\label{eq:single-mass}
 d\left(\frac{C_rt}{n}\right)^h.
\end{equation}
Then, for $t=\lfloor c\ell\log n\rfloor$ with
$c=c(r,\delta)>0$ sufficiently small,
\begin{equation}\label{eq:single-moment}
 \E\Tr(\Gamma^{-1/2}\cS\Gamma^{-1/2})^{2t}
 \le N\left(\frac{C_r\ell}{rd+\ell}\right)^t
 \sum_{\chi\ge0}
 \left(\frac{C_rt}{n}\right)^\chi
 \left(1+\frac{2\chi}{\ell}\right)^t,
\end{equation}
and, for every fixed $B>0$,
\begin{equation}\label{eq:single-tail}
 \bbP\left\{
 \|\Gamma^{-1/2}\cS\Gamma^{-1/2}\|
 >C_{r,\delta,B}\sqrt{\frac{\ell}{rd+\ell}}
 \right\}\le n^{-B\ell}.
\end{equation}
The same conclusion holds for a fixed finite superposition of independent
families, with their masses and diagonal degrees added.
\end{theorem}

\begin{proof}
Cyclic cancellation of the diagonal square roots expands the trace into
closed labeled row words, with one inverse diagonal attached to each
departure row.  Averaging the signs leaves only words in which every label
has even multiplicity.  Such a word has exactly $t$ openings and $t$
closings, according as the label's previous multiplicity is even or odd.
After taking absolute values, its coefficient is bounded by the product of
the $|u_a|$'s.

The filtration point is branchwise.  Fix one chronology, one initial row,
and one concrete branch prefix $P$.  Let $G_P$ be the degree contributed by
the labels already exposed on that branch.  Every denominator in the
completed instance is at least the corresponding chronological denominator
$d+G_P(S,S)$, because unexposed labels add nonnegative degree.  We replace
each final inverse denominator by this chronological upper bound at the
moment the departure row is charged.  The weight accumulated before the
next extension is then measurable with respect to the data on this one
branch.  We sum the conditional extension bounds branch by branch and only
then sum over branches; no sigma-field is shared by incompatible prefixes.

Consider an opening from a branch ending at $S$.  Let
$M_S:=G_P(S,S)$.  Previously exposed inactive labels that are legal at
$S$ have total coefficient mass at most $M_S$.  If the opening has zero
excess, a new label has conditional mass at most $d$ by
\cref{eq:single-mass}.  Hence old and new zero-excess openings together
have charge
\[
 \frac{M_S+d}{d+M_S}=1.
\]
If the opening has excess $h\ge1$, it is a first use; use the baseline
$d$ and \cref{eq:single-mass} to obtain charge
$C_r(C_rt/n)^h$.

At a closing after step $i$, let $X_S$ be the coefficient mass of active
labels legal at $S$.  By \cref{lem:active} and $|u_a|\le1$,
$X_S\le(\ell+2\chi_i)/r\le(\ell+2\chi)/r$, because excess is
nondecreasing along the branch, while these same labels contribute $X_S$
to the chronological degree.  Thus the total closing charge is at most
\[
 \frac{X_S}{d+X_S}
 \le\frac{\ell+2\chi}{rd+\ell+2\chi}
 \le \frac{\ell}{rd+\ell}
      \left(1+\frac{2\chi}{\ell}\right),
\]
and there are exactly $t$ closings, so
multiplying the branch charges and using \cref{eq:excess-sum} gives the
claimed power of $\ell/(rd+\ell)$ and the excess factor.  The factor
$2^\chi$ from \cref{lem:chronology}(1) is absorbed by increasing the
constant in $(C_rt/n)^\chi$.  Summing the chronologies and initial rows
gives
\[
 N\left(\frac{C_r\ell}{rd+\ell}\right)^t
 \sum_{\chi\ge0}
 \left(\frac{C_rt}{n}\right)^\chi
 \left(1+\frac{2\chi}{\ell}\right)^t,
\]
which is \cref{eq:single-moment}.  Finally,
$(C_rt/n)e^{2t/\ell}\le C_rn^{-\delta+2c}$ by \cref{eq:range}; choose
$c<\delta/4$, sum the geometric series, and apply Markov's inequality
using $N\le n^\ell$.
\end{proof}

\begin{remark}[Where the logarithm went]\label{rem:sharp-openings}
The identity $M_S/(d+M_S)+d/(d+M_S)=1$ is the entire log-removal
mechanism: exposed openings are charged to the degree they created at
the actual departure row, fresh ones to the mean floor, and no maximum
over the $N=\binom n\ell$ row degrees is ever taken.  The branchwise
chronological denominator makes this rigorous: it is fixed before the
conditional mass of a new label is estimated.
\end{remark}

\subsection{The symmetric pair trace theorem}

At odd arity, one move toggles two disjoint residual labels from a
\emph{single} family, and every label carries its own independent
sign, so averaging forces both labels of every surviving step to even
multiplicity: the opening/closing accounting runs symmetrically in the
two labels of a move.  As anticipated in \cref{sec:tech-branch}, the
numerator must be insertion-nonincreasing and the diagonal a monotone
completion.  The theorem below states these requirements as four short
hypotheses: \ref{P-monotone} is the monotonicity just described,
\ref{P-diagonal} is row-wise domination, \ref{P-port} bounds the
amplitude one label can carry at one row, and \ref{P-mass} bounds the
deletion mass of every way a step picks up fresh labels.  For a
\emph{row edge} $e$ joining $S$ and $T$ through the pair $\{a,b\}$, we
call the incidences $(S,a)$ its \emph{ports}.  The amplitudes are
bounded by $1$ in every algorithmic construction, and by $4$ in the one
analysis-only use, the decoupled comparison of
\cref{sec:odd-comparison}; $C_0$ also absorbs the $r$-dependent
constants of the mass hypotheses.

\begin{theorem}[Symmetric pair trace]\label{thm:packet-trace}
Fix $r\ge1$, $\delta>0$, and $C_0\ge1$, and let $\ell$ satisfy
\cref{eq:range}.  Let $\Xi$ be a single family of marked labels, each
label $a$ carrying a residual $C_a\in\binom{[n]}{2r}$: either a marked
Poisson process of finite intensity, or a deterministic number of
i.i.d.\ marked labels.  An unordered pair $p=\{a,b\}$ of distinct labels
with disjoint residuals carries the symmetric partial matching
$B_{C_a,C_b}$; write $\mathcal E(p)$ for its undirected row edges, and
for $e=(p,\{S,T\})\in\mathcal E(p)$ put $A_e:=E_{S,T}+E_{T,S}$, so that
$S\sd T=C_a\sd C_b$.  Given the full unsigned marked configuration
$\xi$, let $\omega_e(\xi)\in[0,C_0]$ be the retained amplitude of the
row edge $e$, and let $G_\xi\succeq0$ be diagonal.  Assume:
\begin{enumerate}[label=\textup{(P\arabic*)}]
\item\label{P-monotone} For every row edge already exposed on the
current branch, $\omega_e(\xi)$ is insertion-nonincreasing, and $G_\xi$
is insertion-nondecreasing.
\item\label{P-diagonal} Pointwise, for every row $S$,
\begin{equation}\label{eq:packet-diagonal}
 G_\xi(S,S)\ge
 \sum_{e:S\in e}\omega_e(\xi).
\end{equation}
\item\label{P-port} For every row $S$ and physical label $a$,
\begin{equation}\label{eq:packet-port}
 \sum_{e:S\in e,\,a\in p(e)}\omega_e(\xi)\le C_0.
\end{equation}
\item\label{P-mass} After any fixed chronological branch prefix of a
$2t$-step trace, the deletion masses obtained after dropping the
insertion-nonincreasing retention factors of the fresh edges satisfy,
for total first-use excess $h$,
\begin{align}
 \text{two new labels:}&\quad C_0d(C_0t/n)^h,\label{eq:packet-mass-two}\\
 \text{exactly one new label, the partner among at most $4t$ exposed:}
 &\quad C_0t\beta(C_0t/n)^h,\label{eq:packet-mass-one}\\
 \text{a new partner of one fixed exposed label:}
 &\quad C_0\beta(C_0t/n)^h,\label{eq:packet-mass-partner}
\end{align}
where $d>0$ is a deterministic floor and $0<\beta\le C_0$ satisfies
$\beta^2\le C_0d/n$ and $t\beta\le c_0d$ for a sufficiently small
$c_0=c_0(r,C_0)$.  In the Poisson case these are the deletion masses
of \cref{eq:poisson-insertion}; in the
fixed-sample case they are those of
\cref{eq:fixed-deletion}, with the floor $d$ held fixed throughout.
\end{enumerate}
Let the physical labels carry signs $\sigma_a$ that, conditional on the
unsigned configuration $\xi$, are independent Rademacher variables, and
define
\begin{equation}\label{eq:packet-operator}
 \cS_\xi:=\sum_p\sum_{e\in\mathcal E(p)}
 \omega_e(\xi)\,\sigma_{a(p)}\sigma_{b(p)}A_e,
 \qquad
 \cR_\xi:=dI+G_\xi.
\end{equation}
Then, for $t=\lfloor c\ell\log n\rfloor$ with
$c=c(r,\delta,C_0)>0$ sufficiently small,
\begin{equation}\label{eq:packet-moment}
 \E\Tr(\cR_\xi^{-1/2}\cS_\xi\cR_\xi^{-1/2})^{2t}
 \le N\left(\frac{C_{r,C_0}\ell}{d}\right)^t
 \sum_{\chi\ge0}
 \left(\frac{C_{r,C_0}t}{n}\right)^{\chi/2}
 \left(1+\frac{2\chi}{\ell}\right)^t,
\end{equation}
and, for every fixed $B>0$,
\begin{equation}\label{eq:packet-tail}
 \bbP\left\{
 \|\cR_\xi^{-1/2}\cS_\xi\cR_\xi^{-1/2}\|
 >C_{r,\delta,B,C_0}\sqrt{\frac\ell d}
 \right\}\le n^{-B\ell}.
\end{equation}
\end{theorem}

\begin{proof}
Cyclic cancellation of the diagonal square roots expands the trace into
closed labeled row words of length $2t$, each step carrying one
unordered pair of disjoint residual labels, with one inverse diagonal
attached to each departure row and the product of the retained
amplitudes as coefficient.  Conditional on $\xi$ the signs are
independent, so averaging them leaves only words in which \emph{every}
physical label has even total multiplicity.  Call a label \emph{active}
when it has been used an odd number of times; an appearance of an
inactive label is an \emph{opening} and an appearance of an active
label a \emph{closing}.  Every label ends inactive, so its appearances
alternate openings and closings, and exactly half of the $4t$ label
appearances of the word are closings.  Writing $j_i\in\{0,1,2\}$ for
the number of active labels in the pair used at step $i$, this reads
\begin{equation}\label{eq:pair-parity}
 \sum_{i=1}^{2t}j_i=2t.
\end{equation}
Track the excess \cref{eq:excess} over both residual labels of every
step; the identities \cref{eq:parity-identity,eq:excess-sum} hold
verbatim.  A word has at most $2t$ distinct labels, each first use
recycles at most $2r$ exposed vertices, so $\chi\le4rt$, and by
\cref{lem:chronology}(2) a word of final excess $\chi$ has at most
$C_r^t2^\chi$ chronologies.

The deferred-decision step is branchwise.  Fix one chronology and one
prefix $P$, keep all labels first used on $P$ as explicit atoms, and
regard the unexposed atoms as an independent background.  Let $\Phi_P$
be the product of the edge amplitudes and inverse departure-row
diagonals already accumulated on this prefix, evaluated in the explicit
configuration plus the background; by \ref{P-monotone}, $\Phi_P$ is
insertion-nonincreasing.  Define the expected weight of all completions
of $P$ recursively, backwards from the leaves, and apply
\cref{lem:branch-recursion} (in the fixed-sample case, its form from
\cref{rem:fixed-recursion}).  For an extension using only exposed
labels, sum its amplitude pointwise through
\cref{eq:packet-diagonal,eq:packet-port}.  For an extension introducing
one or two fresh labels, the deletion step moves the selected atoms to
the explicit list, insertion monotonicity deletes them from the old
prefix factors, and \ref{P-mass} bounds the fresh integral.  The
deterministic backward bounds of \cref{eq:branch-recursion} are exactly
the following step charges, in which $S$ denotes the departure row,
$\chi$ the final excess, $h$ the excess created at the step, and
\[
 A_\chi:=\frac{\ell+2\chi}{r},
 \qquad
 \zeta:=\frac{C_{r,C_0}t}{n}.
\]

\emph{Steps with $j_i=0$.}  If both labels are exposed and inactive,
the available row edges at $S$ have total amplitude at most
$G_\xi(S,S)$ by \ref{P-diagonal}, so the charge is at most
$G_\xi(S,S)/(d+G_\xi(S,S))\le1$.  If the step uses fresh labels, the
two-new and one-new masses of \ref{P-mass} total at most
$C_0d(C_0t/n)^h+C_0t\beta(C_0t/n)^h\le2C_0d(C_0t/n)^h$, using
$t\beta\le c_0d$; dividing by the floor $d$, the charge is at most
$\zeta^h\le\zeta^{h/2}$ after adjusting constants, since
$t/n\le n^{-\delta}$ in the range \cref{eq:range}.

\emph{Steps with $j_i=1$.}  First choose the active label of the pair:
it is legal at $S$, so \cref{lem:active} allows at most $A_\chi$
choices, the excess being nondecreasing along the branch.  If the
partner is exposed, \ref{P-port} bounds the available amplitude by
$C_0$ per choice, while \ref{P-diagonal} bounds the total by
$G_\xi(S,S)$; hence the charge is at most
\[
 \min\left\{\frac{G_\xi(S,S)}{d+G_\xi(S,S)},\,
 \frac{C_0A_\chi}{d}\right\}
 \le\min\left\{1,\frac{C_0A_\chi}{d}\right\}
 \le C_0\sqrt{\frac{A_\chi}{d}}.
\]
If the partner is fresh with excess $h$,
\cref{eq:packet-mass-partner} gives at most
$C_0A_\chi\beta(C_0t/n)^h/d$.  Since $\chi\le4rt$ gives
$A_\chi\le C_rt$, and $\beta^2\le C_0d/n$,
\[
 \frac{A_\chi\beta}{d}
 =\sqrt{\frac{A_\chi}{d}}\cdot
 \sqrt{\frac{A_\chi\beta^2}{d}}
 \le\sqrt{\frac{A_\chi}{d}}\cdot
 \sqrt{\frac{C_rC_0t}{n}},
\]
so this charge is at most
$\sqrt{A_\chi/d}\,\zeta^{h+1/2}\le\sqrt{A_\chi/d}\,\zeta^{h/2}$ after
adjusting constants.

\emph{Steps with $j_i=2$.}  Choose one of the two active labels, at
most $2A_\chi$ choices by \cref{lem:active}, and apply
\ref{P-port}: the charge is at most
$\min\{1,2C_0A_\chi/d\}\le2C_0A_\chi/d$.

Multiply the charges along the word.  By \cref{eq:pair-parity}, the
total exponent of $A_\chi/d$ is
\[
 \tfrac12\,\#\{i:j_i=1\}+\#\{i:j_i=2\}
 =\tfrac12\sum_{i=1}^{2t}j_i=t,
\]
and by \cref{eq:excess-sum} the fresh factors multiply to
$\zeta^{\chi/2}$; the $j_i=0$ charges and all remaining constants
contribute $C_{r,C_0}^t$.  Since
$A_\chi=(\ell/r)(1+2\chi/\ell)$, one chronology therefore contributes
at most $(C_{r,C_0}\ell/d)^t(1+2\chi/\ell)^t\zeta^{\chi/2}$.  Apply the
chronology count, absorb $2^\chi$ into the constant of
$\zeta^{\chi/2}$, and sum over the $N$ initial rows to obtain
\cref{eq:packet-moment}.  Finally,
$(C_{r,C_0}t/n)^{1/2}e^{2t/\ell}\le C_{r,C_0}n^{-\delta/2+2c}$ by
\cref{eq:range}; choose $c<\delta/4$, sum the geometric series, and
apply Markov's inequality using $N\le n^\ell$.
\end{proof}

\section{Even arity: one normalized matrix}\label{sec:even}

Throughout this section $k=2r$ is even; a single degree-plus-floor
normalized matrix performs detection and recovery (and refutation, in
\cref{sec:refutation}).  For the inference statements, the algorithm
first applies \cref{lem:subsample} with
$\mu:=\tfrac12C_{k,\delta}\rho^{-2}\Dkl\le m/2$ and splits the
retained clauses into a \emph{spectral pool} and a \emph{validation
pool}, each Poisson of mean $m_{\rm s}:=\mu/2$; the fixed-$m$ theorems
are assembled in \cref{sec:exact}.  The deterministic certificate
(\cref{prop:even-certificate}) uses the full input, no subsampling, and
no randomness at all.

Given a clause family, define on the slice $\Omega_\ell$ the signed
adjacency matrix of its Kikuchi graph, its degree matrix, and their
normalization:
\begin{equation}\label{eq:even-operator}
 A:=\sum_{a}y_aA_{F_a},
 \qquad
 D:=\sum_{a}P_{F_a},
 \qquad
 \Gamma:=D+\bd I,
 \qquad
 K:=\Gamma^{-1/2}A\,\Gamma^{-1/2},
\end{equation}
where
\begin{equation}\label{eq:even-mean-degree}
 \bd:=\bar m\,\frac{b_{r,\ell}}{\binom n{2r}}
 =\Theta_r\!\left(\bar m\,\frac{\ell^r}{n^r}\right)
\end{equation}
is a deterministic scalar and $\bar m$ is the mean clause count of the
family: $\bar m=m_{\rm s}$ for the spectral pool, and $\bar m=m$ when the
matrix is built from the full fixed-$m$ input.  The sample condition
\cref{eq:inference-threshold} is equivalent, up to $r$-dependent
constants, to $\rho^2\bd\ge C\ell$.  For every realization and every
signing,
\begin{equation}\label{eq:even-contraction}
 \|K\|\le1,
\end{equation}
because \cref{eq:domination} gives
$|v^{\mathsf T}Av|\le v^{\mathsf T}Dv\le v^{\mathsf T}\Gamma v$.

\subsection{Null norm and detection}

\begin{proposition}[Even null norm]\label{prop:even-null}
Let the clause family be either a Poisson pool of mean $\bar m$ or
exactly $\bar m$ i.i.d.\ clauses; in particular, the full fixed-$m$ input
is allowed.  For every fixed $B>0$, under the null,
\begin{equation}\label{eq:even-null}
 \bbP\left\{
 \|K\|>C_{r,\delta,B}\sqrt{\frac{\ell}{r\bd+\ell}}
 \right\}\le n^{-B\ell}.
\end{equation}
\end{proposition}
\begin{proof}
Apply \cref{thm:single-trace} with $u_a=1$ and $d=\bd$; only the mass
hypothesis \cref{eq:single-mass} needs checking.  In either model,
conditional on any exposure prefix the unexposed supports remain
i.i.d.\ uniform, so the zero-excess mass of new labels legal at a row
is at most $\bar m\,b_{r,\ell}/\binom n{2r}=\bd$ (for the fixed count,
$(\bar m-\#\text{exposed})\,b_{r,\ell}/\binom n{2r}\le\bd$).
Prescribing $h$ recycled exterior vertices confines them to the exposed
set, of size $O_r(t)$: at most $\binom\ell r$ interior choices,
$\binom{C_rt}h$ recycled choices, and free choices of the remaining
exterior vertices give mass at most $\bd\,(C_rt/n)^h$.
\end{proof}

\begin{lemma}\label{lem:even-witness}
Under the planted law, for the spectral pool, with probability $1-o(1)$,
\begin{equation}\label{eq:even-witness}
 \lambda_{\max}(K)\ \ge\ \frac\rho2-o(\rho),
\end{equation}
uniformly over the levels in \cref{eq:range}.
\end{lemma}
\begin{proof}
The gauge commutes with $\Gamma$ and $ZAZ=\sum_a\eta_aA_{F_a}$.
Testing the gauged matrix with $\Gamma^{1/2}\1$ and using
\cref{eq:support-count} gives the Rayleigh quotient
$\sum_a\eta_a/(M_{\rm s}+m_{\rm s})$, with $M_{\rm s}$ the realized
pool size.  At the threshold
$\rho^2m_{\rm s}\ge c\,\Dkl\ge cn\to\infty$, so by Chebyshev
$\sum_a\eta_a=(1+o(1))\rho m_{\rm s}$ and
$M_{\rm s}=(1+o(1))m_{\rm s}$ with probability $1-o(1)$.
\end{proof}

\begin{proposition}[Even detection]\label{prop:even-detection}
If $m\ge C_{k,\delta}\rho^{-2}\Dkl$ with $C_{k,\delta}$ sufficiently
large, then outputting \textsf{planted} if and only if a
\cref{lem:krylov} run on $K$, at accuracy $\gamma=\rho/12$, returns a
Rayleigh quotient at least $\rho/3$, achieves strong detection.
\end{proposition}
\begin{proof}
Increase $C_{k,\delta}$ so that the bound of \cref{prop:even-null} is
below $\rho/4$, which is possible since the null bound is at most
$C_{r,\delta}\rho/\sqrt{C_{k,\delta}}$ at the threshold.  Under the
null every Rayleigh quotient is then below $\rho/3$, while under
planting \cref{lem:even-witness,lem:krylov} give a quotient at least
$\rho/2-o(\rho)-\rho/12\ge\rho/3$; all failure events are $o(1)$.
\end{proof}

\subsection{Planted comparison and weak recovery}

Put $\tau=m_{\rm s}/\binom n{2r}$ for the spectral pool, so that
\begin{equation}\label{eq:even-population}
 \E[ZAZ]=\rho\tau J_r,
 \qquad
 \bd=\tau\,b_{r,\ell}.
\end{equation}
The next proposition shows that, after gauging and normalizing, the
planted matrix is uniformly close to this population mean: symmetrizing
with an independent copy turns the centered matrix into a single-label
signed superposition, to which the trace theorem applies directly.

\begin{proposition}[Even planted comparison]\label{prop:even-comparison}
Assume $\bd\ge C_{r,\delta}\ell$.  With
probability $1-o(1)$,
\begin{equation}\label{eq:even-comparison}
 \left\|\Gamma^{-1/2}\bigl(ZAZ-\rho\tau J_r\bigr)\Gamma^{-1/2}\right\|
 \ \le\ C_{r,\delta}\sqrt{\frac\ell\bd}.
\end{equation}
\end{proposition}
\begin{proof}
Let $(A',D')$ be an independent planted Poisson copy of the spectral
pool, with sample count $M'$.  Superpose the two point processes.
Conditional on the superposed supports and noise marks, the assignment of
each point to one of the two copies is an independent fair sign
$\varsigma_j$, so
\[
 Z(A-A')Z\ \overset{d}{=}\ \sum_j\varsigma_j\eta_jA_{F_j},
\]
a single-label signed superposition in which every sign is attached to
one fixed support and repeated supports remain distinct labels.  With
$R:=D+D'+2\bd I$, \cref{thm:single-trace} applied to the superposition
(coefficients $u_j=\eta_j$, which are independent of the signs
$\varsigma_j$) gives
\[
 \bbP\bigl\{\|R^{-1/2}Z(A-A')ZR^{-1/2}\|
 >C_{r,\delta}\sqrt{\ell/\bd}\bigr\}\le n^{-B_0\ell},
\qquad B_0:=10k.
\]
Moreover $\E'R=D+3\bd I\preceq3\Gamma$; and since
$m_{\rm s}\le n^{3k}$ at fixed $\rho$, on
$\cH:=\{M'\le n^{4k}\}$ one has
$D'\preceq M'I\preceq n^{4k}\Gamma$ (as $\Gamma\succeq\bd I$,
$\bd\ge1$), while the Poisson tail gives
$\E'[\1_{\cH^c}R]\preceq n^{-B_0\ell}\Gamma$.  Finally
$ZAZ-\rho\tau J_r=\E'[Z(A-A')Z]$; apply \cref{lem:copy-transfer} with
$C=3$, $M=n^{4k}+3$, $\zeta=n^{-B_0\ell}$, and $p_0=n^{-B_0\ell}$, the
term $M\sqrt{p_0}\le n^{4k-5k\ell}$ being absorbed.
\end{proof}

\begin{proposition}[Even weak recovery]\label{prop:even-recovery}
At $m\ge C_{k,\delta}\rho^{-2}\Dkl$ with $C_{k,\delta}$ sufficiently
large, the top-of-spectrum computation for $K$, followed by the rounding
and validation of \cref{sec:common-recovery}, outputs a Boolean vector
$\widehat x$ with $|\ip{\widehat x}{\xstar}|\ge c_kn$ with probability
$1-o(1)$, in $n^{\ell+O_k(1)}$ time.
\end{proposition}
\begin{proof}
Let $v$ be a unit vector with $v^{\mathsf T}Kv\ge\rho/3$, supplied under
the planted law by \cref{lem:krylov} at accuracy $\rho/12$ together with
\cref{lem:even-witness}.  Put
\[
 w:=\Gamma^{-1/2}Zv,\qquad u:=w/\|w\|,
\]
so that $w^{\mathsf T}\Gamma w=1$ and
$w^{\mathsf T}(ZAZ)w=v^{\mathsf T}Kv\ge\rho/3$.  Choosing
$C_{k,\delta}$ large enough that the right side of
\cref{eq:even-comparison} is at most $\rho/12$,
\[
 \rho\tau\,w^{\mathsf T}J_rw\ \ge\ \frac\rho3-\frac\rho{12}
 \ \ge\ \frac\rho4 .
\]
Since $\bd\|w\|^2\le w^{\mathsf T}\Gamma w=1$ and $\bd=\tau b_{r,\ell}$,
dividing by $\|w\|^2$ gives
$u^{\mathsf T}J_ru\ge b_{r,\ell}/4$.  The vector available to the
algorithm is $\Gamma^{-1/2}v/\|\Gamma^{-1/2}v\|=Zu$, because the
diagonal $\Gamma$ commutes with $Z$ and $Z^2=I$; no knowledge of
$\xstar$ is needed.  The algorithm computes the one-particle density
matrix $\rhoone(Zu)$ exactly, in $n^{\ell+O_k(1)}$ time, and
\cref{thm:transfer} (with $s=r\le\ell$) shows
$\bar x^{\mathsf T}\rhoone(Zu)\bar x\ge c_{r}>0$.
\Cref{lem:rounding} produces a list of $O_k(\log n)$ Boolean vectors
containing one of constant overlap, and \cref{lem:validation} selects
one on the independent validation pool; the runtime is dominated by the
$n^{\ell+O_k(1)}$ matrix--vector and assembly costs
(\cref{sec:evaluation}).
\end{proof}

\section{Odd arity: rooted pairs for inference}\label{sec:odd}

Throughout this section $k=2r+1$ is odd, $\rho>0$ is fixed, $\ell$
satisfies \cref{eq:range}, and the spectral and validation pools are
Poisson pools of mean
$m_{\rm s}=\tfrac14C_{k,\delta}\rho^{-2}\Dkl$ each, obtained from the
input by \cref{lem:subsample}.  A single odd clause changes the parity
of a row and cannot act on one slice, so we implement the rooted pairs
of \cref{sec:tech-odd}: clauses are rooted at one variable and moved
along pairs with a common root and disjoint even residuals.

Each definition below realizes one design decision from the overview:
the \emph{mutual exact-two} rule makes each root's contribution a
partial matching and the whole operator a deterministic contraction;
\emph{dummy candidates} pad every rooted type to a common intensity,
smoothing the planted mean into a Johnson matrix and keeping the linear
part of the centered planted matrix small; and the \emph{capped
completion} charges the potential pairs of each cell, capped at one, so
that insertion can only increase it (the monotonicity \ref{P-monotone})
while each root contributes at most one to the diagonal, which is what
the coherent oracle of \cref{sec:quantum} needs---the classical and
quantum odd operators coincide.

\subsection{Candidates, the exact-two graph, and the capped completion}
\label{sec:odd-construction}

Root every clause of the spectral pool uniformly at one of its $k$
variables: $F_e=\{u(e)\}\mathbin{\dot\cup}C(e)$ with $|C(e)|=2r$.  The
rooted type space is
\[
 \cT:=\{(u,C):u\in[n],\ C\in\tbinom{[n]\setminus\{u\}}{2r}\},
 \qquad |\cT|=n\binom{n-1}{2r},
\]
and rooting a uniform clause at a uniform variable produces a uniform
rooted type, so one rooted type has intensity
\begin{equation}\label{eq:odd-pi}
 \pi:=\frac{m_{\rm s}}{|\cT|}.
\end{equation}
For $S\in\Omega_\ell$, the number of residual types at root $u$ legal at
$S$ is
\begin{equation}\label{eq:odd-cell}
 D_{u,S}=
 \begin{cases}
 \binom{\ell-1}r\binom{n-\ell}r,&u\in S,\\[1mm]
 \binom\ell r\binom{n-\ell-1}r,&u\notin S.
 \end{cases}
\end{equation}
Put
\begin{equation}\label{eq:odd-alpha}
 D_+:=\max_{u,S}D_{u,S}=\Theta_r(\ell^rn^r),
 \qquad q:=D_+^{-1},
 \qquad \alpha:=\pi D_+
 =\Theta_r\!\left(C_{k,\delta}\rho^{-2}\sqrt{\ell/n}\right).
\end{equation}
We call $\alpha$ the \emph{cell sparsity}: it is the expected number of
real candidates in one root--row cell.  Since $\rho$ is fixed and
$\ell/n=o(1)$ by \cref{eq:range}, one has $\alpha=o(1)$; in particular,
$\alpha\le1/4$ for all sufficiently large $n$.

Generate a Poisson candidate process of intensity $q$ at every rooted
type and mark each candidate real with probability $\alpha$: by
thinning, the real candidates are exactly the rooted spectral pool, and
operationally the algorithm adds an independent dummy process of
intensity $q-\pi$.  Write $\theta_e\in\{0,1\}$ for the real mark;
multiplicities at one rooted type remain distinct physical labels.

The expected total number of candidates is
$\lambda_{\rm cand}:=q|\cT|=\Theta_r(n^{r+1}/\ell^r)$.  The algorithm
aborts candidate generation if the realized count exceeds
\begin{equation}\label{eq:candidate-cap}
 C_{\rm cap}:=4\lambda_{\rm cand}+(\log n)^2,
\end{equation}
in which case detection outputs \textsf{null} and recovery outputs an
arbitrary assignment; by the Poisson Chernoff bound the abort
probability is at most $e^{-c(\lambda_{\rm cand}+\log^2n)}$,
superpolynomially small.

For a root $u$ and a row $S$, let
$\cL_u(S):=\{e:u(e)=u,\ |C(e)\cap S|=r\}$
be the set of candidates legal in the cell $(u,S)$.  For distinct
candidates $e,f$ with common root $u$ and disjoint residuals, and a row
$S$ at which both residuals are legal, put $T:=S\sd C(e)\sd C(f)$; both
residuals are then legal at $T$ as well.  Retain the undirected row edge
$g=(u;\{e,f\};\{S,T\})$ if and only if
\begin{equation}\label{eq:exact-two}
 \cL_u(S)=\cL_u(T)=\{e,f\},
\end{equation}
record it once, with any fixed canonical ordering of $\{S,T\}$, and set
$A_g:=E_{S,T}+E_{T,S}$.  We call \cref{eq:exact-two} the \emph{mutual
exact-two} rule; it is symmetric in $e$ and $f$, decided by the
candidate process alone, and ignores the marks.

Define the observed matrix, the candidate population, and the capped
completion:
\begin{align}
 \cC_{\rm o}&:=\sum_{g=(u;\{e,f\})}\theta_e\theta_f\,y_ey_f\,A_g,
 \label{eq:odd-C}\\
 \cP_{\rm o}&:=\alpha^2\sum_gA_g,
 \label{eq:odd-population}\\
 h_u(S)&:=\sum_{\substack{\{e,f\}\subseteq\cL_u(S)\\
 C(e)\cap C(f)=\varnothing}}(\theta_e+\alpha)(\theta_f+\alpha),
 \qquad
 g_u(S):=\min\{1,h_u(S)\},
 \label{eq:odd-completed}\\
 \cG_{\rm o}(S,S)&:=\sum_{u=1}^ng_u(S),
 \qquad
 d:=n\alpha^2,
 \qquad
 \Gamma_{\rm o}:=\cG_{\rm o}+dI,
 \qquad
 H_{\rm o}:=\Gamma_{\rm o}^{-1/2}\cC_{\rm o}\,\Gamma_{\rm o}^{-1/2}.
 \label{eq:odd-H}
\end{align}
The sums defining $\cC_{\rm o}$ and $\cP_{\rm o}$ run over the retained
edges; the potential $h_u(S)$ runs over all unordered disjoint
candidate pairs in the cell, retained or not.  At the inference
threshold,
\begin{equation}\label{eq:odd-d-scale}
 d=\Theta_r(C_{k,\delta}^2\rho^{-4}\ell).
\end{equation}
The capped completion in \cref{eq:odd-completed} is the key dependence
repair: inserting a candidate may destroy an old retained edge, but it
can only add terms to every potential $h_u(S)$, hence only increase
$\cG_{\rm o}$.  This is exactly hypothesis \ref{P-monotone} of
\cref{thm:packet-trace}.

\begin{lemma}[Rootwise matching and contraction]\label{lem:odd-uniqueness}
For every candidate and mark realization:
\begin{enumerate}
\item for every root $u$, the retained graph
$K_u:=\sum_{g:u(g)=u}A_g$ is a partial matching; in particular, at most
one retained edge is incident with a given cell $(u,S)$, and for every
fixed candidate $e$ the subgraph $L_e:=\sum_{g:e\in g}A_g$ is a partial
matching as well;
\item if a retained edge $g=(u;\{e,f\};\{S,T\})$ has
$\theta_e\theta_f=1$, then $g_u(S)=g_u(T)=1$; irrespective of the
marks, $g_u(S)\ge\alpha^2$ and $g_u(T)\ge\alpha^2$;
\item $\|H_{\rm o}\|\le1$ and
$\|\Gamma_{\rm o}^{-1/2}\cP_{\rm o}\,\Gamma_{\rm o}^{-1/2}\|\le1$.
\end{enumerate}
\end{lemma}
\begin{proof}
(1) If a retained edge is incident with the cell $(u,S)$, its two
labels are exactly the two elements of $\cL_u(S)$, so no second
retained edge at root $u$ meets $S$; passing to a subgraph gives the
claim for $L_e$.  (2) The pair $\{e,f\}$ lies in $\cL_u(S)$ and in
$\cL_u(T)$ with disjoint residuals, so it contributes
$(\theta_e+\alpha)(\theta_f+\alpha)$ to both potentials, at least $1$
when $\theta_e=\theta_f=1$ and at least $\alpha^2$ always.  (3) By (1)
and (2), for every $v$,
\[
 |v^{\mathsf T}\cC_{\rm o}v|
 \le\sum_g\theta_{e(g)}\theta_{f(g)}
 \bigl(v_{S(g)}^2+v_{T(g)}^2\bigr)
 \le\sum_Sv_S^2\sum_{u=1}^ng_u(S)
 =v^{\mathsf T}\cG_{\rm o}v,
\]
because at most one retained edge meets each cell and each retained
observed edge has $g_u=1$ at both endpoints.  Likewise
$|v^{\mathsf T}\cP_{\rm o}v|\le v^{\mathsf T}\cG_{\rm o}v$, since the
population weight $\alpha^2$ of a retained edge is dominated by the
cell caps at both endpoints.  Congruence by $\Gamma_{\rm o}^{-1/2}$ and
$\cG_{\rm o}\preceq\Gamma_{\rm o}$ prove (3).
\end{proof}

\subsection{Population and diagonal scale}

Define the deterministic mean population row sum
\begin{equation}\label{eq:odd-population-def}
 d_\circ:=\frac1N\,\E\,\1^{\mathsf T}\cP_{\rm o}\1 .
\end{equation}
For later use, define the common support size of two disjoint residuals
by
\begin{equation}\label{eq:odd-Delta}
 \Delta_{r,\ell}:=\Tr(P_CP_D)
 =\binom{2r}{r}^{2}\binom{n-4r}{\ell-2r}
 \qquad(C\cap D=\varnothing).
\end{equation}

\begin{lemma}\label{lem:odd-palm-degree}
There are constants $0<c_r<C_r<\infty$ such that
\begin{equation}\label{eq:odd-dcirc}
 c_rd\le d_\circ\le C_rd,
 \qquad
 \E\cP_{\rm o}=\frac{d_\circ}{b_{2r,\ell}}J_{2r}.
\end{equation}
\end{lemma}
\begin{proof}
Fix a row $S$ and Palm-pin two distinct candidates $e,f$ with a common
root $u$, disjoint residuals, and both residuals legal at $S$; let
$T:=S\sd C(e)\sd C(f)$.  The pinned pair is retained exactly when no
further candidate lands in the cells $(u,S)$ or $(u,T)$, whose mean
occupancy is at most $q(D_{u,S}+D_{u,T})\le2$; by the Poisson void
probability, retention has probability between $e^{-2}$ and $1$.  For
roots $u\notin S$, the number of ordered disjoint legal residual type
pairs is
$\binom\ell r\binom{\ell-r}r\binom{n-\ell-1}r\binom{n-\ell-1-r}r
\ge c_rD_+^2$, so the Mecke equation, the intensity $q^2$ per ordered
type pair, and the weight $\alpha^2$ give an expected row sum at least
$e^{-2}c_rn\alpha^2$; the trivial upper count gives at most
$nq^2D_+^2\alpha^2=d$.  Permutation equivariance, and the fact that
every retained edge joins rows at Johnson distance $2r$, make
$\E\cP_{\rm o}$ a scalar multiple of $J_{2r}$; its constant row sum
identifies the scalar as $d_\circ/b_{2r,\ell}$.
\end{proof}

\begin{proposition}
\label{prop:odd-population}
For every fixed $B>0$, the population estimate
\begin{equation}\label{eq:odd-pop-conc}
 \left\|\cP_{\rm o}-\frac{d_\circ}{b_{2r,\ell}}J_{2r}\right\|
 \le C_{r,B}\alpha^2(\sqrt{n\log N}+\log N)=o(d)
\end{equation}
holds with probability at least $1-2N^{-B}$.  Independently of $B$,
with probability $1-o(1)$,
\begin{equation}\label{eq:odd-raw-trace}
 \Tr\Gamma_{\rm o}\le C_rdN
\end{equation}
and
\begin{equation}\label{eq:odd-max-diagonal}
 \max_{S\in\Omega_\ell}\Gamma_{\rm o}(S,S)\le C_r(d+\log N).
\end{equation}
\end{proposition}
\begin{proof}
\emph{Concentration.}  Write $\cP_{\rm o}=\alpha^2\sum_{u=1}^nK_u$
with $K_u$ the retained matching at root $u$: the root blocks are
independent with $\|K_u\|\le1$, so matrix Bernstein~\cite{Tropp} gives
\cref{eq:odd-pop-conc}, which is $o(d)=o(n\alpha^2)$ because
$\log N\le\ell\log n=o(n)$.

\emph{Trace scale.}  Put $\Lambda:=q\binom{n-1}{2r}$,
$s:=\alpha\Lambda=m_{\rm s}/n$, and
$Z_u:=\sum_{e:u(e)=u}(\theta_e+\alpha)$, a compound Poisson variable
with $\E Z_u=2s$, $\Var Z_u\le2s$, and
$\Var(Z_u^2)\le C(s^3+s^2+s)$.  Since $g_u\le h_u$ and every disjoint
residual pair is legal on exactly $\Delta_{r,\ell}$ rows, dropping the
disjointness restriction gives
$\Tr\cG_{\rm o}\le\Delta_{r,\ell}\sum_uZ_u^2/2$.  Independence over
roots, $s\ge c_{r,\rho}>0$ in the fixed-bias level range, and Chebyshev
give $\sum_uZ_u^2\le Cns^2$ with probability $1-o(1)$; the elementary
estimates $\Delta_{r,\ell}/N\le C_r(\ell/n)^{2r}$ and
$\Lambda\le C_r(n/\ell)^r$ give $\Delta_{r,\ell}\Lambda^2\le C_rN$, so
the right side is at most $C_rn\alpha^2N=C_rdN$.

\emph{Maximum.}  Fix $S$.  The $g_u(S)$ are independent over $u$, lie
in $[0,1]$, and the second factorial moment of the cell process gives
$\E g_u(S)\le\E h_u(S)\le\tfrac12(qD_{u,S})^2(2\alpha)^2\le2\alpha^2$;
scalar Bernstein and a union bound over the $N$ rows finish.
\end{proof}

\subsection{The log-free null norm}

\begin{proposition}[Odd null norm]\label{prop:odd-null}
For every fixed $B>0$, under the null,
\begin{equation}\label{eq:odd-null}
 \bbP\left\{\|H_{\rm o}\|>
 C_{r,\delta,B}\sqrt{\frac\ell d}\right\}\le n^{-B\ell}.
\end{equation}
\end{proposition}
\begin{proof}
Apply \cref{thm:packet-trace}, with $C_0=C_0(r)$ sufficiently large, to
the candidate family, with retained amplitudes
\[
 \omega_g(\xi):=\theta_e\theta_f\,
 \1\{\cL_u(S)=\cL_u(T)=\{e,f\}\}
\]
on the row edges $g=(u;\{e,f\};\{S,T\})$ and amplitude zero on pairs
without a common root.  Under the null, the signs $y_e$ of the real
candidates are independent Rademacher conditional on the unsigned
marked configuration (dummy candidates may be given auxiliary
independent signs, on which nothing depends).  Inserting a candidate
can only destroy an old exact-two indicator and only increase every
potential $h_u(S)$, hence $\cG_{\rm o}$, which is \ref{P-monotone};
and by \cref{lem:odd-uniqueness}, at most one retained edge meets each
cell, its amplitude is dominated by the cell cap $g_u$ at both
endpoints, and a fixed candidate carries at most one retained edge per
row, which are \ref{P-diagonal} and \ref{P-port}.

It remains to verify the fresh masses \ref{P-mass}, with $d=n\alpha^2$
and $\beta:=\alpha$.  At a fixed root--row cell the total candidate
intensity is at most $qD_{u,S}\le1$, and requiring a candidate to be
real contributes a factor $\alpha$.  If a first-used residual recycles
$h$ exterior vertices, choosing those vertices from the $O_r(t)$
exposed set costs $(C_rt/n)^h$.  Consequently two new real labels with
a common root, both legal at the current row, have total mass at most
$C_rn\alpha^2(C_rt/n)^h=C_rd(C_rt/n)^h$; one new real label paired with
one of at most $4t$ exposed labels has mass at most
$C_rt\alpha(C_rt/n)^h$; and a new partner of one fixed exposed label
has mass at most $C_r\alpha(C_rt/n)^h$.  The side conditions hold:
$\beta^2=\alpha^2=d/n$, and
\begin{equation}\label{eq:odd-mixed-small}
 \frac{t\alpha}{d}=\frac{t}{n\alpha}
 \le C_r\sqrt{\frac\ell n}\log n=o(1)
\end{equation}
for $t\asymp\ell\log n$, using $\rho\le1$ and \cref{eq:range}.
\end{proof}

\subsection{A centered planted comparison}
\label{sec:odd-comparison}

Condition on the candidate process and set $Y_e:=\theta_e\eta_e$, so
that, conditional on the skeleton, the $Y_e$ are independent with mean
$\mu:=\rho\alpha$; put $X_e:=Y_e-\mu$.  The gauge cancels the common
root together with the planted residual monomials, so
$Z\cC_{\rm o}Z=\sum_{g}Y_{e(g)}Y_{f(g)}A_g$, and expanding
$Y_eY_f=(X_e+\mu)(X_f+\mu)$ gives the exact centering
\begin{equation}\label{eq:odd-centering}
 Z\cC_{\rm o}Z-\rho^2\cP_{\rm o}
 =\cQ+\mu\cL,
 \qquad
 \cQ:=\sum_gX_{e(g)}X_{f(g)}A_g,
 \qquad
 \cL:=\sum_eX_eL_e,
\end{equation}
with $L_e=\sum_{g:e\in g}A_g$ the matchings of
\cref{lem:odd-uniqueness}.  The linear part is small because the
thinning keeps $\mu$ tiny---this is where the dummy candidates earn
their keep---and the quadratic part is bounded by decoupling against an
independent copy of the marks.

\begin{proposition}[Odd planted comparison]\label{prop:odd-comparison}
Assume $m\ge C_{k,\delta}\rho^{-2}\Dkl$, with the leading constant
sufficiently large.  Then, with probability $1-o(1)$,
\begin{equation}\label{eq:odd-comparison}
 \left\|\Gamma_{\rm o}^{-1/2}
 \bigl(Z\cC_{\rm o}Z-\rho^2\cP_{\rm o}\bigr)
 \Gamma_{\rm o}^{-1/2}\right\|
 \le C_{r,\delta}\sqrt{\frac\ell d}.
\end{equation}
\end{proposition}
\begin{proof}
\emph{Linear part.}  Each retained edge belongs to two labels and each
root graph is a partial matching, so
$\sum_eL_e^2=2\sum_gA_g^2\preceq2\sum_uK_u^2\preceq2nI$.  Conditional
on the skeleton the $X_eL_e$ are independent and centered, with
$\|X_eL_e\|\le2$ and variance proxy
$\|\sum_e\E[X_e^2]L_e^2\|\le2n\alpha$, since
$\Var(Y_e)\le\E Y_e^2=\alpha$.  Matrix Bernstein gives
$\|\cL\|\le C_B(\sqrt{n\alpha\log N}+\log N)$ with probability
$1-N^{-B}$, hence, with $\mu=\rho\alpha$ and
$\Gamma_{\rm o}\succeq dI=n\alpha^2I$,
\[
 \left\|\Gamma_{\rm o}^{-1/2}(\mu\cL)\,\Gamma_{\rm o}^{-1/2}\right\|
 \le C_B\,\rho\left[\sqrt{\frac{\log N}{n\alpha}}
 +\frac{\log N}{n\alpha}\right].
\]
Relative to $\sqrt{\ell/d}=\sqrt\ell/(\sqrt n\,\alpha)$, the two terms
have squared ratios at most
$C_{k,\delta}(\sqrt{\ell/n}\,\log n+\ell\log^2n/n)=o(1)$, using
$\log N\le\ell\log n$, $\rho^2\alpha\le C_{k,\delta}\sqrt{\ell/n}$ from
\cref{eq:odd-alpha}, and \cref{eq:range}; so the linear part is
$o(\sqrt{\ell/d}\,)$, uniformly in $\rho\le1$.

\emph{Quadratic part, decoupled.}  Let $Y'$ be an independent copy of
the marks on the same skeleton and put $\Delta_e:=Y_e-Y'_e\in[-2,2]$.
For distinct labels $\E'[\Delta_e\Delta_f\mid Y]=X_eX_f$, so
$\cQ=\E'[\cQ_\Delta\mid Y]$ with
$\cQ_\Delta:=\sum_g\Delta_{e(g)}\Delta_{f(g)}A_g$.  Under the joint law
of $(Y,Y')$, conditional on the skeleton and the magnitudes
$(|\Delta_e|)_e$, the signs $\sign(\Delta_e)$ of the nonzero increments
are independent Rademacher variables, by exchangeability of each pair
$(Y_e,Y'_e)$ and independence across labels.  (The Rademacher property
holds given the magnitudes, not conditional on the realized $Y$ alone;
this is why the trace bound is proved under the joint law and only then
transferred.)  Define
\[
 h_u^\Delta(S):=\sum_{\substack{\{e,f\}\subseteq\cL_u(S)\\
 C(e)\cap C(f)=\varnothing}}|\Delta_e\Delta_f|,
 \qquad
 g_u^\Delta:=\min\{4,h_u^\Delta\},
 \qquad
 \Gamma_\Delta:=dI+\diag\Bigl(\sum_{u=1}^ng_u^\Delta(S)\Bigr),
\]
and apply \cref{thm:packet-trace}, with $C_0=C_0(r)\ge4$ sufficiently
large, signs $\sign(\Delta_e)$, and amplitudes
$\omega_g=|\Delta_{e(g)}\Delta_{f(g)}|\,
\1\{\cL_u(S)=\cL_u(T)=\{e,f\}\}\le4$: insertion destroys old
indicators and increases $\Gamma_\Delta$ (\ref{P-monotone}); a retained
edge contributes $|\Delta_e\Delta_f|\le g_u^\Delta$ at both endpoint
cells, at most one retained edge meets each cell, and a fixed label
carries at most one retained edge per row (\ref{P-diagonal},
\ref{P-port}); and the masses verified in the proof of
\cref{prop:odd-null} apply verbatim with the realness factor $\alpha$
replaced by $\E'|\Delta_e|\le2\alpha$, so \ref{P-mass} holds with
$\beta=C_r\alpha$.  Hence, under the joint law,
$\|\Gamma_\Delta^{-1/2}\cQ_\Delta\Gamma_\Delta^{-1/2}\|
\le C_{r,\delta,B}\sqrt{\ell/d}$ except with probability $n^{-B\ell}$.

\emph{Transfer.}  Conditional on $Y$, independence across labels and
$\E'|\Delta_e|\le\theta_e+\alpha$ give $\E'h_u^\Delta\le h_u$, so
concavity and $\min\{4,x\}\le4\min\{1,x\}$ give
$\E'g_u^\Delta\le4g_u$; pointwise, $g_u^\Delta>0$ forces a disjoint
candidate pair in the cell, whence $g_u\ge\alpha^2$ and
$g_u^\Delta\le4\alpha^{-2}g_u$.  Thus
$\E'\Gamma_\Delta\preceq4\Gamma_{\rm o}$ and
$\Gamma_\Delta\preceq4\alpha^{-2}\Gamma_{\rm o}$, while the rootwise
matching and the cap give
$|v^{\mathsf T}\cQ_\Delta v|\le v^{\mathsf T}\Gamma_\Delta v$
pointwise, as in \cref{lem:odd-uniqueness}.  The conditional form of
\cref{lem:copy-transfer}, with $C=4$,
$M=4\alpha^{-2}\le C_{k,\delta}\,n$ (by \cref{eq:odd-alpha} and
$\rho\le1$), $\zeta=0$, and $p_0=n^{-B\ell}$ with $B$ large, yields
$\|\Gamma_{\rm o}^{-1/2}\cQ\,\Gamma_{\rm o}^{-1/2}\|
\le C_{r,\delta}\sqrt{\ell/d}$ with probability $1-o(1)$.  Combining
the two parts through \cref{eq:odd-centering} proves the claim.
\end{proof}

\subsection{Detection and weak recovery}

\begin{proposition}\label{prop:odd-witness}
There is a constant $c_r^{\rm sig}>0$, depending only on $r$, such that
under the planted law, with probability $1-o(1)$,
\begin{equation}\label{eq:odd-witness}
 \lambda_{\max}(H_{\rm o})\ge c_r^{\rm sig}\rho^2.
\end{equation}
\end{proposition}
\begin{proof}
On the events of \cref{prop:odd-population}, using
\cref{lem:odd-palm-degree},
\[
 \1^{\mathsf T}\cP_{\rm o}\1
 =d_\circ N+o(d)N\ge c_rdN,
 \qquad
 \1^{\mathsf T}\Gamma_{\rm o}\1
 =\Tr\Gamma_{\rm o}\le C_rdN.
\]
Test the gauged normalized matrix with
$g=\Gamma_{\rm o}^{1/2}\1/\|\Gamma_{\rm o}^{1/2}\1\|$.
By \cref{prop:odd-comparison}, its Rayleigh quotient is at least
$c_r\rho^2-C_{r,\delta}\sqrt{\ell/d}$.  By
\cref{eq:odd-d-scale}, the second term is at most
$C_{r,\delta}\rho^2/C_{k,\delta}$; choose the sample constant large.
\end{proof}

\begin{theorem}[Odd detection]\label{thm:odd-detection}
If $m\ge C_{k,\delta}\rho^{-2}\Dkl$ with $C_{k,\delta}$ sufficiently
large, a gap-free Krylov run on $H_{\rm o}$
to additive accuracy $c_r^{\rm sig}\rho^2/4$, followed by a fixed
threshold between the null and planted constants, achieves strong
detection.
\end{theorem}
\begin{proof}
By \cref{prop:odd-null,eq:odd-d-scale}, under the null
$\|H_{\rm o}\|\le C_{r,\delta}\rho^2/C_{k,\delta}$ with probability
$1-o(1)$, while under planting \cref{prop:odd-witness} gives a top
eigenvalue at least $c_r^{\rm sig}\rho^2$; increase $C_{k,\delta}$ to
separate the constants and apply \cref{lem:krylov}, the deterministic
contraction being \cref{lem:odd-uniqueness}(3).
\end{proof}

\begin{proposition}[Odd weak recovery]\label{prop:odd-recovery}
Under the hypotheses of \cref{thm:odd-detection}, the high-Rayleigh vector
for $H_{\rm o}$, followed by \cref{thm:transfer,lem:rounding,lem:validation},
outputs $\widehat x\in\{\pm1\}^n$ with
$\ip{\widehat x}{\xstar}\ge c_kn$ with probability $1-o(1)$, in
$n^{\ell+O_k(1)}$ time.
\end{proposition}
\begin{proof}
Let $v$ be a unit vector with
$v^{\mathsf T}H_{\rm o}v\ge c_r\rho^2$, and put
$w:=\Gamma_{\rm o}^{-1/2}Zv$ and $u:=w/\|w\|$, so that
$w^{\mathsf T}\Gamma_{\rm o}w=1$.  Choosing the leading sample constant
so that the comparison error in \cref{eq:odd-comparison} is at most
$c_r\rho^2/2$ gives $w^{\mathsf T}\cP_{\rm o}w\ge c'_r$; by
\cref{eq:odd-pop-conc} and the floor $d\|w\|^2\le1$,
\[
 \frac{d_\circ}{b_{2r,\ell}}w^{\mathsf T}J_{2r}w
 \ge c'_r-o(d)\|w\|^2\ge c'_r/2,
\]
and since $d_\circ\|w\|^2\le C_r$, division by $\|w\|^2$ gives
$u^{\mathsf T}J_{2r}u\ge c''_rb_{2r,\ell}$.  The vector available to
the algorithm is
$\Gamma_{\rm o}^{-1/2}v/\|\Gamma_{\rm o}^{-1/2}v\|=Zu$, because the
diagonal gauge commutes with the normalizer.  Apply
\cref{thm:transfer} with $s=2r$, then gap-free list rounding and
independent validation; since $k$ is odd, validation over the list and
its negatives fixes the orientation.
\end{proof}

\section{From weak to exact recovery: proof of the inference theorems}
\label{sec:exact}

\subsection{Boosting by targeted voting}

The spectral stage ends at constant overlap; reducing the error further
is independent of the hierarchy, and its target-error dependence
reappears in the minimax bound of \cref{thm:minimax}.

\begin{proof}[Proof of \cref{thm:main-boost} (upper bound)]
Apply \cref{lem:subsample} to the $m_{\rm c}$ fresh clauses with
$\mu=m_{\rm c}/2$, obtaining a Poisson cleanup pool; the failure
probability $e^{-m_{\rm c}/16}$ is $o(1)$ since $m_{\rm c}\ge n$, by
\cref{eq:cleanup-curve} and $\vartheta\le\vartheta_0(k)<1$.  For
every cleanup clause, draw a target $i\in F$ uniformly and cast the vote
\begin{equation}\label{eq:targeted-vote}
 Y_{F,i}:=y_F\prod_{j\in F\setminus\{i\}}\widetilde x_j .
\end{equation}
Let $U_i$ be the sum of the votes targeted at $i$ and output
$\widehat x_i:=\sign(U_i)$, breaking a zero tie uniformly.

Choose $\sigma\in\{\pm1\}$ so that
$q:=n^{-1}\ip{\sigma\widetilde x}{\xstar}\ge c$.  This sign is used only
in the analysis; the algorithm never needs to know it.  Conditional on
targeting $i$, the remaining set
$F\setminus\{i\}$ is a uniform $(k-1)$-subset of $[n]\setminus\{i\}$,
and sampling without replacement differs from independent sampling by
$O_k(1/n)$, so
\begin{equation}\label{eq:vote-mean}
 \E\bigl[\sigma^{k-1}x_i^\star\,Y_{F,i}\ \big|\ i\text{ targeted}\bigr]
 =\rho\bigl(q^{k-1}+O_k(n^{-1})\bigr)\ \ge\ c_k\rho .
\end{equation}
By Poisson marking, the votes split into $n$ independent Poisson
coordinate processes of mean $\bar m:=m_{\rm c}/(2n)$ each.  For one
coordinate, the moment generating function of the compound Poisson sum
gives, with $V:=\sigma^{k-1}x_i^\star Y\in[-1,1]$ and $\E V\ge c_k\rho$,
\[
 \bbP\{\widehat x_i\ne\sigma^{k-1}x_i^\star\}
 \le\inf_{0<s\le1}\exp\bigl(\bar m(\E e^{-sV}-1)\bigr)
 \le\exp\bigl(\bar m(-sc_k\rho+s^2)\bigr)\Big|_{s=c_k\rho/2}
 \le\exp\bigl(-c'_k\rho^2m_{\rm c}/n\bigr).
\]
Choose the constant in \cref{eq:cleanup-curve} so that the right side is
at most $\vartheta^4$.  The coordinate errors are independent.  If
$\vartheta n\ge10\log n$, a Chernoff bound on the dominating
$\operatorname{Bin}(n,\vartheta^4)$ count gives at most $\vartheta n$
errors with probability $1-o(1)$; if $\vartheta n<10\log n$, the same
bound with a larger fixed constant makes each coordinate error
probability at most $n^{-5}$, so no error occurs.  Since
$\sigma^{k-1}=\sigma$ for even $k$ and $\sigma^{k-1}=1$ for odd $k$, the
output matches the identifiability of the model, proving the claimed
bound on $\err_k$; the running time is $O(km_{\rm c}+n)$.
\end{proof}

\subsection{Proof of \texorpdfstring{\cref{thm:main-inference}}{the inference theorem}}

\begin{figure}[t]
\centering
\fbox{\begin{minipage}{0.94\textwidth}\small
\textbf{The inference algorithm at level $\ell$.}
\emph{Input:} $I=(F_a,y_a)_{a\le m}$, level $\ell$, bias $\rho$.
\begin{enumerate}[leftmargin=1.7em,itemsep=0.1em,topsep=0.2em]
\item\emph{Subsample and split.}  Retain a Poisson pool of mean
$\tfrac12C_{k,\delta}\rho^{-2}\Dkl$ (\cref{lem:subsample}), split into
spectral and validation pools; on failure, output \textsf{null} and an
arbitrary assignment.
\item\emph{Build the operator.}  Even $k$: form
$A$, $D$, $\Gamma=D+\bd I$, and $K$ as in \cref{eq:even-operator}.
Odd $k$: root every spectral clause, adjoin the dummy candidate
process, build the mutual exact-two graph and the capped diagonal, and
form $H_{\rm o}$ as in \cref{eq:odd-H}; abort if the candidate count
exceeds $C_{\rm cap}$ of \cref{eq:candidate-cap}.
\item\emph{Detect and retain a high-Rayleigh vector.}  Run
\cref{lem:krylov} once on $K$ (resp.\ $H_{\rm o}$), output
\textsf{planted} iff the Rayleigh quotient is at least $\rho/3$
(resp.\ $(c_r^{\rm sig}/2)\rho^2$), and retain the vector $v$.
\item\emph{List.}  Form
$u':=\Gamma^{-1/2}v/\|\Gamma^{-1/2}v\|$ (with
$\Gamma_{\rm o}$ at odd arity), compute $\rhoone(u')$ exactly, project
onto eigenvalues $\ge c_k/4$, and apply randomized clipped rounding
$O_k(\log n)$ times (\cref{lem:rounding}).
\item\emph{Validate.}  Output the maximizer of $|T(z)|$ over the list
on the validation pool; for odd $k$, of $T(z)$ over the list and its
negatives (\cref{lem:validation}).
\end{enumerate}
\emph{Cost:} $n^{\ell+O_k(1)}$ time, $\binom n\ell\,n^{O_k(1)}$ working
memory.
\end{minipage}}
\caption{The algorithm of \cref{thm:main-inference}.  Only the operator
built in Step~2 and the acceptance thresholds depend on the parity of
$k$.}
\label{fig:inference}
\end{figure}

\begin{proof}[Proof of \cref{thm:main-inference}]
The algorithm is summarized in \cref{fig:inference}.  Detection is
\cref{prop:even-detection} for even $k$ and \cref{thm:odd-detection} for
odd $k$; weak recovery is \cref{prop:even-recovery} and
\cref{prop:odd-recovery}.

Every probabilistic ingredient (the subsampling coupling, the trace and
population events, the candidate cap, the Krylov failure, and the
rounding and validation failures) has been bounded uniformly over the
levels in \cref{eq:range} for fixed $(k,\delta,\rho)$, and there are
$O_k(1)$ of them, so the total error is $o(1)$ uniformly.  For the
running time: by \cref{lem:construction}, all operators, matrix--vector
products, and the exact one-particle density matrix cost
$n^{\ell+O_k(1)}$; the candidate count is polynomial by the abort
\cref{eq:candidate-cap}; and \cref{lem:krylov} uses
$O_{k,\rho}(\log N)$ products, on deterministic contractions by
\cref{eq:even-contraction} and \cref{lem:odd-uniqueness}(3), so its
accuracy guarantees apply as stated.
\end{proof}

Combining \cref{thm:main-inference} with \cref{thm:main-boost},
$C_{k,\delta}\rho^{-2}\Dkl+C_k\rho^{-2}n\log n$ clauses give exact
recovery with probability $1-o(1)$.

\section{Refutation: sound certificates at both parities}
\label{sec:refutation}

This section constructs the certificates behind
\cref{thm:main-refutation}; exact evaluation and the sum-of-squares
proofs follow in \cref{sec:evaluation}.  A certificate may not
subsample: it must account for every clause of the given input.  At
even arity this costs nothing: the normalized matrix of
\cref{sec:even}, built from the full input, is already a deterministic
certificate.  At odd arity we root the actual clauses, retain a
same-root disjoint row edge exactly when its pair is the mutual
exact-two pair at both endpoint cells, and charge every omitted edge to
an explicit scalar deficit; the retention is insertion-decreasing, the
capped diagonal plus floor is insertion-increasing (\ref{P-monotone}
again), and the null value is bounded by the pair trace run directly in
the fixed-$m$ model, with no Poissonization and no auxiliary coloring.
\Cref{fig:refutation} summarizes the procedure.  This section relies on
\cref{sec:trace}, on \cref{sec:even}
(\cref{eq:even-operator,prop:even-null}), and on the exact-two rule and
the counting quantities $D_+$ and $\Delta_{r,\ell}$ from
\cref{sec:odd}; it is independent of the rest of \cref{sec:odd}.

\begin{figure}[t]
\centering
\fbox{\begin{minipage}{0.94\textwidth}\small
\textbf{The refutation certificate at level $\ell$.}
\emph{Input:} $I=(F_a,y_a)_{a\le m}$, level $\ell$, target advantage
$\eps$.
\begin{enumerate}[leftmargin=1.7em,itemsep=0.1em,topsep=0.2em]
\item\emph{Even $k$.}  Build $A$, $D$, $\Gamma$ from \emph{all} $m$
clauses (\cref{eq:even-operator} with $\bar m=m$); compute an exact
enclosure $U\ge\|K\|$ at tolerance $\eps/4$ by rational bisection with
inertia tests (\cref{sec:solvers}); output $\Cert(I):=2U$.  No
randomness, no subsampling, no blocking.
\item\emph{Odd $k$.}  If $m$ exceeds the threshold size, partition the
clauses into deterministic blocks (\cref{lem:blocking}).  In each
block: root every clause uniformly at random; enumerate the retained
mutual exact-two edges, the capped diagonal $G_{\rm ref}$, and the
scalars $\Def$, $N_{\rm int}$; compute an exact enclosure
$U\ge\max\{0,\lambda_{\max}(H_{\rm ref})\}$; output the block value
\cref{eq:odd-cert-value}.  Combine blocks by the weighted clipped
average of \cref{lem:blocking}.
\end{enumerate}
\emph{Guarantee:} the output dominates $\sup_x|V_I(x)|$ on every input
and for every rooting; only the value attained under the null is
probabilistic.
\end{minipage}}
\caption{The algorithm of \cref{thm:main-refutation}.}
\label{fig:refutation}
\end{figure}

\subsection{Even arity: a deterministic certificate}

\begin{proposition}[Even deterministic certificate]\label{prop:even-certificate}
Build $A,D,\Gamma,K$ from \emph{all} $m$ input clauses, with
$\bar m=m$ in \cref{eq:even-mean-degree}.  For every input and every
$x\in\{\pm1\}^n$,
\begin{equation}\label{eq:even-certificate}
 |V_I(x)|\ \le\ 2\|K\| .
\end{equation}
Consequently, any exactly verified $U\ge\|K\|$ gives the sound
certificate $\Cert(I):=2U$.  If the bisection of \cref{sec:solvers} is
run to tolerance $U\le\|K\|+\eps/4$, using the a priori bracket
$[0,1]$ from \cref{eq:even-contraction}, then for every
$m\ge C_{k,\delta}\eps^{-2}\Dkl$ the certificate is at most
$\eps$ under the null with probability $1-o(1)$.  The procedure is
deterministic and requires neither subsampling nor blocking.
\end{proposition}
\begin{proof}
For $z_S=x^S$ one has $z_Sz_T=x^{S\sd T}$, so by
\cref{eq:support-count},
\[
 z^{\mathsf T}Az=m\,\Tr P_F\cdot V_I(x),
 \qquad
 z^{\mathsf T}\Gamma z=z^{\mathsf T}Dz+\bd\,\|z\|^2
 =2m\,\Tr P_F,
\]
using $\bd N=m\Tr P_F$.  The generalized Rayleigh inequality applied to
$A$ and to $-A$ proves \cref{eq:even-certificate}.  For the null value,
\cref{prop:even-null} applies directly to the full fixed-$m$ instance,
for every $m$, with no Poisson device; since $\bd$ is increasing in $m$,
the bound $\|K\|\le C_{r,\delta,B}\sqrt{\ell/(r\bd+\ell)}$ is at most
$\eps/4$ for all $m\ge C_{k,\delta}\eps^{-2}\Dkl$ once
$C_{k,\delta}$ is large, and the tolerance adds at most $\eps/2$.
\end{proof}

\subsection{Odd arity: the exact-two certificate graph}
\label{sec:oddref-construction}

Throughout the rest of this section $k=2r+1$, the full input is used
without subsampling, and the only internal randomness is the rooting.
Root every clause independently and uniformly and write
\[
 F_a=\{u_a\}\mathbin{\dot\cup}C_a,
 \qquad |C_a|=2r,
 \qquad
 g_u(x):=\sum_{a:u_a=u}y_ax^{C_a}.
\]
For every $x\in\{\pm1\}^n$,
\begin{equation}\label{eq:odd-cauchy}
 \bigl(mV_I(x)\bigr)^2
 =\Bigl(\sum_ux_ug_u(x)\Bigr)^2
 \le n\sum_ug_u(x)^2
 =n\bigl(m+P_0(x)+P_{\rm int}(x)\bigr),
\end{equation}
where $P_0$ is the ordered sum over distinct same-root pairs with disjoint
residuals and $P_{\rm int}$ the ordered sum over those with intersecting
residuals.  If $N_{\rm int}$ counts the latter ordered pairs, then
$|P_{\rm int}(x)|\le N_{\rm int}$.

By \cref{eq:odd-Delta}, a disjoint physical pair has exactly
$\Delta_{r,\ell}$ compatible ordered rows; since the involution
$S\mapsto S\sd C_a\sd C_b$ has no fixed point, the pair generates
$q_0:=\Delta_{r,\ell}/2$ compatible unordered row transitions.

The certificate graph is the mutual exact-two graph of
\cref{sec:odd-construction}, built from the rooted clauses themselves;
no dummy candidates and no marks are needed.  For a root $u$ and a row
$S$, let $\cL_u(S):=\{a:u_a=u,\ |C_a\cap S|=r\}$.  A \emph{complete row
edge} is a tuple $e=(\{a,b\},\{S,T\})$ with $a\ne b$, $u_a=u_b=:u$,
$C_a\cap C_b=\varnothing$, both residuals legal at $S$, and
$T=S\sd C_a\sd C_b$, recorded once per unordered row pair, so that each
unordered disjoint physical pair generates exactly $q_0$ complete row
edges.  Retain $e$ if and only if
\begin{equation}\label{eq:hard-isolation}
 \cL_u(S)=\cL_u(T)=\{a,b\},
\end{equation}
write $\omega_e\in\{0,1\}$ for this retention indicator, and put
$A_e:=E_{S,T}+E_{T,S}$.  Exactly as in \cref{lem:odd-uniqueness}, at
most one retained edge is incident with each cell $(u,S)$; in
particular, for every row $S$ and clause $a$,
\begin{equation}\label{eq:hard-port-capacity}
 \sum_{e:S\in e,\,a\in e}\omega_e\le1.
\end{equation}
Define
\begin{equation}\label{eq:odd-hard-ops}
\begin{gathered}
 C_{\rm ref}:=\sum_e\omega_ey_{a(e)}y_{b(e)}A_e,
 \qquad
 G_{\rm ref}(S,S):=\sum_{u=1}^ng^{\rm ref}_u(S),\\
 g^{\rm ref}_u(S):=\1\bigl\{\exists\,\{a,b\}\subseteq\cL_u(S):
 C_a\cap C_b=\varnothing\bigr\},
\end{gathered}
\end{equation}
and
\begin{equation}\label{eq:odd-dstar}
 d_*:=\binom m2\frac1n
 \frac{\binom{n-1-2r}{2r}}{\binom{n-1}{2r}}
 \frac{\Delta_{r,\ell}}N
 =\Theta_r\!\left(\frac{m^2\ell^{2r}}{n^{2r+1}}\right),
 \qquad
 \Gamma_{\rm ref}:=G_{\rm ref}+d_*I,
 \qquad
 H_{\rm ref}:=\Gamma_{\rm ref}^{-1/2}C_{\rm ref}\Gamma_{\rm ref}^{-1/2}.
\end{equation}
The scalar $d_*$ is the exact fixed-$m$ mean of the complete row
degree.  Inserting a clause can only destroy old retention indicators,
and can only increase the capped diagonal $G_{\rm ref}$: this is
\ref{P-monotone} again, now in the fixed-sample model.  A retained edge
forces $g^{\rm ref}_u=1$ at both endpoint cells, and at most one
retained edge meets each cell, so for every vector $v$,
\begin{equation}\label{eq:odd-hard-contraction}
 |v^{\mathsf T}C_{\rm ref}v|
 \le\sum_e\omega_e(v_{S(e)}^2+v_{T(e)}^2)
 \le v^{\mathsf T}G_{\rm ref}v,
 \qquad\text{hence}\qquad \|H_{\rm ref}\|\le1.
\end{equation}

For an unordered disjoint physical pair $p$, let
$W_p:=\sum_{e\text{ generated by }p}\omega_e$.  With $Q$ the number of
such pairs and $E_{\rm ret}$ the retained edge set,
\begin{equation}\label{eq:hard-deficit}
 \Def:=\sum_p(q_0-W_p)
 =q_0Q-|E_{\rm ret}|
 =\sum_{e\text{ complete}}(1-\omega_e).
\end{equation}
The deficit is an explicitly computable nonnegative integer, and it is
how every discarded transition is paid for; likewise
$N_{\rm int}=\sum_uM_u(M_u-1)-2Q$, with $M_u$ the number of clauses
rooted at $u$.

\begin{lemma}[Deterministic odd certificate]\label{lem:odd-cert}
Set $\bar g:=N^{-1}\Tr G_{\rm ref}$.  For every input, every rooting,
and every $x\in\{\pm1\}^n$,
\begin{equation}\label{eq:odd-pair-upper}
 P_0(x)\le
 \frac{2N(\bar g+d_*)}{\Delta_{r,\ell}}
 \max\{0,\lambda_{\max}(H_{\rm ref})\}
 +\frac{4\Def}{\Delta_{r,\ell}}.
\end{equation}
Consequently, for every certified
$U\ge\max\{0,\lambda_{\max}(H_{\rm ref})\}$,
\begin{equation}\label{eq:odd-cert-value}
 \Cert_{\rm odd}(I):=\frac1m\left\{n\left[
 m+\frac{2N(\bar g+d_*)}{\Delta_{r,\ell}}U
 +\frac{4\Def}{\Delta_{r,\ell}}+N_{\rm int}
 \right]\right\}^{1/2}
\end{equation}
is an upper bound on $\sup_x|V_I(x)|$.
\end{lemma}
\begin{proof}
For $z_S=x^S$, every row edge generated by $p=\{a,b\}$ satisfies
$z^{\mathsf T}A_ez=2x^{C_a\sd C_b}$, so
$z^{\mathsf T}C_{\rm ref}z=2\sum_pW_py_ay_bx^{C_a\sd C_b}$ while
$P_0(x)=2\sum_py_ay_bx^{C_a\sd C_b}$.  Summing the termwise bound
$2s\le2(W/q_0)s+2(1-W/q_0)$, valid for $s\in\{\pm1\}$ and
$0\le W\le q_0$, gives
$P_0(x)\le q_0^{-1}z^{\mathsf T}C_{\rm ref}z+2q_0^{-1}\Def$.  Since
$z_S^2=1$, one has
$z^{\mathsf T}\Gamma_{\rm ref}z=\Tr\Gamma_{\rm ref}=N(\bar g+d_*)$, so
the generalized Rayleigh bound and $q_0=\Delta_{r,\ell}/2$ prove
\cref{eq:odd-pair-upper}; combine it with \cref{eq:odd-cauchy} and
$|P_{\rm int}(x)|\le N_{\rm int}$.
\end{proof}

\subsection{The null norm: the pair trace in the fixed-sample model}

Define the cell scale of the rooted input by
\begin{equation}\label{eq:odd-mu}
 \pi_0:=\frac m{(2r+1)\binom n{2r+1}},
 \qquad
 \mu_0:=\pi_0D_+
 =\Theta_r\!\left(\frac{m\ell^r}{n^{r+1}}\right),
\end{equation}
so that $\mu_0$ bounds the mean number of clauses legal in one
root--row cell, and $d_*\asymp_rn\mu_0^2$ by \cref{eq:odd-dstar}.

\begin{proposition}[Odd refutation null norm]\label{prop:iso-null}
Let $t=\lfloor c\ell\log n\rfloor$ and assume
\begin{equation}\label{eq:iso-conditions}
 d_*\ge C_r\ell,
 \qquad
 \mu_0\le c_r,
 \qquad
 t\mu_0\le c_rd_*.
\end{equation}
Then the fixed-$m$ null model satisfies
\begin{equation}\label{eq:iso-null}
 \|H_{\rm ref}\|\le C_{r,\delta}\sqrt{\frac\ell{d_*}}
\end{equation}
with failure probability $n^{-\Omega_r(\ell)}$.
\end{proposition}
\begin{proof}
Apply \cref{thm:packet-trace}, with $C_0=C_0(r)$ sufficiently large and
in its fixed-sample form,
directly to the $m$ rooted clause occurrences: the labels are i.i.d.\
uniform rooted types, the fresh-label steps are integrated by
\cref{eq:fixed-deletion} through \cref{rem:fixed-recursion}, and the
deterministic floor $d_*$ is held fixed throughout.  Under the null,
the signs $y_a$ are independent Rademacher conditional on the rooted
supports.  The retained amplitude of a complete row edge is its
retention indicator \cref{eq:hard-isolation}: inserting a clause can
only destroy an old indicator and can only increase $G_{\rm ref}$,
which is \ref{P-monotone}; a retained edge is dominated by the capped
diagonal at both endpoint cells, and each clause carries at most one
retained edge per row (\cref{eq:hard-port-capacity}), which are
\ref{P-diagonal} and \ref{P-port}.

For the fresh masses \ref{P-mass}, take $\beta:=C_r\mu_0$.  A root--row
cell has clause mass at most $\mu_0$, and prescribing $h$ recycled
exterior vertices costs $(C_rt/n)^h$ as always.  Hence two new clauses
with a common root, both legal at the current row, have total mass at
most $C_rn\mu_0^2(C_rt/n)^h\le C_rd_*(C_rt/n)^h$; pairing one new
clause with one of at most $4t$ exposed clauses costs at most
$C_rt\mu_0(C_rt/n)^h$; and a new partner of one fixed exposed clause
costs at most $C_r\mu_0(C_rt/n)^h$.  The side conditions hold:
$\beta\le C_0$ since $\mu_0\le c_r$,
$\beta^2\le C_rd_*/n$ because $d_*\asymp_rn\mu_0^2$, and
$t\beta\le c_rd_*$ is \cref{eq:iso-conditions}.  The tail
\cref{eq:packet-tail} proves the claim.
\end{proof}

\subsection{Degree, deficit, and one-block tightness}

\begin{lemma}\label{lem:odd-corrections}
Fix $\eps>0$ and suppose
$m=\Theta_{r,\delta}(\eps^{-2}\Dkl)$, so that
\[
 \mu_0=\Theta_r(\eps^{-2}\sqrt{\ell/n}),
 \qquad
 d_*=\Theta_r(\eps^{-4}\ell).
\]
If $\mu_0=o(1)$, then with probability $1-O(\sqrt{\mu_0})-o(1)$,
\begin{equation}\label{eq:odd-corrections}
 \bar g\le(1+o(1))d_*,
 \qquad
 \Def\le C_r\sqrt{\mu_0}\,Nd_*,
 \qquad
 \frac{nN_{\rm int}}{m^2}=o(1).
\end{equation}
\end{lemma}
\begin{proof}
\emph{Degree.}  A fixed index pair is a same-root disjoint pair with
probability
$\frac1n\binom{n-1-2r}{2r}/\binom{n-1}{2r}=\Theta_r(1/n)$; disjoint
index pairs are independent, and two pair indicators sharing one clause
are simultaneously one with probability $O_r(n^{-2})$.  Hence
$\E Q=\binom m2\frac1n\binom{n-1-2r}{2r}/\binom{n-1}{2r}$ exactly, and
$\Var Q=O_r(m^2/n+m^3/n^2)=o((\E Q)^2)$, so
$Q=(1+o(1))\E Q$ with probability $1-o(1)$.  Every disjoint pair is
legal on $\Delta_{r,\ell}$ ordered rows, so
$\Tr G_{\rm ref}\le\Delta_{r,\ell}Q$, while
$d_*=\Delta_{r,\ell}\,\E Q/N$ by \cref{eq:odd-dstar}; this proves the
first claim.

\emph{Deficit.}  Define the triple-collision count
\begin{equation}\label{eq:triple-collision}
 T_3:=\sum_{u\in[n],\,S\in\Omega_\ell}\binom{|\cL_u(S)|}{3}.
\end{equation}
We claim $\Def\le3T_3$.  An omitted complete edge fails
\cref{eq:hard-isolation} at one of its endpoint cells; since its two
clauses are legal there, that cell contains a third legal clause.
Assign the omitted edge to one such cell.  A pair
$\{a,b\}\subseteq\cL_u(S)$ determines the opposite endpoint
$T=S\sd C_a\sd C_b$, so at most $\binom{|\cL_u(S)|}2$ omitted edges are
assigned to the cell $(u,S)$, and cells with fewer than three legal
clauses receive none.  Since $\binom q2\le3\binom q3$ for $q\ge3$,
summing over cells proves the claim.  A cell count is binomial with
mean at most $\mu_0$, so its third factorial moment is at most
$\mu_0^3$ and $\E\binom{|\cL_u(S)|}3\le\mu_0^3/6$; hence
\[
 \E T_3\le C_r\,nN\mu_0^3\le C_r\mu_0\,Nd_* ,
\]
using $n\mu_0^2\le C_rd_*$.  Markov's inequality at scale
$\sqrt{\mu_0}\,Nd_*$ gives the second claim with probability
$1-O_r(\sqrt{\mu_0})$.

\emph{Intersections.}  Two rooted clauses share a root with probability
$1/n$ and, conditionally, have intersecting residuals with probability
$O_r(1/n)$, so
$\E N_{\rm int}=O_r(m^2/n^2)$.  Markov's inequality therefore gives
$nN_{\rm int}/m^2=o(1)$ with probability $1-o(1)$.  Under the
quantitative condition $\ell/n\le c_{r,\delta}\eps^{12}$, the lower
bound $\ell\ge2r$ implies $n\eps^2\to\infty$; applying Markov after
division by $\eps^2$ strengthens this to
$nN_{\rm int}/m^2=o(\eps^2)$ in probability.  This completes the third
claim.
\end{proof}

\begin{theorem}[Odd two-sided strong refutation on one block]\label{thm:odd-one-block}
Fix $0<\eps<1$ independently of $n$.  There is $C_{r,\delta}$ such that,
for all sufficiently large $n$, whenever
$m=C_{r,\delta}\eps^{-2}\Dkl$ up to a fixed constant factor, the
certificate \cref{eq:odd-cert-value}, evaluated with an exact upper
enclosure
$\max\{0,\lambda_{\max}(H_{\rm ref})\}\le U
\le\max\{0,\lambda_{\max}(H_{\rm ref})\}+\eps^2/(20C_r)$, is sound on every input
and at most $\eps$ under the null with probability $1-o(1)$.
\end{theorem}
\begin{proof}
Soundness is \cref{lem:odd-cert}.  At the stated $m$,
$\mu_0=\Theta_r(\eps^{-2}\sqrt{\ell/n})=o(1)$ because $\eps$ is fixed and
\cref{eq:range} gives $\ell/n=o(1)$.  Also $d_*\ge C_r\ell$ after
increasing the sample constant, and
$t\mu_0/d_*\le C_r\eps^2\sqrt{\ell/n}\log n=o(1)$ by
\cref{eq:range}; hence \cref{prop:iso-null} applies.  The exact identity
\[
 \frac{Nd_*}{\Delta_{r,\ell}}
 =\binom m2\frac1n
   \frac{\binom{n-1-2r}{2r}}{\binom{n-1}{2r}}
 =\Theta_r(m^2/n)
\]
and \cref{lem:odd-corrections} give
\begin{equation}\label{eq:odd-cert-tight}
 \Cert_{\rm odd}(I)^2
 \le C_r\frac{\Dkl}{m}
 +C_r\sqrt{\mu_0}
 +C_r\sqrt{\frac\ell{d_*}}
 +C_r\bigl(U-\max\{0,\lambda_{\max}(H_{\rm ref})\}\bigr)
 +o(1).
\end{equation}
The diagonal term $n/m$ is absorbed into $\Dkl/m$ because
$\Dkl\ge n$ for $\ell\le n/3$.  Choose the leading sample constant so
that the first and third terms total at most $\eps^2/5$; the enclosure
tolerance contributes at most $\eps^2/20$; and the deficit and the
remaining $o(1)$ terms contribute at most $\eps^2/5$ for large $n$,
since $\sqrt{\mu_0}=o(1)$ at fixed $\eps$.  Quantitatively,
$\ell/n\le c_{r,\delta}\eps^{12}$ gives
$\sqrt{\mu_0}\le C_rc_{r,\delta}^{1/4}\eps^2$ and the strengthened
intersection estimate is $o(\eps^2)$, so the same budget is uniform in
the range of \cref{rem:quantitative-ranges}; the right side is then at
most $\eps^2$.
\end{proof}

\begin{lemma}\label{lem:blocking}
Fix $\eps\in(0,1)$ and a certificate procedure $\Cert_0$ that, on
instances whose size lies in $[m_0,2m_0]$, is pointwise sound and satisfies
$\Cert_0\le\eps/4$ under the null with probability $1-o(1)$, uniformly in
that window.  Then for every $m\ge m_0$ the following is pointwise sound
and at most $\eps$ under the null with probability $1-o(1)$: partition the
clause indices into deterministic blocks $I_j$ of sizes
$m_j\in[m_0,2m_0]$ and output
\[
 \Cert(I):=\sum_j\frac{m_j}m\min\{1,\Cert_0(I_j)\}.
\]
\end{lemma}
\begin{proof}
Soundness follows from $V_I=\sum_j(m_j/m)V_{I_j}$ and
$\sup_x|V_{I_j}(x)|\le1$.  Under the null the blocks are independent:
for boundedly many, a union bound makes all blocks good; otherwise the
clipped values lie in $[0,1]$ with means at most $\eps/4+o(1)$ and
weights at most twice the reciprocal of the block count, and weighted
Hoeffding gives a value below $\eps$ with probability $1-o(1)$.
\end{proof}

\section{Exact evaluation and sum-of-squares proofs}\label{sec:evaluation}

This section verifies the construction costs of every object used by
the algorithms (\cref{lem:construction}), converts the spectral
certificates of \cref{sec:refutation} into deterministically sound
exact evaluations, and turns those into explicit sum-of-squares proofs,
completing the proof of \cref{thm:main-refutation}.

\subsection{Construction costs}

\begin{lemma}\label{lem:construction}
On inputs with $m=n^{O_k(1)}$, all objects used by the algorithms are
assembled in $n^{\ell+O_k(1)}$ time: the even matrices $A,D,\Gamma$; the
odd candidate list, root--row cells, exact-two flags, retained edges,
and the capped diagonal $\cG_{\rm o}$; all complete odd-refutation row
edges and their retention indicators, the capped diagonal
$G_{\rm ref}$, and the scalars $Q,\Def,N_{\rm int}$; matrix--vector
products with every normalized operator; and the one-particle density
matrix $\rhoone(u')$ of a stored vector.
\end{lemma}
\begin{proof}
Each even clause is legal on at most $N$ rows, so building $A,D$ costs
$O(mN)n^{O_k(1)}$.  The expected number of odd candidates is
$q|\cT|=\Theta_r(n^{r+1}/\ell^r)$; a residual $C$ is legal on
$\Tr P_C=\binom{2r}{r}\binom{n-2r}{\ell-r}\le N$ rows.  Sorting the
root--row cells gives the lists $\cL_u(S)$, hence the exact-two flags
and the retained edges.  The capped diagonal is obtained by
accumulating the pair potentials $h_u(S)$ cell by cell and capping;
under the candidate cap \cref{eq:candidate-cap}, the total work is
$n^{\ell+O_k(1)}$ deterministically.  Mutual exact-twoness leaves at
most one retained edge per root and row.

For refutation, enumerate same-root disjoint physical pairs, count $Q$,
and list their $q_0$ compatible row edges.  Sorting the cell lists
$\cL_u(S)$ yields the exact-two flags; retain exactly the edges
satisfying \cref{eq:hard-isolation}, and accumulate $G_{\rm ref}$,
$|E_{\rm ret}|$, and $\Def=q_0Q-|E_{\rm ret}|$.  On a threshold block
the edge count is
$\Theta_r(N\ell\eps^{-4})$ in expectation; the deterministic
bound under $m=n^{O_k(1)}$ is still $n^{\ell+O_k(1)}$.  The intersection
count $N_{\rm int}$ is a direct same-root pair enumeration.  Sparse
matrix--vector products have the same complexity.  Finally,
$\rhoone(u')$ has $n^2$ entries, each a sum over at most
$\binom n{\ell-1}\le N$ terms.
\end{proof}

The exact spectral enclosures of \cref{sec:solvers} evaluate each
certificate with one-sided correctness in $n^{O_k(\ell)}$ bit operations;
this conservative bound includes coefficient growth, rational bisection,
and dense exact elimination.

\begin{remark}\label{rem:evaluation-exponent}
The exact-evaluation exponent $n^{O_k(\ell)}$ exceeds the construction
exponent $n^{\ell+O_k(1)}$, but this does not move refutation off the
curve \cref{eq:scale}: since $\Dkl$ is polynomial in $\ell$, running
the evaluation at a constant fraction of the level changes the clause
requirement only by a $k$-dependent constant factor.
\end{remark}

\subsection{Explicit sum-of-squares witnesses}

Work modulo the Boolean ideal $x_i^2=1$, and let
$z(x)=(x^S)_{S\in\Omega_\ell}$, a vector of degree-$\ell$ monomials.
We first record the exact factorization used below.

\begin{lemma}
\label{lem:rational-psd}
Every rational positive semidefinite matrix $M$ admits a permutation
$\Pi$, a rational unit lower-triangular matrix $L$, and a diagonal matrix
$\Lambda$ with nonnegative rational entries such that
\[
 \Pi^{\mathsf T}M\Pi=L\Lambda L^{\mathsf T}.
\]
Consequently every rational quadratic form defined by a positive
semidefinite matrix is a sum of squares of rational linear forms.
\end{lemma}
\begin{proof}
If a diagonal entry of a positive semidefinite matrix is zero, the
corresponding row and column vanish by
$|M_{ij}|^2\le M_{ii}M_{jj}$ and may be removed.  Otherwise permute a
positive rational diagonal entry into the first position.  A rational
Gaussian step factors off this pivot and leaves its rational Schur
complement, which is again positive semidefinite.  Induction gives the
stated $L\Lambda L^{\mathsf T}$ factorization.  Finally, each positive
rational number is a sum of four rational squares (apply Lagrange's
four-square theorem to the product of its numerator and denominator), so
all diagonal weights can be absorbed into rational linear forms.
\end{proof}

For even $k$, the certified pencil is rational and positive semidefinite,
$\theta\Gamma-A\succeq0$.  Apply \cref{lem:rational-psd} to obtain
$\Pi^{\mathsf T}(\theta\Gamma-A)\Pi=L\Lambda L^{\mathsf T}$.  Hence
\begin{equation}\label{eq:sos-pencil}
 \theta\,z(x)^{\mathsf T}\Gamma z(x)-z(x)^{\mathsf T}Az(x)
 =\sum_i\Lambda_{ii}
 \bigl((L^{\mathsf T}\Pi^{\mathsf T}z(x))_i\bigr)^2,
\end{equation}
and the last sentence of \cref{lem:rational-psd} converts this into a
literal sum of squares of rational polynomials of degree at most $\ell$.
Since $z^{\mathsf T}Az=m\Tr P_F\,V_I(x)$ and
$z^{\mathsf T}\Gamma z=2m\Tr P_F$ identically modulo the ideal, this is
a degree-$2\ell$ sum-of-squares proof of $V_I(x)\le2\theta$, and the
factorization of $\theta\Gamma+A$ gives the other side.

For odd $k$, the Cauchy--Schwarz step \cref{eq:odd-cauchy} has an exact
algebraic form: Lagrange's identity, with $\sum_ux_u^2$ replaced by $n$
using the ideal, gives
\begin{equation}\label{eq:sos-lagrange}
 n\sum_ug_u(x)^2-\Bigl(\sum_ux_ug_u(x)\Bigr)^2
 =\sum_{u<v}\bigl(x_vg_u(x)-x_ug_v(x)\bigr)^2 ,
\end{equation}
with ideal multipliers of degree $O_k(1)$.  The retained pair bound is
supplied by \cref{eq:sos-pencil} applied to the pencil
$\theta\Gamma_{\rm ref}-C_{\rm ref}$, and every deficit and intersection
correction is a termwise application of
\begin{equation}\label{eq:sos-monomial}
 1\pm x^A=\tfrac12(1\pm x^A)^2
 \qquad\text{modulo }x_i^2=1,
\end{equation}
with nonnegative integer coefficients such as $q_0-W_p$.  These
identities together give a sum-of-squares proof of
$\Cert'^2-V_I(x)^2\ge0$, where, for any prescribed rational tolerance
$\tau_0>0$, exact rational upper approximation of the square root produces
$\Cert_{\rm odd}\le\Cert'\le\Cert_{\rm odd}+\tau_0$.  The
rationalization error can therefore be made arbitrarily smaller than the
target advantage rather than consuming a fixed fraction of it.  Since
$\Cert'>0$, the identity
\[
 \Cert'\mp V_I
 =\frac{(\Cert'\mp V_I)^2+(\Cert'^2-V_I^2)}{2\,\Cert'}
\]
turns it into separate one-sided proofs without increasing the degree.
Every polynomial involved has degree $O_k(\ell)$.

\begin{corollary}[Explicit sum-of-squares certificates]\label{cor:sos}
At the sample scales of \cref{thm:main-refutation}, random $k$XOR has,
with probability $1-o(1)$, a degree-$O_k(\ell)$ sum-of-squares proof of
the certified two-sided strong-refutation bound, computable in
$n^{O_k(\ell)}$ bit operations.
\end{corollary}

\subsection{Proof of \texorpdfstring{\cref{thm:main-refutation}}{the refutation theorem}}

\begin{proof}[Proof of \cref{thm:main-refutation}]
For even $k$, use all clauses and output $\Cert(I)=2U$, where $U$ is
an exact enclosure of $\|K\|$ at tolerance $\eps/4$.  This is sound and
at most $\eps$ under the null for every
$m\ge C_{k,\delta}\eps^{-2}\Dkl$ by
\cref{prop:even-certificate}; the procedure is deterministic and uses no
blocking.

For odd $k$, run the one-block construction at the internal target
$\eps/8$.  Increase the leading constant in the main theorem so that its
sample threshold dominates the corresponding
$C_{r,\delta}(\eps/8)^{-2}\Dkl$ threshold.  Use eigenvalue tolerance
$(\eps/8)^2/(20C_r)$ and then choose a rational upper approximation
$\Cert'$ to the square-root certificate with error at most $\eps/8$.
Thus a null block has $\Cert'\le\eps/4$ with probability $1-o(1)$,
while soundness holds for every input and rooting.  Apply
\cref{lem:blocking} to all larger inputs.  This explicit slack also makes
the degree-$O_k(\ell)$ sum-of-squares bound certify the advertised target
rather than a larger one.  For a clipped block, use either its computed
sum-of-squares proof or the trivial proof of the bound one; nonnegative
weighted sums of the block proofs preserve the degree.  The construction
costs are \cref{lem:construction}, the exact evaluation costs are those
of \cref{sec:solvers}, and the sum-of-squares proofs are \cref{cor:sos}.
\end{proof}

\section{Matching lower bounds in the same model}\label{sec:lower}

This section proves \cref{thm:main-lower} and the minimax bound
\cref{thm:minimax}.  Everything here concerns the fixed-$m$ laws $Q_m$
and $P_{m,x}$ of \cref{sec:results} directly: no Poissonization, no
conditioning, and no rare-collision estimate appears anywhere.
Conditional on the supports, the sign monomials form an exact
orthonormal basis, and the whole calculation reduces to an ordered count
of clause families with a prescribed parity boundary.  These families
are exactly the labeled even covers when the boundary is empty, which is
the promised mirror image of the trace expansions of \cref{sec:trace}.

\subsection{Sign degree}\label{sec:sign-degree}

Write $F=(F_1,\ldots,F_m)$ and $y=(y_1,\ldots,y_m)$ for the supports and
signs of an instance, and for $A\subseteq[m]$ put
$y_A:=\prod_{a\in A}y_a$.  Conditional on $F$, the family
$(y_A)_{A\subseteq[m]}$ is the Walsh orthonormal basis of the sign space
under $Q_m$, so every $g\in L^2(Q_m)$ has a unique expansion
$g(F,y)=\sum_Ag_A(F)y_A$ with square-integrable coefficients, and
Parseval holds conditionally.

\begin{definition}[Sign degree]\label{def:sign-degree}
For $D\ge0$, let $\cH_D\subseteq L^2(Q_m)$ be the closed subspace of
functions
\begin{equation}\label{eq:sign-degree}
 g(F,y)=\sum_{|A|\le D}g_A(F)\,y_A ,
\end{equation}
where each coefficient $g_A$ may be an arbitrary square-integrable
function of the \emph{entire} support sequence.  Write $\Pi_{\le D}$ for
the orthogonal projection onto $\cH_D$.
\end{definition}

The class is deliberately generous on the support side: the support
multihypergraph may be preprocessed without restriction, repeated
supports may be located and used freely, and only the number of sign
coordinates in a Walsh monomial is charged.  Products obey
$\cH_d\cdot\cH_d\subseteq\cH_{2d}$, because
$y_Ay_{A'}=y_{A\sd A'}$ and $|A\sd A'|\le|A|+|A'|$.

\begin{lemma}
\label{lem:class-containment}
Each of the following lies in $\cH_D$:
\begin{enumerate}
\item every polynomial of degree at most $D$ in the signed one-hot
variables $X_{a,\alpha}:=y_a\1\{F_a=\alpha\}$,
$\alpha\in\binom{[n]}k$;
\item every polynomial of degree at most $D$ in the aggregated signed
counts $Z_\alpha:=\sum_{a}y_a\1\{F_a=\alpha\}$;
\item every function of Efron--Stein degree at most $D$ in the
independent observations $(F_a,y_a)_{a\le m}$.
\end{enumerate}
Moreover, an independent random seed may be adjoined to the support-side
data, so the lower bounds below also cover randomized procedures of sign
degree $D$.
\end{lemma}
\begin{proof}
A monomial of degree at most $D$ in the one-hot variables involves at
most $D$ distinct sample indices; after reducing $y_a^2=1$ and
conditioning on $F$, it is a support-dependent coefficient times $y_A$
with $|A|\le D$.  Expanding a polynomial in the counts produces the same
form.  A term of an Efron--Stein decomposition depending on at most $D$
sample coordinates is, conditional on all supports, a function of at
most $D$ signs, whose Walsh expansion uses only subsets of those
coordinates.  For the last claim, a seed $U$ independent of the instance
is absorbed into the coefficients $g_A(F,U)$, and all arguments below
apply conditionally on $U$.
\end{proof}

\subsection{Exact boundary expansions}

For $A\subseteq[m]$, define the parity boundary of the corresponding
clause family in the realized supports,
\[
 \partial_FA:=\bigsd_{a\in A}F_a\ \subseteq[n];
\]
two labels with the same support cancel in the boundary, as they must.
For fixed $x$, the likelihood ratio is
\begin{equation}\label{eq:lr}
 \frac{dP_{m,x}}{dQ_m}(F,y)=\prod_{a=1}^m\bigl(1+\rho\,y_ax^{F_a}\bigr),
\end{equation}
since the support laws agree and the one-sign ratio is
$1+\rho y_ax^{F_a}$.  For $B\subseteq[n]$ define the tilted likelihood
ratio $L_B:=\E_x[x^B\,dP_{m,x}/dQ_m]$, with $x$ uniform; in particular
$L_\varnothing=dP_m/dQ_m$ for the planted mixture $P_m$.

\begin{lemma}\label{lem:boundary-lr}
For every $B\subseteq[n]$,
\begin{equation}\label{eq:boundary-expansion}
 L_B(F,y)=\sum_{A\subseteq[m]:\ \partial_FA=B}\rho^{|A|}\,y_A,
 \qquad
 \Pi_{\le D}L_B=\sum_{\substack{|A|\le D\\ \partial_FA=B}}
 \rho^{|A|}\,y_A,
\end{equation}
and consequently
\begin{equation}\label{eq:projection-norm}
 \bigl\|\Pi_{\le D}L_B\bigr\|_{L^2(Q_m)}^2
 =\sum_{t=0}^{\min\{D,m\}}\binom mt\rho^{2t}\,q_{t,B},
 \qquad
 q_{t,B}:=\bbP\{F_1\sd\cdots\sd F_t=B\},
\end{equation}
where $F_1,\ldots,F_t$ are independent uniform $k$-sets.
\end{lemma}
\begin{proof}
Expanding \cref{eq:lr},
$x^B\,dP_{m,x}/dQ_m=\sum_A\rho^{|A|}y_Ax^{B\sd\partial_FA}$, and the
uniform average of $x^U$ is $\1\{U=\varnothing\}$; conditional Walsh
truncation is the exact orthogonal projection.  For the norm, apply
Parseval conditionally on $F$ and average: each labeled set $A$ of size
$t$ has boundary probability $q_{t,B}$, and there are $\binom mt$ of
them.
\end{proof}

\subsection{The boundary count}

\begin{lemma}\label{lem:boundary-count}
Let $B\subseteq[n]$ with $b:=|B|$, and let $F_1,\ldots,F_t$ be
independent uniform $k$-subsets of $[n]$.  If $kt-b$ is negative or odd
then $q_{t,B}=0$.  Otherwise, for $n\ge2k$,
\begin{equation}\label{eq:boundary-count}
 q_{t,B}\ \le\ \left(\frac{2kt}n\right)^{(kt+b)/2} .
\end{equation}
\end{lemma}
\begin{proof}
The boundary condition says each vertex of $B$ has odd total degree and
each other vertex even total degree, so $kt-b$ is twice the number
$h:=(kt-b)/2$ of surplus incidence pairs, proving the parity and range
constraints.  Let $N^{\rm ord}_{t,B}$ be the number of ordered support
tuples with boundary $B$, and give each support an ordering of its $k$
elements: each tuple yields exactly $(k!)^t$ vertex sequences of length
$K:=kt$, all with the stated degree parities, and distinct tuples yield
distinct sequences.  Dropping the constraint that entries within a block
are distinct,
\[
 (k!)^tN^{\rm ord}_{t,B}
 \ \le\
 K!\sum_{\substack{a_1,\ldots,a_n\ge0\\ \sum_va_v=h}}
 \frac1{\prod_{v\in B}(2a_v+1)!\ \prod_{v\notin B}(2a_v)!}
 \ \le\
 \frac{K!}{2^hh!}\,n^h,
\]
using $(2a)!\ge2^aa!$, $(2a+1)!\ge2^aa!$, and the multinomial identity
$\sum_{\sum a_v=h}\prod_v(a_v!)^{-1}=n^h/h!$.  Since
$(k!)^t\binom nk^t=(n)_k^t\ge(n/2)^K$ for $n\ge2k$, dividing by
$\binom nk^t$ gives
\[
 q_{t,B}\ \le\ \frac{K!}{h!}\,2^{K-h}\,n^{-(K-h)}
 \ \le\ \left(\frac{2K}n\right)^{K-h},
\]
because $K!/h!$ is a product of $K-h$ factors, each at most $K$.
Finally $K-h=(kt+b)/2$.
\end{proof}

The count permits repeated supports.  In particular
$q_{2,\varnothing}=1/\binom nk$, because two labeled clauses have empty
boundary exactly when their supports coincide; the resulting collision
term $\binom m2\rho^4/\binom nk$ in \cref{eq:projection-norm} is genuine.
On the level curve it has scale $\Theta_k(\ell^{-(k-2)})$, independently
of $\rho$; it can therefore be of constant order at bounded level (and
for every level when $k=2$).  We handle it exactly rather than discard it.

\subsection{The master bound}

Define, with the convention $D^{k/2-1}:=1$ for $k=2$,
\begin{equation}\label{eq:LambdaD}
 \Lambda_D:=e\,(2k)^{k/2}\,
 \frac{\rho^2mD^{k/2-1}}{n^{k/2}} .
\end{equation}

\begin{theorem}[Fixed-sample boundary likelihood bound]\label{thm:master}
Assume $D\le n/(2k)$ and $\Lambda_D\le1/2$.
\begin{enumerate}
\item For every nonempty $B\subseteq[n]$ of fixed size $b$,
\begin{equation}\label{eq:master-B}
 \bigl\|\Pi_{\le D}L_B\bigr\|_2^2
 \ \le\ \frac{\Lambda_D}{1-\Lambda_D}
 \left(\frac{2kD}n\right)^{b/2}
 \ \le\ C_{k,b}\left(\frac Dn\right)^{b/2}.
\end{equation}
\item For the planted mixture,
\begin{equation}\label{eq:master-empty}
 \bigl\|\Pi_{\le D}(L_\varnothing-1)\bigr\|_2^2
 \ \le\ \frac{\Lambda_D^2}{1-\Lambda_D}\ \le\ 2\Lambda_D^2 .
\end{equation}
\end{enumerate}
\end{theorem}
\begin{proof}
By $\binom mt\le(em/t)^t$ and \cref{lem:boundary-count},
\begin{equation}\label{eq:master-term}
 \binom mt\rho^{2t}q_{t,B}
 \ \le\
 \left[e(2k)^{k/2}\frac{\rho^2m\,t^{k/2-1}}{n^{k/2}}\right]^{t}
 \left(\frac{2kt}n\right)^{b/2}
 \ \le\ \Lambda_D^{\,t}\left(\frac{2kD}n\right)^{b/2}
\end{equation}
for $t\le D$, since $t^{k/2-1}$ and $(2kt/n)^{b/2}$ are nondecreasing
in $t$.  For nonempty $B$ the $t=0$ term is absent, and summing the
geometric series over $t\ge1$ in \cref{eq:projection-norm} proves
\cref{eq:master-B}; for $B=\varnothing$, the $t=0$ term is the removed
constant and the $t=1$ term vanishes, so the series starts at $t=2$
(which includes the collision term), giving \cref{eq:master-empty}.
\end{proof}

\subsection{Detection}

For $g\in L^2(Q_m)$ with $\Var_{Q_m}(g)>0$, define the standardized mean
advantage
$\Adv(g):=|\E_{P_m}g-\E_{Q_m}g|/\sqrt{\Var_{Q_m}(g)}$.

\begin{theorem}[Log-free detection lower bound, fixed-sample model]
\label{thm:ld-detection}
Assume $D\le n/(2k)$ and $\Lambda_D\le1/2$.  Then
\begin{equation}\label{eq:ld-detection}
 \sup_{\substack{g\in\cH_D\\\Var_{Q_m}(g)>0}}\Adv(g)
 \ =\ \bigl\|\Pi_{\le D}(L_\varnothing-1)\bigr\|_2
 \ \le\ \sqrt2\,\Lambda_D .
\end{equation}
In particular, if $m=o\bigl(n^{k/2}/(\rho^2D^{k/2-1})\bigr)$ then the
advantage of every sign-degree-$D$ statistic tends to zero, and strong
separation in the low-degree sense is impossible; the same holds, a
fortiori, for every encoding covered by \cref{lem:class-containment}.
\end{theorem}
\begin{proof}
Centering, assume $\E_{Q_m}g=0$.  Then
$\E_{P_m}g=\ip g{L_\varnothing-1}_{L^2(Q_m)}
=\ip g{\Pi_{\le D}(L_\varnothing-1)}$ because $g\in\cH_D$, and
Cauchy--Schwarz gives
$\Adv(g)\le\|\Pi_{\le D}(L_\varnothing-1)\|_2$.  If the
projection is nonzero, equality is attained after normalizing
$g=\Pi_{\le D}(L_\varnothing-1)$; if it vanishes, both sides are zero.
Apply \cref{thm:master}.
\end{proof}

\begin{remark}[Fixed assignment versus mixture]\label{rem:mixture}
The theorem is stated for the uniform planted mixture, the standard
detection likelihood ratio; this loses nothing, since a test that
strongly separates $Q_m$ from $P_{m,x}$ uniformly in $x$ also separates
$Q_m$ from the mixture, contradicting \cref{thm:ld-detection}.
\end{remark}

\begin{remark}[What a low-degree converse is evidence for]
\label{rem:ld-status}
Low-degree lower bounds are not unconditional runtime lower bounds, and
recent work has sharpened what they can and cannot be taken to mean.
Holmgren and Wein~\cite{HolmgrenWein} gave the first counterexamples to
the low-degree conjecture; Buhai, Hsieh, Jain, and Kothari~\cite{BHJK}
refuted its quasi-polynomial form, and Mao~\cite{MaoLD} its
polynomial-time form, in both cases through constructed
permutation-invariant ensembles; and Jia and
Vijayaraghavan~\cite{JiaVijay} exhibit a natural problem the method
mispredicts.  None of these constructions is the planted $k$XOR
ensemble, and none contradicts the theorems of this section, which
concern a specific class of statistics and stand on their own.  In the
other direction, Hsieh, Kane, Kothari, Li, Mohanty, and
Tiegel~\cite{HKKLMT} convert low-degree bounds into rigorous lower
bounds against concrete algorithm classes.  We therefore read
\cref{tab:ledger} as: matched bounds within the low-degree framework,
and evidence for, but not proof of, computational optimality.
\end{remark}

\subsection{Recovery}

\begin{theorem}[Log-free recovery lower bound, fixed-sample model]
\label{thm:ld-recovery}
Let $d\le n/(4k)$ and suppose
$e(2k)^{k/2}\rho^2m(2d)^{k/2-1}/n^{k/2}\le1/2$.  Then every estimator
$f=(f_1,\ldots,f_n)$ with $f_i\in\cH_d$ and $|f_i|\le1$ pointwise obeys
\begin{equation}\label{eq:ld-recovery}
 \E_{x,\,I\sim P_{m,x}}
 \Bigl(\frac1n\sum_{i=1}^nx_if_i(I)\Bigr)^{2}
 \ \le\ \frac1n+C_k\sqrt{\frac dn},
\end{equation}
the expectation being over a uniform hidden assignment and the planted
instance.  Consequently, if $d=o(n)$, no such estimator attains
nonvanishing absolute overlap below
$m\asymp_kn^{k/2}/(\rho^2d^{k/2-1})$.  The statement is
parity-independent and already respects the global sign ambiguity for
even $k$.
\end{theorem}
\begin{proof}
Expand the square.  The diagonal terms contribute at most $1/n$ by
boundedness.  For $i\ne j$, the change of measure gives the exact
identity
\[
 \E_{x,\,P_{m,x}}\bigl[x_ix_jf_if_j\bigr]
 =\E_{Q_m}\bigl[f_if_j\,L_{\{i,j\}}\bigr]
 =\E_{Q_m}\bigl[f_if_j\,\Pi_{\le2d}L_{\{i,j\}}\bigr],
\]
because $f_if_j\in\cH_{2d}$.  Cauchy--Schwarz, the pointwise bound
$\|f_if_j\|_2\le1$, and \cref{thm:master} with $b=2$ give
$|\E_{Q_m}[f_if_jL_{\{i,j\}}]|\le C_k\sqrt{d/n}$.  Averaging the
off-diagonal terms proves \cref{eq:ld-recovery}.
\end{proof}

\begin{corollary}\label{cor:ld-signed}
Let $k$ be odd and let $f_i\in\cH_d$ satisfy the weaker normalization
$\sum_i\|f_i\|_{L^2(Q_m)}^2\le n$.  Under the same condition,
\[
 \Bigl|\frac1n\sum_{i=1}^n
 \E_{x,\,P_{m,x}}[x_if_i(I)]\Bigr|
 \ \le\ C_k\Bigl(\frac dn\Bigr)^{1/4}.
\]
\end{corollary}
\begin{proof}
$\E_{x,P_{m,x}}[x_if_i]=\ip{f_i}{\Pi_{\le d}L_{\{i\}}}$, and
\cref{thm:master} with $b=1$ gives
$\|\Pi_{\le d}L_{\{i\}}\|_2^2\le C_k(d/n)^{1/2}$.  Apply
Cauchy--Schwarz in $L^2(Q_m)$ and then over $i$.
\end{proof}

\subsection{Minimax lower bounds for accurate and exact recovery}
\label{sec:minimax}

The argument in this subsection is standard, and sharp fixed-noise
thresholds are known in closely related parity-measurement models: at
$k=2$ for censored edge measurements~\cite{ABBS,HWX}, and at every
arity in the Bernoulli-hypergraph model of~\cite{ALS}, whose
exact-recovery threshold
$\frac{n\log n}{k(1-\sqrt{1-\rho^2})}(1+o(1))$ has small-bias
expansion $\tfrac2kn\rho^{-2}\log n\,(1+o(1))$; the corresponding
exponential rate at $k=2$ appears in~\cite{XJL}.  We need only the
uniform rate, and we prove it directly in the fixed-$m$,
with-replacement clause model used throughout, as a rate for the
expected error $\err_k$, which quotients the global sign exactly at
even arity.

\begin{theorem}[Minimax total-sample lower bounds for accurate and
exact recovery]\label{thm:minimax}
Even with unlimited computation,
\begin{equation}\label{eq:minimax}
 \inf_{\widehat x}\ \sup_{\xstar\in\{\pm1\}^n}
 \E_{I\sim P_{m,\xstar},\,U}
 \bigl[\err_k\bigl(\widehat x(I,U),\xstar\bigr)\bigr]
 \ \ge\ c_k\exp\bigl(-C_k\rho^2m/n\bigr),
\end{equation}
where $U$ denotes the estimator's internal randomness.  Consequently,
minimax expected error $\vartheta$ requires
$m=\Omega_k(n\rho^{-2}\log(1/\vartheta))$.  Moreover, by a separate
argument, exact recovery with probability $1-o(1)$ requires
$m=\Omega_k(n\rho^{-2}\log n)$.
\end{theorem}
\begin{proof}
We prove Bayes lower bounds under convenient finite priors; each is also
a minimax bound.  The distributions throughout are the fixed-$m$ laws.

\emph{Step 1: removing the sign ambiguity, for even $k$.}  For odd $k$
set $\widetilde x:=\widehat x$ and skip to Step~2.  For even $k$,
write $e:=\min_\sigma d_H(\widehat x,\sigma\xstar)\le n/2$, draw an
anchor $J$ uniformly from $[n]$, reveal $x_J^\star$, and let
$\widetilde x$ be $\widehat x$ multiplied by the sign that makes it
agree with $\xstar$ at $J$: the anchor lies among the $e$ disagreeing
coordinates with probability $e/n$, in which case
$d_H(\widetilde x,\xstar)=n-e$, and otherwise it equals $e$, so
$\E_J\,d_H(\widetilde x,\xstar)=2e(1-e/n)\le2e$ and
\begin{equation}\label{eq:anchor}
 \E\,\err_k(\widehat x,\xstar)
 \ \ge\ \tfrac12\,\E\,\tfrac1nd_H(\widetilde x,\xstar),
\end{equation}
and it suffices to bound the right side, for estimators that may
additionally see $x^\star_J$.

\emph{Step 2: reduction to single coordinates.}  Since
$d_H(\widetilde x,\xstar)=\sum_i\1\{\widetilde x_i\ne x_i^\star\}$, it
suffices to lower bound each
$\bbP\{\widetilde x_i\ne x_i^\star\}$, $i\ne J$, by the corresponding
genie-aided error in which an oracle reveals every coordinate of
$\xstar$ except $x_i^\star$; the oracle can only decrease the error.

\emph{Step 3: the genie-aided test.}  Given all other coordinates,
each clause containing $i$, after dividing out the revealed variables,
is one binary-symmetric observation of $x_i^\star$ with bias $\rho$,
and clauses not containing $i$ are uninformative.  Their number is
$N_i\sim\operatorname{Bin}(m,k/n)$, the one-sample Kullback--Leibler
divergence is
$\kappa(\rho)=\rho\log\frac{1+\rho}{1-\rho}\le C\rho^2$ for
$\rho\le1/2$, and the Bretagnolle--Huber inequality, averaged over
$N_i$, gives
\[
 \bbP\{\widetilde x_i\ne x_i^\star\}
 \ \ge\ \tfrac14\,\E\,e^{-\kappa(\rho)N_i}
 =\tfrac14\Bigl(1-\tfrac kn(1-e^{-\kappa(\rho)})\Bigr)^{m}
 \ \ge\ c_k\,e^{-C_k\rho^2m/n}.
\]
For $\rho>1/2$, the event $N_i=0$ alone gives the same form.
Substituting into Step~2 and \cref{eq:anchor} proves
\cref{eq:minimax}.

For exact recovery with $\rho\le1/2$, consider the $n$ planted
assignments obtained by flipping one coordinate of a fixed reference
(fixing one anchor coordinate when $k$ is even).  The divergence between
any such hypothesis and the reference is at most
$m\cdot\tfrac kn\,\kappa(\rho)\le C_k\rho^2m/n$, so Fano's inequality
over the $n$ hypotheses forces $m=\Omega_k(n\rho^{-2}\log n)$.  For
$\rho>1/2$, the coupon-collector bound shows that
$\Omega_k(n\log n)$ clauses are needed to cover every coordinate, which
is again $\Omega_k(n\rho^{-2}\log n)$ since $\rho^{-2}\le4$.
\end{proof}

\subsection{Refutation converses}\label{sec:ref-converse}

Two different converse notions exist for refutation, and they should not
be conflated.  First, a pointwise sound refuter is in particular a
distinguisher, so the detection bound above yields a log-free,
same-model converse for refuters whose \emph{output} is a bounded
low-sign-degree statistic.

\begin{theorem}[Lower bound for bounded low-degree refuters]
\label{thm:ld-refutation}
Fix $0<\eps\le1/4$ and $D\le n/(2k)$, and let
\begin{equation}\label{eq:ld-ref-threshold}
 m\ \le\ c_k\,\frac{n^{k/2}}{\eps^2D^{k/2-1}}
\end{equation}
for a sufficiently small $c_k$.  Then there is no statistic
$C\in\cH_D$ (randomized seeds allowed) satisfying all of:
$0\le C\le1$ pointwise; soundness
$C(I)\ge\sup_x|V_I(x)|$ on every input; and
$Q_m\{C\le\eps\}\ge1-\delta_n$ with $\delta_n=o(\eps^2)$.
For fixed $\eps$ the last condition is ordinary null success with
probability $1-o(1)$.
\end{theorem}
\begin{proof}
Suppose such a $C$ exists.  First, soundness and null success force
$m\eps^2\to\infty$: for a fixed assignment $x_0$, under $Q_m$ the value
$V_I(x_0)$ is $m^{-1}$ times a sum of $m$ independent Rademacher signs,
so if $\eps\sqrt m$ stayed bounded along a subsequence, Rademacher
anti-concentration would keep $Q_m\{|V_I(x_0)|\le\eps\}$ bounded away
from one, contradicting null success since $C\ge|V_I(x_0)|$.

Now compare $Q_m$ with the planted mixture at bias $\rho:=4\eps$.  Under
$P_{m,x}$,
$V_I(x)=m^{-1}\sum_a\eta_a$ with $\E\eta_a=4\eps$, so Hoeffding's
inequality gives $V_I(x)\ge3\eps$ with probability
$1-e^{-cm\eps^2}=1-o(1)$, uniformly in $x$; by soundness,
$\E_{P_m}C\ge3\eps(1-o(1))$.  Under the null,
$\E_{Q_m}C\le\eps+\delta_n$ and
$\Var_{Q_m}(C)\le\E_{Q_m}C^2\le\eps^2+\delta_n$.  Hence
$\Adv(C)\ge2-o(1)$.  But \cref{thm:ld-detection} at bias $4\eps$ bounds
$\Adv(C)\le C_k\,m\eps^2D^{k/2-1}/n^{k/2}$, which is smaller than $2$
once $c_k$ in \cref{eq:ld-ref-threshold} is small.  This is a
contradiction.
\end{proof}

The theorem is an incompatibility statement: pointwise soundness,
bounded low sign degree, and a small null value cannot coexist below
the curve.  The spectral certificate of \cref{thm:main-refutation} is
not a low-degree polynomial of the input, so \cref{thm:ld-refutation}
is \emph{not} a matching converse for it, and it is not entered as
such in \cref{tab:ledger}.

The second notion, and the community standard for certification, is a
sum-of-squares degree lower bound.  There the relevant converse is that
of Kothari, Mori, O'Donnell, and Witmer~\cite{KMOW}, and two facts about
it deserve emphasis.

\begin{remark}[The canonical sum-of-squares converse, in this model]
\label{rem:kmow}
First, the random-CSP model of~\cite{KMOW}, specialized to the parity
predicate, pushes forward exactly to the law $Q_m$, so the refutation
row of \cref{tab:ledger} is a single-model comparison; what is not
matched is the hierarchy strength.  Second, the $k$XOR predicate
supports a $(k-1)$-wise uniform distribution on satisfying assignments,
so the parameter $t$ in their theorem is $k-1$, and for $k\ge3$ and
clause density $\Delta=m/n$ their bound shows that sum-of-squares
degree up to order $n/(\Delta^{2/(k-2)}\log\Delta)$ cannot even weakly
refute; at the upper curve
$\Delta\asymp_k\eps^{-2}(n/\ell)^{(k-2)/2}$ this reads
\[
 d\ \lesssim_k\ \frac{\ell\,\eps^{4/(k-2)}}{\log\Delta},
\]
which is separated from the degree-$O_k(\ell)$ upper certificate of
\cref{cor:sos} by the logarithmic factor and, at vanishing target
advantage, by the $\eps$ dependence.  The boundary-count estimates of
this section do not close this gap: coefficientwise control of tilted
likelihood ratios does not construct a positive semidefinite
pseudoexpectation.  A log-free sum-of-squares lower bound at
degree $\Theta_k(\ell)$ would make the refutation row of
\cref{tab:ledger} sharp on both sides.
\end{remark}

\subsection{Proof of \texorpdfstring{\cref{thm:main-lower}}{the lower-bound theorem}}

\begin{proof}[Proof of \cref{thm:main-lower}]
Part (1) is \cref{thm:ld-detection} after absorbing the explicit
constant of \cref{eq:LambdaD} into $c_k$ and $C_k$; part (2) is
\cref{thm:ld-recovery}.  The accompanying cleanup converse
(\cref{rem:cleanup-converse}) is \cref{thm:minimax}.
\end{proof}

\section{Quantum detection and recovery}\label{sec:quantum}

The theorems so far are classical.  This section inserts the normalized
operators into the guided sparse-Hamiltonian framework of~\cite{SOKB} and obtains quantum
detection and weak recovery at every arity and fixed positive bias, on
the log-free sample curve, with fourth-root dependence on the Kikuchi
dimension.  There are four steps, one per subsection: import a
clause-product \emph{guide state} whose analysis is parity-free
(\cref{sec:quantum-guide}); observe that the capped odd operator of
\cref{sec:odd} is already suited to a coherent oracle
(\cref{sec:quantum-completion}); establish a two-threshold spectral
promise (\cref{sec:quantum-promise}); and filter, then recover from the
whole accepted spectral band (\cref{sec:quantum-recovery}).  Three
overheads separate the final cost from a clean fourth root: the guide
preparation costs $L^{O_k(L)}$, two condition-number comparisons cost
powers of $\chi_\Gamma\le C_k\log n$, and tomography costs polynomially
many copies; all three are absorbed into the runtime
\cref{eq:quantum-runtime}.

Throughout, $k\ge2$ and $\rho\in(0,1]$ are fixed independently of $n$;
the threshold retains its explicit $\rho^{-2}$ dependence, and all
further polynomial dependence on $1/\rho$ is absorbed into constants.
The level range is narrower than the classical one because the known
guide preparation has an $L^{O_k(L)}$ overhead and requires
$L^2\log n=o(n)$; this implies $L\log n\le n^{2/3}$ for large $n$, so
the classical spectral theorems may be invoked with one fixed value of
their range parameter.  For a requested level $\ell\ge k$, set
\begin{equation}\label{eq:quantum-divisible-level}
        L:=k\left\lfloor\frac\ell k\right\rfloor,
        \qquad N_L:=\binom nL,
        \qquad \mathsf D_{k,L}:=\frac{n^{k/2}}{L^{k/2-1}},
\end{equation}
so that $\ell/2\le L\le\ell$: replacing $\ell$ by $L$ changes the
sample curve by a $k$-dependent factor, never increases the slice
dimension, and puts the guide in the divisible setting $L=ck$.

\subsection{The clause-product guide}\label{sec:quantum-guide}

We first isolate what is imported from~\cite{SOKB}; all theorem and
equation numbers in this subsection refer to arXiv:2406.19378v2.  The
relevant theorem there is stated in the even-arity $k$XOR section, but
its guide calculation and preparation circuit do not use the parity of
$k$.

\begin{lemma}\label{lem:parity-free-guide}
Let $L=ck$ with $c\in\mathbb N$, let $n\ge kL$, and let $\mathcal A$ be the
normalized degree-$k$ planted polynomial of
\cite[Definition~36 and equation~(C6)]{SOKB}.  Write $G$ for its guiding
vector, $\bar z=N_L^{-1/2}((x^\star)^S)_{|S|=L}$, and $\mu_{\rm g}$ for
the mean of one normalized coefficient, whose variance is one.  Assume
\begin{equation}\label{eq:guide-mean-window}
 \frac{L}{\sqrt{\binom nk}}\le\mu_{\rm g}\le\frac1{10L^{k/2}}.
\end{equation}
Then every deterministic unit vector $a$, independent of the guide
polynomial and satisfying $\langle a,\bar z\rangle>0$, obeys
\begin{equation}\label{eq:raw-guide-overlap}
 \langle a,G\rangle
 \ge \frac12\sqrt{\operatorname{Part}_k(L)}\,\mu_{\rm g}^c
       \langle a,\bar z\rangle
\end{equation}
except with probability at most
$C_kc^2/\bigl(\binom nk\,\mu_{\rm g}^2|\langle a,\bar z\rangle|^2\bigr)$;
moreover,
\begin{equation}\label{eq:guide-norm-mean}
 \mathbb E\|G\|^2\le1.0202.
\end{equation}
Suppose in addition that the guide instance is a Poisson pool of mean
$m_{\rm g}$ with
$n\le m_{\rm g}=O(n^{k/2})$ and $L^2\log n=o(n)$.
Except with probability $O(m_{\rm g}^{-1/2})$ over that pool, there is a
quantum circuit of $O(m_{\rm g}L\log n)$ gates which prepares
\begin{equation}\label{eq:raw-guide-preparation}
 \beta_{\rm g}\,\frac{G}{\|G\|}\,|0\rangle+|\perp\rangle|1\rangle,
 \qquad
 \beta_{\rm g}\ge c_kL^{-L/2}.
\end{equation}
The statement holds for both parities of $k$.
\end{lemma}
\begin{proof}
The proof of \cite[Theorem~38]{SOKB} gives the overlap estimate, and
\cite[Lemmas~40 and 41]{SOKB} the norm estimate.  Those arguments,
equations (C7)--(C30), use only independence, the first two coefficient
moments, the Boolean identity $((x^\star)^S)^2=1$, and the overlap
pattern of two partitions into $k$-sets; no step uses the parity of
$k$.  The window \cref{eq:guide-mean-window} and $n\ge kL$ imply all
hypotheses of Lemmas~39--41 there.  The preparation statement is
\cite[Theorem~69]{SOKB}; its circuit, equations (E42)--(E57), prepares
$c$ clause states, rejects intersecting $k$-sets, and sorts their
union, again independently of the parity of $k$.
\end{proof}

Use an independent guide pool of mean
$m_{\rm g}=A_{\rm g}\rho^{-2}\mathsf D_{k,L}$, with
$A_{\rm g}=A_{\rm g}(k)$ a fixed constant, so that, with
$\mathcal B_k:=\binom nk$,
\begin{equation}\label{eq:guide-mu-scale}
 \mathcal B_k\mu_{\rm g}^2=\rho^2m_{\rm g}
 =A_{\rm g}\mathsf D_{k,L}.
\end{equation}
The hypotheses of \cref{lem:parity-free-guide} are then automatic in
the quantum range: $\mathsf D_{k,L}\ge n$ and $L^2\log n=o(n)$ give
$\rho^2m_{\rm g}\ge L^2$, which is the lower bound in
\cref{eq:guide-mean-window};
$\mu_{\rm g}^2L^k\le C_kL^{k/2+1}/n^{k/2}=o(1)$ gives the upper bound;
and $n\le m_{\rm g}\le C_{k,\rho}n^{k/2}$ and $n\ge kL$ hold as well.
Finally, $\operatorname{Part}_k(L)=L!/((k!)^cc!)$ and
$\binom nk\le n^k/k!$, followed by Stirling's formula and
$\binom nL\ge(n/L)^L$, give
\begin{equation}\label{eq:guide-combinatorial-overlap}
 \operatorname{Part}_k(L)\mu_{\rm g}^{2c}
 \ge \frac{L!}{c!}n^{-L/2}L^{(1-k/2)L/k}\exp(-C_kL)
 \ge N_L^{-1/2}\exp(-C_kL).
\end{equation}

\subsection{The odd operator}\label{sec:quantum-completion}

For even $k$, the quantum Hamiltonian is the operator $(K,\Gamma)$ of
\cref{eq:even-operator} at level $L$; at odd arity it is
$(H_{\rm o},\Gamma_{\rm o})$ of \cref{eq:odd-H}, unchanged, at level
$L$---each root contributes at most one to the capped diagonal, which
is insertion-nondecreasing and computable cell by cell, so no separate
quantum construction is needed.  The next proposition collects the
facts used by the guided filter; all were proved in \cref{sec:odd}.

\begin{proposition}\label{prop:quantum-odd-normalization}
Fix odd $k=2r+1$ and fixed $\rho>0$.  At level $L$ and with a main
pool of mean $A_k\rho^{-2}\mathsf D_{k,L}$, choose $A_k$ sufficiently
large.  Then, for all sufficiently large $n$, the following hold.
\begin{enumerate}
\item Deterministically,
\begin{equation}\label{eq:quantum-odd-domination}
 \|H_{\rm o}\|\le1,
 \qquad
 \left\|\Gamma_{\rm o}^{-1/2}
 \cP_{\rm o}\,\Gamma_{\rm o}^{-1/2}\right\|\le1.
\end{equation}
\item Under either the null or planted support law, with probability
$1-o(1)$,
\begin{equation}\label{eq:quantum-odd-diagonal-bounds}
 \Tr\Gamma_{\rm o}\le C_kdN_L,
 \qquad
 \max_S\Gamma_{\rm o}(S,S)
 \le C_k(d+\log N_L).
\end{equation}
\item Under the null, for every fixed $B>0$,
\begin{equation}\label{eq:quantum-odd-null}
 \Pr\left\{\|H_{\rm o}\|>
 C_{k,B}\sqrt{L/d}\right\}\le n^{-BL}.
\end{equation}
Under planting, with probability $1-o(1)$,
\begin{equation}\label{eq:quantum-odd-comparison}
 \left\|\Gamma_{\rm o}^{-1/2}
 \bigl(Z\cC_{\rm o}Z-\rho^2\cP_{\rm o}\bigr)
 \Gamma_{\rm o}^{-1/2}\right\|
 \le C_k\sqrt{L/d}.
\end{equation}
\end{enumerate}
\end{proposition}
\begin{proof}
Part (1) is \cref{lem:odd-uniqueness}(3).  Part (2) is
\cref{eq:odd-raw-trace,eq:odd-max-diagonal}: their proofs use only the
candidate process and the real marks, whose joint law is the same under
the null and planted support laws.  Part (3) is
\cref{prop:odd-null,prop:odd-comparison} at level $L$.  The quantum
range implies $L\log n\le n^{2/3}$ for large $n$, so the trace theorems
apply with one fixed range parameter, and the odd cell sparsity
$\alpha=\Theta_k(\rho^{-2}\sqrt{L/n})$ tends to zero, as required.
\end{proof}

\subsection{The spectral promise and the weighted guide}
\label{sec:quantum-promise}

In the even case set $(H,\Gamma)=(K,\Gamma)$ from
\cref{eq:even-operator}; in the odd case
$(H,\Gamma)=(H_{\rm o},\Gamma_{\rm o})$ from \cref{eq:odd-H}.  Write
$\kappa_k:=1$ for even $k$ and $\kappa_k:=2$ for odd $k$.  After
increasing the main-pool sample constant, the even null and planted
comparison bounds, \cref{prop:quantum-odd-normalization}, and
\cref{prop:odd-population} give constants
$0<a_0<b_-<b_+<a_1<A_2$ such that, with probability $1-o(1)$,
\begin{equation}\label{eq:quantum-promise}
\begin{gathered}
 \|H\|\le a_0\rho^{\kappa_k}\quad\text{under the null};
 \qquad
 \langle p,Hp\rangle\ge a_1\rho^{\kappa_k},
 \quad
 \|H\|\le A_2\rho^{\kappa_k}
 \quad\text{under planting},\\
 \text{where}\quad
 p:=\Gamma^{1/2}\bar z/\|\Gamma^{1/2}\bar z\| .
\end{gathered}
\end{equation}
Indeed, at odd arity
$N_L^{-1}\1^{\mathsf T}\cP_{\rm o}\1\ge c_kd$ with high probability
while $N_L^{-1}\1^{\mathsf T}\Gamma_{\rm o}\1\le C_kd$ by
\cref{eq:quantum-odd-diagonal-bounds}, so testing the gauged matrix on
$\Gamma_{\rm o}^{1/2}\1$ and using
\cref{eq:quantum-odd-comparison,eq:quantum-odd-domination,eq:odd-pop-conc}
gives the planted witness, the planted norm bound, and, for every
high-Rayleigh vector after unweighting, the same constant normalized
$J_{2r}$-energy as in \cref{prop:odd-recovery}; the even case is the
corresponding argument from
\cref{prop:even-comparison,prop:even-recovery}, at signal scale $\rho$.
Choose $b_-$ above the uniform comparison-error threshold of the
recovery proofs and put
\begin{equation}\label{eq:two-cutoffs}
 \lambda_\pm:=b_\pm\rho^{\kappa_k},
 \qquad
 \Pi_\pm:=\1_{[\lambda_\pm,\infty)}(H),
 \qquad
 \|\Pi_+p\|^2\ge\frac{a_1-b_+}{A_2-b_+}\ge c_k,
\end{equation}
the last bound from \cref{eq:quantum-promise}; the relative gap between
the two cutoffs is constant.

\begin{lemma}\label{lem:quantum-condition-number}
Let $d=\bd$ in the even case and $d=n\alpha^2$ in the odd case, and put
$B_\Gamma:=C_k(d+\log N_L)$.  With probability $1-o(1)$ under either
model,
\begin{equation}\label{eq:quantum-condition-number}
 d\le\Gamma(S,S)\le B_\Gamma\quad\text{for every }S,
 \qquad
 \chi_\Gamma:=\frac{B_\Gamma}d\le C_k\log n.
\end{equation}
\end{lemma}
\begin{proof}
For even $k$, a row degree has mean $d$; scalar Chernoff and a union
bound over the $N_L$ rows give the maximum.  For odd $k$, this is
\cref{eq:quantum-odd-diagonal-bounds}.  At the threshold
$d\asymp_{k,\rho}L$ at either parity and $\log N_L\le L\log n$, which
gives the ratio bound.
\end{proof}

\begin{lemma}\label{lem:quantum-oracles}
Conditioned on the polynomial-size main and auxiliary clause lists, $H$
has sparsity $n^{O_k(1)}$, and its matrix-entry oracle, its
nonzero-location oracle, and the inverse location map can be
implemented reversibly with $n^{O_k(1)}$ gates and
$\widetilde O(L\log n)$ coherent workspace, without materializing a
vector or matrix of length $N_L$.
\end{lemma}
\begin{proof}
Encode a valid row by its sorted $L$-tuple in $O(L\log n)$ qubits and
extend $H$ by zero on invalid encodings.  A row query streams the
clause list (even $k$), or the root--row cells of the single candidate
list (odd $k$, where mutual exact-twoness leaves at most one retained
edge per root, hence at most $n$ moves, and \cref{eq:odd-completed}
gives the exact capped diagonal); the $j$th distinct neighbor is
selected lexicographically by repeated streaming and counting, so
deduplication and the inverse map need only counters and the current
$L$-set.  All lists have polynomial size with probability $1-o(1)$, failure
being declared otherwise.  Reversible fixed-precision arithmetic
supplies $C(S,T)/\sqrt{\Gamma(S,S)\Gamma(T,T)}$, with missing
adjacency positions padded by flagged dummy indices; since
$\|H\|\le1$, entry precision inverse-polynomial in the total query
bound, i.e.\ $\widetilde O(L\log n)$ bits, suffices.
\end{proof}

The oracle always uses the exact data-dependent diagonal; the known
cap $B_\Gamma$ enters only the clipped guide weighting below, well
defined on every input.  Let $G$ be
the independent raw guide and $g:=G/\|G\|$, and prepare
\begin{equation}\label{eq:weighted-guide}
 g_\Gamma:=\frac{\Gamma^{1/2}g}{\|\Gamma^{1/2}g\|}
\end{equation}
by a controlled rotation of amplitude
$\sqrt{\min\{\Gamma(S,S),B_\Gamma\}/B_\Gamma}$; on the condition-number
event this succeeds with probability at least $d/B_\Gamma$ and prepares
exactly $g_\Gamma$.  Condition on the main pool, so that
$H,\Gamma,\Pi_+$ are deterministic and independent of $G$, and let
$v:=\Pi_+p/\|\Pi_+p\|$ and $a:=\Gamma^{1/2}v/\|\Gamma^{1/2}v\|$, the
phase of $v$ fixed so that $\langle v,p\rangle>0$.  Two elementary
condition-number comparisons give
\begin{equation}\label{eq:two-condition-number-losses}
 |\langle a,\bar z\rangle|^2
 \ge\chi_\Gamma^{-1}\|\Pi_+p\|^2,
 \qquad
 |\langle v,g_\Gamma\rangle|^2
 \ge\chi_\Gamma^{-1}|\langle a,g\rangle|^2.
\end{equation}
By \cref{eq:guide-norm-mean}, the event $\|G\|^2\le C\log n$ fails with
probability $O(1/\log n)$; on it, \cref{lem:parity-free-guide} applied
to $a$ yields
\begin{equation}\label{eq:ideal-weighted-guide-overlap}
 \|\Pi_+g_\Gamma\|^2
 \ge
 \frac{\operatorname{Part}_k(L)\mu_{\rm g}^{2c}}
 {C\chi_\Gamma^2\log n}\,\|\Pi_+p\|^2,
\end{equation}
with failure probability
$O_k(L^2\chi_\Gamma/\mathsf D_{k,L})=o(1)$, because
$\mathsf D_{k,L}\ge n$ and $L^2\log n=o(n)$.  Keep the heralding flag
from \cref{eq:raw-guide-preparation} and the weighting flag, and extend
the Hamiltonian by $H\otimes|00\rangle\langle00|$.  Combining
\cref{eq:guide-combinatorial-overlap,eq:ideal-weighted-guide-overlap}
with $\beta_{\rm g}^2\ge c_kL^{-L}$ and the weighting success
probability $d/B_\Gamma$ gives a coherent guide whose squared
high-energy overlap satisfies
\begin{equation}\label{eq:prepared-guide-overlap}
 \gamma^2
 \ge N_L^{-1/2}L^{-O_k(L)}(\log n)^{-O_k(1)}.
\end{equation}

\begin{lemma}\label{lem:quantum-band-filter}
Let $H$ be an efficiently simulable sparse Hamiltonian and let
$0<\lambda_-<\lambda_+$ have constant relative separation.  Suppose a
preparable normalized state $|\Psi\rangle$ satisfies
$\|\1_{[\lambda_+,\infty)}(H)|\Psi\rangle\|\ge\gamma$.  For every
$\eta>0$, the guided sparse-Hamiltonian procedure can be run so that it
distinguishes this case from $\|H\|<\lambda_-$ and, in the former case,
retains an accepted density operator $\sigma$ satisfying
$\Tr(\1_{[\lambda_-,\infty)}(H)\sigma)\ge1-\eta$.  Relative to
\cite[Theorem~63]{SOKB}, the complexity changes only by polylogarithmic
factors in $1/\eta$.
\end{lemma}
\begin{proof}
Run phase estimation with an acceptance window whose transition lies
strictly inside $[\lambda_-,\lambda_+]$, at precision a constant
multiple of $\lambda_+-\lambda_-$ and failure probability at most
$c\eta\gamma^2$: the component above $\lambda_+$ is accepted with
constant probability, the false-accept probability below $\lambda_-$ is
$O(\eta\gamma^2)$, and amplitude amplification costs
$\widetilde O(1/\gamma)$, leaving relative weight $O(\eta)$ below the
lower cutoff after acceptance.  This is the proof of
\cite[Theorem~63]{SOKB}, with the accepted branch retained rather than
only measured.
\end{proof}

\subsection{Recovery from the accepted band}\label{sec:quantum-recovery}

The guided circuit is retained coherently rather than used only as a
decision procedure, so we need a statement uniform over the whole
accepted spectral band.

\begin{proposition}\label{prop:quantum-cutoff-recovery}
On the high-probability planted event, let $\sigma$ be any density operator
supported on $\operatorname{Ran}\Pi_-$.  Apply the trace-nonincreasing
filter with Kraus operator $\sqrt d\,\Gamma^{-1/2}$ and condition on
success; equivalently, the normalized post-filter state is
\begin{equation}\label{eq:quantum-filtered-state}
 \widetilde\sigma
 :=\frac{\Gamma^{-1/2}\sigma\Gamma^{-1/2}}
 {\Tr(\Gamma^{-1}\sigma)}.
\end{equation}
Then apply the isometry
$|S\rangle\mapsto L^{-1/2}\sum_{i\in S}|i\rangle|S\setminus\{i\}\rangle$
and trace out the second register.  The resulting one-particle state
$\varrho$ satisfies
\begin{equation}\label{eq:quantum-rdm-mass}
 \left\langle\frac{\xstar}{\sqrt n},
 \varrho\frac{\xstar}{\sqrt n}\right\rangle\ge c_k>0.
\end{equation}
For even $k$, the conclusion is understood modulo the unavoidable global
sign.
\end{proposition}
\begin{proof}
Every unit vector $v\in\operatorname{Ran}\Pi_-$ has Rayleigh quotient
at least $b_-\rho^{\kappa_k}$, so the uniform planted comparison shows
that, after applying $\Gamma^{-1/2}$ and the gauge in the analysis, the
normalized vector has constant positive energy in $J_{k/2}$ (even $k$),
respectively $J_{k-1}$ (odd $k$).  The proof of \cref{thm:transfer} is
sesquilinear and extends verbatim to complex vectors, with transposes
replaced by adjoints, giving \cref{eq:quantum-rdm-mass} for every such
$v$.  Diagonalizing $\sigma=\sum_jq_j|v_j\rangle\langle v_j|$ inside
$\operatorname{Ran}\Pi_-$, the post-filter state is the convex mixture
of the normalized $\Gamma^{-1/2}v_j$ with weights proportional to
$q_j\|\Gamma^{-1/2}v_j\|^2$; the one-particle channel is linear, so the
planted-mass bound survives the mixture.  The gauge commutes with the
diagonal filter and is used only in the analysis.
\end{proof}

The filter is implemented by a controlled rotation of amplitude
$\sqrt{d/\Gamma(S,S)}$ and succeeds on every state with probability at
least $d/B_\Gamma$.  To handle the coherent leakage allowed by
\cref{lem:quantum-band-filter}, use the accuracy
\begin{equation}\label{eq:quantum-leakage-choice}
        \eta:=\frac1{\log^2 n}\left(\frac d{B_\Gamma}\right)^2 .
\end{equation}
Indeed, if $\Tr(\Pi_-\sigma)\ge1-\eta$ and
$\sigma_-:=\Pi_-\sigma\Pi_-/\Tr(\Pi_-\sigma)$, the gentle-measurement
lemma gives $\|\sigma-\sigma_-\|_1\le3\sqrt\eta$; the
trace-nonincreasing map $X\mapsto d\,\Gamma^{-1/2}X\Gamma^{-1/2}$
contracts trace norm while retaining mass at least $d/B_\Gamma$ on
$\sigma_-$, so the two normalized post-filter outputs are within trace
distance $O((B_\Gamma/d)\sqrt\eta)=o(1)$, and
\cref{eq:quantum-rdm-mass} survives with a smaller constant.

Standard tomography then recovers the $n$-dimensional one-particle
state: estimating every entry of $\varrho$ to accuracy $n^{-2}$, by
repeated measurement in the standard basis and in the $O(n^2)$
two-dimensional rotated bases revealing the real and imaginary parts of
one off-diagonal entry, controls the operator norm at scale $1/n$;
Hoeffding's inequality and a union bound over entries give a Hermitian
estimate $\widehat\varrho$ with
$\|\widehat\varrho-\varrho\|\le\eta_{\rm tom}=o(1)$ from
$n^{O(1)}\log n$ accepted copies, with failure probability $n^{-1}$;
all of this is absorbed into \cref{eq:quantum-runtime}.

Since $x^\star$ is real, pass to real parts:
$R:=\operatorname{Re}(\varrho)$ is real positive semidefinite with
unit trace and $\bar x^{\mathsf T}R\bar x\ge c_k$ by
\cref{eq:quantum-rdm-mass}, while
$\widehat R:=\operatorname{Re}(\widehat\varrho)$ satisfies
$\|\widehat R-R\|\le\eta_{\rm tom}$.  With $\alpha_0:=c_k/4$ and
$\widehat\Pi:=\1_{[\alpha_0,\infty)}(\widehat R)$, Weyl's inequality
gives $\rank\widehat\Pi\le8/c_k$, and
$c_k-\eta_{\rm tom}\le\bar x^{\mathsf T}\widehat R\bar x
\le(1+\eta_{\rm tom})\|\widehat\Pi\bar x\|^2
+\alpha_0(1-\|\widehat\Pi\bar x\|^2)$ forces
$\|\widehat\Pi\bar x\|^2=\Omega_k(1)$.  Constant rank and constant
planted overlap are exactly the inputs of \cref{lem:rounding}, applied
to $\Ran\widehat\Pi$; an independent validation pool then gives a
Boolean vector of constant overlap.

\begin{theorem}[Log-free quantum detection and weak recovery]
\label{thm:quantum-inference}
Fix $k\ge2$ and a bias $\rho\in(0,1]$, independently of $n$.
Let $\ell\ge k$, let $L$ be as in
\cref{eq:quantum-divisible-level}, and assume
\begin{equation}\label{eq:quantum-range}
 L^2\log n=o(n).
\end{equation}
For a sufficiently large constant $C_k$,
\begin{equation}\label{eq:quantum-sample}
 m\ge C_k\rho^{-2}
       \frac{n^{k/2}}{\ell^{k/2-1}}
\end{equation}
suffices for quantum strong detection and weak recovery, with success
probability $1-o(1)$ uniformly in the planted assignment.  For even $k$,
recovery is understood modulo the unavoidable global sign.  In the reversible
read-only input model, the gate complexity is
\begin{equation}\label{eq:quantum-runtime}
 \widetilde O\!\left(
 N_L^{1/4}L^{O_k(L)}n^{O_k(1)}
 \right),
\end{equation}
and the coherent workspace is $\widetilde O(L\log n)$ qubits, excluding
polynomial-size read-only storage for the input and auxiliary classical
data.
\end{theorem}
\begin{proof}
Generate independent main, guide, and validation pools by
\cref{lem:subsample}.  The range condition implies
$L\log n\le n^{2/3}$, $L/n=o(1)$, and, at fixed $\rho$, vanishing odd
cell sparsity, so the classical main-pool promises are automatic; the
guide hypotheses were verified after \cref{eq:guide-mu-scale}.  Apply
\cref{lem:quantum-band-filter} to the extended Hamiltonian, with
no-instance threshold $a_0\rho^{\kappa_k}$ and yes cutoff $\lambda_+$:
under the null the norm promise forces rejection, while under planting
the coherent guide has high-energy amplitude at least $\gamma$ from
\cref{eq:prepared-guide-overlap}; since the procedure uses the
amplitude $\gamma$, not its square, $1/\gamma$ contributes
$N_L^{1/4}L^{O_k(L)}$ up to polylogarithmic factors, and the oracle
costs, inverse cutoff, guide circuit, and precision factors are
polynomial and absorbed into \cref{eq:quantum-runtime}.  In the
planted case, retain the accepted state (the positive-energy band lies
in the joint herald sector $|00\rangle$, whose projection joins the
postselection), run the filter with thresholds $\lambda_-<\lambda_+$
and the error \cref{eq:quantum-leakage-choice}, prepare the
$n^{O(1)}\log n$ accepted copies used by tomography, and apply the
one-particle channel,
\cref{prop:quantum-cutoff-recovery,lem:rounding,lem:validation}.
\end{proof}

\begin{remark}[Scope of the quantum speedup]\label{rem:quantum-scope}
The algorithm has fourth-root dependence on the explicit Kikuchi
dimension $N_L$, and is genuinely $N_L^{1/4+o(1)}n^{O_k(1)}$ when
$\log L=o(\log(n/L))$, for example when $L=n^{o(1)}$; for polynomially
large $L$ the factor $L^{O_k(L)}$ may rival the fourth-root term, so we
do not describe that range as a quartic speedup.  Nor is this a quartic
separation from every classical architecture: for large constant arity,
\cite{GHOS} give a classical quadratic speedup in part of the range.
We make no quantum refutation claim, since bounded-error phase
estimation is not a zero-error one-sided certificate.
\end{remark}

\section{Boolean CSPs in the random-literal model}\label{sec:boolean-csp}

The upper bounds of this paper are stated for $k$XOR, but fixed Boolean
predicates in the uniform random-scope/random-literal model admit an
exact Fourier reduction: every nonzero Fourier character of a known
local planting law is, exactly, a noisy-XOR channel.  The transfer is
channelwise and carries over the classical inference and two-sided
refutation \emph{upper} bounds, and, for channels of degree at least
two, the quantum detection and weak-recovery upper bound under the
hypotheses of \cref{thm:quantum-inference}; degree-one channels are
handled directly.  It does not transfer the $k$XOR \emph{lower} bounds
to procedures that use the complete local output jointly, and it makes
no quantum-refutation claim.

\paragraph{Related work on general CSPs.}
Concurrent work of Chan, d'Orsi, and Xu~\cite{ChanDOrsiXu} extends
sum-of-squares and spectral refutation to general random CSPs without
uniformly random literals through odd and asymmetric Kikuchi operators;
their model is substantially broader than the random-literal reduction
here, whose point is the log-free transfer of the normalized $k$XOR
bounds.

Fix a constant arity $K$.  A local planting law is a known distribution
$\cW$ on $\{\pm1\}^K$.  A planted constraint chooses an ordered $K$-tuple
of distinct variables uniformly and reveals
\begin{equation}\label{eq:csp-observation}
 Y_j=x_{i_j}^\star W_j,
 \qquad W=(W_1,\ldots,W_K)\sim\cW.
\end{equation}
Under the null, the revealed literal vector is uniform on $\{\pm1\}^K$.
For nonempty $S\subseteq[K]$, define
\begin{equation}\label{eq:csp-fourier-bias}
 \mu_S:=\E_{W\sim\cW}\prod_{j\in S}W_j.
\end{equation}

\begin{lemma}\label{lem:exact-csp-channel}
If $\mu_S\ne0$, then
\begin{equation}\label{eq:csp-derived-label}
 b_a^S:=\sign(\mu_S)\prod_{j\in S}Y_{aj}
\end{equation}
is exactly a planted noisy $|S|$XOR label of bias $|\mu_S|$ on support
$\{i_{aj}:j\in S\}$.  Under the null it is exactly an unbiased XOR label.
\end{lemma}
\begin{proof}
The identity
\[
 \prod_{j\in S}Y_{aj}
 =\left(\prod_{j\in S}W_{aj}\right)
  \left(\prod_{j\in S}x_{i_{aj}}^\star\right)
\]
is exact.  Restricting a uniform ordered tuple of distinct variables to
the positions in $S$ leaves a uniform ordered $|S|$-tuple, and under the
null the product of a nonempty subset of uniform independent literals is
unbiased.
\end{proof}

\begin{lemma}\label{lem:degree-one-channel}
Let $(I_a,Y_a)_{a=1}^m$ be the derived degree-one channel of fixed bias
$\beta>0$, so $I_a$ is uniform on $[n]$ and under planting
$Y_a=x_{I_a}^\star\eta_a$ with $\E\eta_a=\beta$.
\begin{enumerate}
\item The collision statistic
\[
 T:=\sum_{1\le a<b\le m}\1\{I_a=I_b\}Y_aY_b
\]
has null mean zero and variance $\binom m2/n$, while under planting
\[
 \E T=\beta^2\binom m2/n,
 \qquad
 \Var T\le C\left(m^2/n+\beta^2m^3/n^2\right).
\]
Consequently $m=\Theta(\beta^{-2}\sqrt n)$ gives constant statistical
advantage, and $\beta^2m/\sqrt n\to\infty$ gives strong detection.
\item The coordinatewise majority estimator
\[
 \widehat x_i:=\sign\!\left(\sum_{a:I_a=i}Y_a\right)
\]
(with fair tie breaking) has constant overlap with $x^\star$ with
probability $1-o(1)$ once $m\ge C\beta^{-2}n$.
\item On every input,
\[
 \sup_{x\in\{\pm1\}^n}
 \left|\frac1m\sum_aY_ax_{I_a}\right|
 =\frac1m\sum_{i=1}^n
 \left|\sum_{a:I_a=i}Y_a\right|.
\]
Under the null, this quantity is $O(\sqrt{n/m})$ with probability
$1-o(1)$.
\end{enumerate}
\end{lemma}
\begin{proof}
For (1), distinct pair summands have zero null covariance, even when
the index pairs share one occurrence, giving the exact variance; under
planting, a collision cancels the hidden spin, disjoint pair summands
are independent, and two pairs sharing an occurrence have covariance at
most $\beta^2/n^2$, giving the variance bound.  Chebyshev proves both
conclusions.  For (2), conditional on the occupancy
$N_i:=|\{a:I_a=i\}|$ the votes at $i$ are independent signs of mean
$\beta x_i^\star$, so
$\Pr\{\widehat x_i\ne x_i^\star\}\le\E e^{-c\beta^2N_i}
\le\exp(-c'\beta^2m/n)$; for larger inputs, run on a deterministic
prefix of $\lceil C\beta^{-2}n\rceil$ samples with $C$ so large that
the expected error fraction is below a constant less than $1/2$, and
bounded differences (exponent $\Omega_\beta(n)$) upgrades this to
constant overlap with probability $1-o(1)$.  For (3), optimize each
coordinate sign separately; under the null,
$\E|\sum_{a:I_a=i}Y_a|\le\sqrt{N_i}$, so Cauchy--Schwarz bounds the
expectation by $\sqrt{n/m}$, and bounded differences (one observation
moves the certificate by at most $2/m$) gives the high-probability
bound.
\end{proof}

\subsection{Inference for fixed known nonuniform planting laws}

The Walsh characters form a basis, so $\cW$ is nonuniform if and only
if some nonconstant $\mu_S$ is nonzero, and for a known law the
constantly many channels may be enumerated.  For every channel $S$ of
degree $s:=|S|\ge2$, \cref{thm:main-inference} applies at
\begin{equation}\label{eq:csp-inference-curve}
 m\ge C_{K,\delta}|\mu_S|^{-2}
       \frac{n^{s/2}}{\ell^{s/2-1}},
\end{equation}
and no additional odd-channel hypothesis is needed: the local law is
fixed, so every nonzero $|\mu_S|$ is a fixed constant and the relevant
cell-sparsity condition is automatic from \cref{eq:range}.  The quantum
theorem applies with $(k,\rho)$ replaced
by $(s,|\mu_S|)$ and with divisible level $s\lfloor\ell/s\rfloor$,
provided its quantum-range assumptions hold.  One may choose the
informative channel minimizing the resulting sample bound.

Degree-one channels are handled by \cref{lem:degree-one-channel}:
coordinate collisions give constant statistical advantage at
$m=\Theta(|\mu_S|^{-2}\sqrt n)$ and strong detection when
\begin{equation}\label{eq:degree-one-strong-detection}
 |\mu_S|^2m/\sqrt n\longrightarrow\infty,
\end{equation}
while coordinatewise aggregation gives weak recovery at
$m\ge C|\mu_S|^{-2}n$.  An odd-degree channel fixes the orientation
after held-out validation, and can also orient an estimate obtained
from an even channel; if all odd Fourier coefficients vanish, then
$\cW(w)=\cW(-w)$ and absolute orientation is not identifiable.  These
conclusions are channelwise: they do not assert optimality among
algorithms that use the complete local output jointly.

\subsection{Refutation for fixed predicates}

For refutation, let $P:\{\pm1\}^K\to\mathbb R$ be a fixed predicate or
bounded objective with Fourier expansion
\begin{equation}\label{eq:predicate-fourier}
 P(z)=\widehat P(\varnothing)
 +\sum_{\varnothing\ne S\subseteq[K]}
   \widehat P(S)\prod_{j\in S}z_j.
\end{equation}
Define
\[
 \val_I(x):=\frac1m\sum_{a=1}^m
 P(Y_{a1}x_{i_{a1}},\ldots,Y_{aK}x_{i_{aK}}),
 \qquad
 V_S(x):=\frac1m\sum_{a=1}^m
 \prod_{j\in S}Y_{aj}x_{i_{aj}}.
\]
For every input, the pointwise Fourier identity gives
\begin{equation}\label{eq:csp-certificate}
 \sup_x\left|\val_I(x)-\widehat P(\varnothing)\right|
 \le\sum_{\varnothing\ne S\subseteq[K]}
       |\widehat P(S)|\,\Cert_S(I),
\end{equation}
where $\Cert_S$ is a pointwise-sound two-sided certificate for the
corresponding XOR channel, supplied by \cref{thm:main-refutation} for
$|S|\ge2$.  At odd degree it may use auxiliary rooting, but soundness holds
for every auxiliary seed.  For a degree-one channel $S=\{j\}$ one may take
the exact certificate
\begin{equation}\label{eq:degree-one-certificate}
 \Cert_{\{j\}}(I)
 :=\frac1m\sum_{i=1}^n
 \left|\sum_{a:i_{aj}=i}Y_{aj}\right|
 =\sup_{x\in\{\pm1\}^n}|V_{\{j\}}(x)|.
\end{equation}
Under the random-literal null it is $O(\sqrt{n/m})$ with high
probability.  Correlations between channels are irrelevant to soundness,
and a union bound suffices for null tightness because $K$ is fixed.

The pointwise bound \cref{eq:csp-certificate} also transfers the
sum-of-squares proofs.  With
$\Cert_P(I):=\sum_{\varnothing\ne S\subseteq[K]}
|\widehat P(S)|\,\Cert_S(I)$, we claim that for each
$\tau\in\{\pm1\}$ the polynomial
$\Cert_P(I)-\tau\bigl(\val_I(x)-\widehat P(\varnothing)\bigr)$ has a
degree-$O_K(\ell)$ sum-of-squares proof modulo the ideal $x_i^2=1$; the
pointwise bound \cref{eq:csp-certificate}, and with it the supremum
bound, then follow semantically.  Each summand
$|\widehat P(S)|\,\bigl(\Cert_S-\tau\sign(\widehat P(S))V_S\bigr)$ with
$|S|\ge2$ is a nonnegative multiple of a one-sided channel bound,
certified in degree $O_K(\ell)$ by \cref{cor:sos}.  For a degree-one
channel $S=\{j\}$, write $a_i:=\sum_{a:i_{aj}=i}Y_{aj}$; modulo the
ideal, for $\sigma\in\{\pm1\}$,
\[
 |a_i|-\sigma a_ix_i
 =\frac{|a_i|}2\bigl(x_i-\sigma\sign(a_i)\bigr)^2,
\]
with either sign chosen when $a_i=0$, so summing over $i$ certifies
$m\,\Cert_{\{j\}}\mp m\,V_{\{j\}}\ge0$ in degree two.  If
$D_P:=\deg P=0$, the centered objective vanishes identically.  For a
Boolean predicate the Fourier coefficients are rational, so the
exact-evaluation and exact-SoS statements of \cref{sec:evaluation}
apply verbatim; for a general bounded objective, replace the
$\widehat P(S)$ by rational approximations and add the total
approximation error to the certificate as explicit slack, which keeps
it sound on every input since $|V_S(x)|\le1$ pointwise.

Let $D_P\ge2$, assume a common level $\ell\ge K$, and allocate a
predicate-dependent constant fraction of the target error to each
nonzero channel.  Since $\mathsf D_{s,\ell}=n(n/\ell)^{s/2-1}$ is
nondecreasing in $s\ge2$ and the direct degree-one certificate has
scale $\eps^{-2}n$, one obtains
\begin{equation}\label{eq:csp-refutation-curve}
 m\ge C_{P,\delta}\eps^{-2}
       \frac{n^{D_P/2}}{\ell^{D_P/2-1}}.
\end{equation}
For fixed $\eps>0$, this gives a pointwise-sound two-sided strong
refutation of every fixed predicate, with a sum-of-squares proof of
degree $O_K(\ell)$ and no additional asymptotic condition at odd
degree.  If $D_P=1$, the exact linear certificate has scale
$m\ge C_P\eps^{-2}n$.  If Fourier biases or allocated target
tolerances are allowed to vary with $n$, the quantitative extension of
\cref{rem:quantitative-ranges} applies channelwise.

In summary, the extension covers every fixed known nonuniform local
planting law (orientation identifiable exactly when some odd Fourier
channel is nonzero) and every fixed predicate in the uniform
random-scope/random-literal model; it does not cover adversarial
scopes, semirandom literals, or an unknown planting law without a
quantitative promise on an informative Fourier coefficient.

\section{Discussion and open problems}\label{sec:discussion}
\paragraph{A Kikuchi-type Ihara Bass formula}
The name of the Kikuchi method originates from~\cite{WEM}, motivated by the hierarchy of
Kikuchi approximations to free energy in statistical physics~\cite{Kikuchi1951}.  As the authors
of~\cite{WEM} note, their derivation expands about an inconsistent paramagnetic
point, so the connection between Kikuchi matrices and the Hessian of Kikuchi
free energy holds only to first order.  The Kikuchi matrix itself has a more direct physical interpretation
via mean-field theory; see~\cite{Hastings2019} and Appendix~B
of~\cite{SOKB}. That said, a Hessian of Kikuchi free energy can be principally derived, which would give a generalization of the Kikuchi matrix. It is an interesting question whether such a construction, as opposed to the normalization given here, also recovers the log-free SNR-runtime trade-off curve. This further motivates the notion of a generalized Ihara-Bass formula that connects the Kikuchi Hessian to a non-backtracking operator. The reference~\cite{dOrsiTrevisan} studies the Ihara-Bass formula in the context of random CSPs but stops at Bethe-level corrections and does not explore the higher Kikuchi-type corrections.

\paragraph{A log-free sum-of-squares lower bound.}
The refutation converse of~\cite{KMOW} retains a logarithmic degree
loss and an $\eps$ dependence (\cref{sec:ref-converse}); a
degree-$\Theta_k(\ell)$ sum-of-squares lower bound at
$m\asymp\eps^{-2}\Dkl$ would make the trade-off sharp on both sides,
and \cref{thm:ld-refutation} suggests that the obstruction is the
positivity of the pseudoexpectation, not the sampling model.

\paragraph{Vanishing bias at odd arity.}
The odd construction is uniform in the wider range of
\cref{rem:quantitative-ranges}, but extending the same $\rho^{-2}$
curve into the dense rooted-cell regime $\rho\ll(\ell/n)^{1/4}$
appears to require an odd operator that controls dense root cells
without discarding a constant fraction of the pair polynomial; the
non-spectral algorithm of~\cite{Mao2026} already achieves
unconditional bias dependence in its level range.

\paragraph{Faster non-spectral algorithms.}
For large constant $k$, the algorithm of~\cite{GHOS}
improves the running time quadratically in part of the parameter range; we leave open
whether it inherits the log-free sample curve or can be combined with
the normalized operators here.

\paragraph{Constants.}
We have not optimized the constants $C_{k,\delta}$ and $c_k$, and we
claim matching constants in no task. Determining the true constant, in the spirit of a
Kesten--Stigum threshold, is an interesting open question.

\bigskip

%\paragraph{Acknowledgments.}  To be added.

\subsection*{On the use of AI models.}
The central proof technique of this paper (sharp spectral bounds for
the normalized Kikuchi matrix via an accurate count of trace
walks) and the corresponding applications to detection, recovery, and refutation were conceived and developed by the authors. We have subsequently used a Large Language Model (GPT 5.5) to assist with transferring this technique to the matching lower bounds in the same model and the extension to more general predicates and planting laws. We have
 also used Large Language Models to assist with writing and proofreading throughout. We have verified every mathematical
statement.


\begin{thebibliography}{99}\small
\bibitem{SHmoore} A.~Schmidhuber and M.~B.~Hastings, \emph{A spectral proof of the hypergraph Moore bound}, arXiv:2607.26028 (2026).
\bibitem{ZK} L.~Zdeborov\'a and F.~Krzakala, \emph{Statistical physics of inference: thresholds and algorithms}, Adv.\ Phys.\ 65 (2016), 453--552.
\bibitem{BPW} A.~S.~Bandeira, A.~Perry, and A.~S.~Wein, \emph{Notes on computational-to-statistical gaps: predictions using statistical physics}, Port.\ Math.\ 75 (2018), 159--186.
\bibitem{KWB} D.~Kunisky, A.~S.~Wein, and A.~S.~Bandeira, \emph{Notes on computational hardness of hypothesis testing: predictions using the low-degree likelihood ratio}, in \emph{Mathematical Analysis, its Applications and Computation} (ISAAC 2019), Springer Proc.\ Math.\ Stat.\ 385, Springer, 2022, 1--50.
\bibitem{RichardMontanari} E.~Richard and A.~Montanari, \emph{A statistical model for tensor PCA}, NeurIPS 2014, 2897--2905.
\bibitem{WEM} A.~S.~Wein, A.~El~Alaoui, and C.~Moore, \emph{The Kikuchi hierarchy and tensor PCA}, J.\ ACM 72(5) (2025), Article 35; preliminary version in FOCS 2019.
\bibitem{Kikuchi1951} R.~Kikuchi, \emph{A theory of cooperative phenomena}, Phys.\ Rev.\ 81 (1951), 988--1003.
\bibitem{Hastings2019} M.~B.~Hastings, \emph{Classical and quantum algorithms for tensor principal component analysis}, Quantum 4 (2020), 237.
\bibitem{GKM} V.~Guruswami, P.~K.~Kothari, and P.~Manohar, \emph{Algorithms and certificates for Boolean CSP refutation: smoothed is no harder than random}, STOC 2022, 678--689.
\bibitem{HKM} J.-T.~Hsieh, P.~K.~Kothari, and S.~Mohanty, \emph{A simple and sharper proof of the hypergraph Moore bound}, SODA 2023, 2324--2344.
\bibitem{HKMMS} J.-T.~Hsieh, P.~K.~Kothari, S.~Mohanty, D.~Munh\'a Correia, and B.~Sudakov, \emph{Small even covers, locally decodable codes and restricted subgraphs of edge-colored Kikuchi graphs}, Int.\ Math.\ Res.\ Not.\ 2025, no.~5, article rnaf045.
\bibitem{SOKB} A.~Schmidhuber, R.~O'Donnell, R.~Kothari, and R.~Babbush, \emph{Quartic quantum speedups for planted inference}, Phys.\ Rev.\ X 15 (2025), 021077; arXiv:2406.19378v2.
\bibitem{SZ} A.~Schmidhuber and A.~Zlokapa, \emph{Quartic quantum speedups for community detection}, arXiv:2510.08494.
\bibitem{AGKM} O.~Alrabiah, V.~Guruswami, P.~K.~Kothari, and P.~Manohar, \emph{A near-cubic lower bound for 3-query locally decodable codes from semirandom CSP refutation}, STOC 2023, 1438--1448.
\bibitem{KothariManohar} P.~K.~Kothari and P.~Manohar, \emph{An exponential lower bound for linear 3-query locally correctable codes}, STOC 2024, 776--787.
\bibitem{HopkinsSteurer} S.~B.~Hopkins and D.~Steurer, \emph{Efficient Bayesian estimation from few samples: community detection and related problems}, FOCS 2017, 379--390.
\bibitem{Feige} U.~Feige, \emph{Small linear dependencies for binary vectors of low weight}, in \emph{Building Bridges: Between Mathematics and Computer Science}, Springer, 2008, 283--307.
\bibitem{RRS} P.~Raghavendra, S.~Rao, and T.~Schramm, \emph{Strongly refuting random CSPs below the spectral threshold}, STOC 2017, 121--131.
\bibitem{dOrsiTrevisan} T.~d'Orsi and L.~Trevisan, \emph{A Ihara--Bass formula for non-Boolean matrices and strong refutations of random CSPs}, CCC 2023, LIPIcs 264, Article 27, 27:1--27:16.
\bibitem{ChanDOrsiXu} S.~O.~Chan, T.~d'Orsi, and J.~Xu, \emph{Strongly refuting random CSP without literals}, arXiv:2604.27336 (2026).
\bibitem{KX} P.~K.~Kothari and J.~Xu, \emph{Smooth trade-off for tensor PCA via sharp bounds for Kikuchi matrices}, in \emph{Proceedings of the 2026 Annual ACM--SIAM Symposium on Discrete Algorithms (SODA)}, SIAM, 2026, 2617--2632, doi:10.1137/1.9781611978971.95.
\bibitem{ABBS} E.~Abbe, A.~S.~Bandeira, A.~Bracher, and A.~Singer, \emph{Decoding binary node labels from censored edge measurements: phase transition and efficient recovery}, IEEE Trans.\ Netw.\ Sci.\ Eng.\ 1 (2014), 10--22.
\bibitem{HWX} B.~Hajek, Y.~Wu, and J.~Xu, \emph{Achieving exact cluster recovery threshold via semidefinite programming: extensions}, IEEE Trans.\ Inform.\ Theory 62 (2016), 5918--5937; arXiv:1502.07738.
\bibitem{ALS} K.~Ahn, K.~Lee, and C.~Suh, \emph{Community recovery in hypergraphs}, IEEE Trans.\ Inform.\ Theory 65 (2019), 6561--6579; arXiv:1709.03670.
\bibitem{BKNPRW} A.~S.~Bandeira, D.~Kunisky, P.~Nizi\'c-Nikolac, L.~Pesenti, and R.~Wang, \emph{The hypergraph Moore bound}, arXiv:2607.14068v2 (2026).
\bibitem{NaorVerstraete} A.~Naor and J.~Verstra\"ete, \emph{Parity check matrices and product representations of squares}, Combinatorica 28 (2008), 163--185.
\bibitem{Tropp} J.~A.~Tropp, \emph{User-friendly tail bounds for sums of random matrices}, Found.\ Comput.\ Math.\ 12 (2012), 389--434.
\bibitem{Mao2026} S.~Mao, \emph{Near optimal algorithms for noisy $k$-XOR under low-degree heuristic}, arXiv:2604.10457 (2026).
\bibitem{Feige02} U.~Feige, \emph{Relations between average case complexity and approximation complexity}, STOC 2002, 534--543.
\bibitem{CGL} A.~Coja-Oghlan, A.~Goerdt, and A.~Lanka, \emph{Strong refutation heuristics for random $k$-SAT}, Combin.\ Probab.\ Comput.\ 16 (2007), 5--28.
\bibitem{AOW} S.~R.~Allen, R.~O'Donnell, and D.~Witmer, \emph{How to refute a random CSP}, FOCS 2015, 689--708.
\bibitem{BarakMoitra} B.~Barak and A.~Moitra, \emph{Noisy tensor completion via the sum-of-squares hierarchy}, COLT 2016, 417--445.
\bibitem{Grigoriev} D.~Grigoriev, \emph{Linear lower bound on degrees of Positivstellensatz calculus proofs for the parity}, Theoret.\ Comput.\ Sci.\ 259 (2001), 613--622.
\bibitem{Schoenebeck} G.~Schoenebeck, \emph{Linear level Lasserre lower bounds for certain $k$-CSPs}, FOCS 2008, 593--602.
\bibitem{AlekhnovichRazborov} M.~Alekhnovich and A.~A.~Razborov, \emph{Lower bounds for polynomial calculus: non-binomial case}, FOCS 2001, 190--199; expanded version, Proc.\ Steklov Inst.\ Math.\ 242 (2003), 18--35.
\bibitem{ODonnellWitmer} R.~O'Donnell and D.~Witmer, \emph{Goldreich's PRG: evidence for near-optimal polynomial stretch}, CCC 2014, 1--12.
\bibitem{MoriWitmer} R.~Mori and D.~Witmer, \emph{Lower bounds for CSP refutation by SDP hierarchies}, APPROX/RANDOM 2016, LIPIcs 60, 41:1--41:30; arXiv:1610.03029.
\bibitem{KMOW} P.~K.~Kothari, R.~Mori, R.~O'Donnell, and D.~Witmer, \emph{Sum of squares lower bounds for refuting any CSP}, STOC 2017, 132--145.
\bibitem{BHLM} A.~Basu, J.-T.~Hsieh, A.~D.~Lin, and P.~Manohar, \emph{Solving random planted CSPs below the $n^{k/2}$ threshold}, ICALP 2026; arXiv:2507.10833 (2025).
\bibitem{KothariECCC} P.~K.~Kothari, \emph{Kikuchi graphs of random hypergraphs are approximately Johnson}, ECCC Report TR26-099; arXiv:2606.08597 (2026).
\bibitem{SchrammWein} T.~Schramm and A.~S.~Wein, \emph{Computational barriers to estimation from low-degree polynomials}, Ann.\ Statist.\ 50 (2022), 1833--1858.
\bibitem{Ale03} M.~Alekhnovich, \emph{More on average case vs approximation complexity}, FOCS 2003, 298--307.
\bibitem{BHK+16} B.~Barak, S.~B.~Hopkins, J.~Kelner, P.~K.~Kothari, A.~Moitra, and A.~Potechin, \emph{A nearly tight sum-of-squares lower bound for the planted clique problem}, FOCS 2016, 428--437.
\bibitem{HKP+17} S.~B.~Hopkins, P.~K.~Kothari, A.~Potechin, P.~Raghavendra, T.~Schramm, and D.~Steurer, \emph{The power of sum-of-squares for detecting hidden structures}, FOCS 2017, 720--731.
\bibitem{Hop18} S.~B.~Hopkins, \emph{Statistical Inference and the Sum of Squares Method}, Ph.D.\ thesis, Cornell University, 2018.
\bibitem{BKW} A.~Blum, A.~Kalai, and H.~Wasserman, \emph{Noise-tolerant learning, the parity problem, and the statistical query model}, J.\ ACM 50 (2003), 506--519.
\bibitem{FPV} V.~Feldman, W.~Perkins, and S.~Vempala, \emph{On the complexity of random satisfiability problems with planted solutions}, STOC 2015, 77--86.
\bibitem{XJL} M.~Xu, V.~Jog, and P.-L.~Loh, \emph{Optimal rates for community estimation in the weighted stochastic block model}, Ann.\ Statist.\ 48 (2020), 183--204; arXiv:1706.01175.
\bibitem{MuscoMusco} C.~Musco and C.~Musco, \emph{Randomized block Krylov methods for stronger and faster approximate singular value decomposition}, NIPS 2015, 1396--1404.
\bibitem{KW} J.~Kuczy\'nski and H.~Wo\'zniakowski, \emph{Estimating the largest eigenvalue by the power and Lanczos algorithms with a random start}, SIAM J.\ Matrix Anal.\ Appl.\ 13 (1992), 1094--1122.
\bibitem{BunchKaufman} J.~R.~Bunch and L.~Kaufman, \emph{Some stable methods for calculating inertia and solving symmetric linear systems}, Math.\ Comp.\ 31 (1977), 163--179.
\bibitem{Bareiss} E.~H.~Bareiss, \emph{Sylvester's identity and multistep integer-preserving Gaussian elimination}, Math.\ Comp.\ 22 (1968), 565--578.
\bibitem{HolmgrenWein} J.~Holmgren and A.~S.~Wein, \emph{Counterexamples to the low-degree conjecture}, ITCS 2021, LIPIcs 185, article 75.
\bibitem{BHJK} R.-D.~Buhai, J.-T.~Hsieh, A.~Jain, and P.~K.~Kothari, \emph{The quasi-polynomial low-degree conjecture is false}, FOCS 2025; arXiv:2505.17360.
\bibitem{MaoLD} S.~Mao, \emph{The polynomial-time low-degree conjecture is false}, arXiv:2607.20318 (2026).
\bibitem{JiaVijay} H.~Jia and A.~Vijayaraghavan, \emph{Low-degree method fails to predict robust subspace recovery}, COLT 2026, PMLR 336, 3751--3781.
\bibitem{HKKLMT} J.-T.~Hsieh, D.~M.~Kane, P.~K.~Kothari, J.~Li, S.~Mohanty, and S.~Tiegel, \emph{Rigorous implications of the low-degree heuristic}, STOC 2026; arXiv:2601.05850.
\bibitem{GHOS} M.~Gupta, W.~He, R.~O'Donnell, and N.~G.~Singer, \emph{A classical quadratic speedup for planted $k$XOR}, arXiv:2508.09422 (2025).
\end{thebibliography}
\end{document}